\documentclass[%
 reprint,
 superscriptaddress,
 nofootinbib,
 amsmath,amssymb,
 aps,
 pra,
 floatfix,
]{revtex4-2}
\usepackage{graphicx}
\usepackage{amsmath,amsfonts,amssymb}
\usepackage{mathtools}
\usepackage[normalem]{ulem}
\usepackage{hyperref}
\usepackage{cleveref}
\usepackage{physics}
\usepackage{tikz}
\usetikzlibrary{arrows.meta,positioning,shapes.geometric,shapes.multipart,calc,decorations.pathreplacing}
\usepackage{bm}
\usepackage{amsthm}
\usepackage{thmtools}
\usepackage{thm-restate}
\usepackage[table,dvipsnames]{xcolor}

\usepackage[enableskew]{youngtab}
\usepackage{ytableau}
\usepackage{bbm}
\usepackage[T1]{fontenc}
\usepackage[utf8]{inputenc}

\usepackage[normalem]{ulem}

\newcounter{subfigure}

\ytableausetup{smalltableaux}

\newtheorem{theorem}{Theorem}%
\newtheorem{proposition}{Proposition}%
\newtheorem{lemma}{Lemma}%
\newtheorem{definition}{Definition}%
\newtheorem{remark}{Remark}%
\newtheorem{conjecture}{Conjecture}%
\newtheorem{problem}{Problem}%

\theoremstyle{plain}
\newtheorem{res}{Result}
\usepackage[framemethod=TikZ]{mdframed}
\definecolor{thmgray}{gray}{0.95}
\definecolor{conclusionbackground}{gray}{0.95}
\newcommand{\thmspaceafter}{2mm}
\newenvironment{result}%
{\begin{mdframed}[backgroundcolor=thmgray,nobreak,roundcorner=5pt,linewidth=1pt]\begin{res}}%
		{\end{res} \vspace{\thmspaceafter} \end{mdframed}}

\DeclareMathOperator*{\per}{per}

\newcommand{\Sym}{\mathrm{Sym}}
\DeclareMathOperator{\sgn}{sgn}

\newcommand{\sym}{\mathrm{sym}}
\newcommand{\U}{\mathrm{U}}

\renewcommand{\d}{m}

\newcommand{\poly}{\mathrm{poly}}

\newcommand{\const}{8\frac{m}{n}+9}

\definecolor{darkgreen}{RGB}{0,100,0}

\newcommand{\wigner}{HUN-REN Wigner Research Centre for Physics, Konkoly–Thege Miklós út 29-33, H-1525 Budapest, Hungary}
\newcommand{\elte}{E\"otv\"os  Lor\'and  University, Pázmány Péter sétány 1/a, H-1117 Budapest, Hungary}
\newcommand{\helsinki}{University of Helsinki, Yliopistonkatu 4 00100 Helsinki, Finland}
\newcommand{\algo}{Algorithmiq Ltd, Kanavakatu 3C 00160 Helsinki, Finland}
\newcommand{\pas}{Center for Quantum Enabled-Computing, Center for Theoretical Physics of the Polish Academy of Sciences, Al. Lotników 32/46, 02-668 Warsaw, Poland}

\begin{document}
\title{
    General framework for anticoncentration and linear cross-entropy\\ benchmarking in photonic quantum advantage experiments
}

\author{Zoltán Kolarovszki}
\email{kolarovszki.zoltan@wigner.hun-ren.hu}
\affiliation{\wigner}
\affiliation{\elte}
\author{Ágoston Kaposi}
\affiliation{\wigner}
\affiliation{\elte}
\author{Zoltán Zimborás}
\affiliation{\wigner}
\affiliation{\helsinki}
\affiliation{\algo}
\author{Michał Oszmaniec}
\affiliation{\pas}

\begin{abstract}
    Photonic architectures are one of the leading platforms for demonstrating quantum computational advantage, with Boson Sampling  and Gaussian Boson Sampling as the primary schemes. Yet, we lack for these photonic primitives a systematic theoretical understanding of linear cross-entropy benchmarking (LXEB), which is a central tool for testing quantum advantage proposals. In this work, we develop a representation-theoretic framework for the classical computation of average LXEB scores and second moments of output probability distributions, covering a range of quantum advantage experiments based on scattering $n$-photon states through $m$-mode Haar-random interferometers. Our methods apply in any regime, including the saturated regime, where the (expected) number of photons is comparable to the number of optical modes. The same second-moment techniques also allow us to prove anticoncentration for traditional Fock-state Boson Sampling in the saturated regime. Interestingly, for Gaussian Boson Sampling second moments are not sufficient to establish a meaningful anticoncentration statement. The technical core of our approach rests on decomposing two copies of the $n$-particle bosonic space $\mathrm{Sym}^n(\mathbb{C}^m)$ into irreducible representations of $\mathrm{U}(m)$. This reduces two-copy Haar averages to computing purities of initial states after partial traces over particles, highlighting the role that particle entanglement plays for LXEB and anticoncentration.
\end{abstract}

\maketitle

\section{Introduction}

    \begin{figure*}[t]
        \centering
        \begin{tikzpicture}[
    x=1cm,y=1cm,
    font=\small,
    >=Latex,
    line/.style={thick},
    statebox/.style={
        draw, rounded corners=2pt,
        minimum width=1.25cm, minimum height=0.55cm,
        fill=blue!4
    },
    interferometer/.style={
        draw, rounded corners=3pt,
        minimum width=2.9cm, minimum height=4.7cm,
        fill=orange!8
    },
    detector/.style={
        draw, semicircle,
        minimum width=0.8cm, minimum height=0.52cm,
        rotate=270,
        fill=green!8
    },
    box/.style={
        draw, rounded corners=3pt,
        minimum width=3.9cm, minimum height=0.9cm,
        fill=gray!8
    },
    every node/.style={align=center}
]

\def\ytop{1.8}
\def\ymid{0}
\def\ybot{-1.8}

\node[statebox] (s1) at (0,\ytop) {$\sigma^{(1)}$};
\node[statebox] (s2) at (0,\ymid) {$\sigma^{(a)}$};
\node[statebox] (s3) at (0,\ybot) {$\sigma^{(m)}$};

\node at (0,0.9) {$\vdots$};
\node at (0,-0.9) {$\vdots$};

\node[interferometer] (U) at (2.6,0) {};
\node at (2.6,0.45) {\large Haar-random};
\node at (2.6,-0.05) {\large interferometer};
\node at (2.6,-0.8) {$U\sim \mathrm{Haar}(m)$};

\draw[line] (s1.east) -- (U.west |- s1.east);
\draw[line] (s2.east) -- (U.west |- s2.east);
\draw[line] (s3.east) -- (U.west |- s3.east);

\draw[line] (U.east |- 0,\ytop) -- ++(1.2,0);
\node[detector] (d1) at ($(U.east |- 0,\ytop)+(1.0,0)$) {};
\draw[line] ($(U.east |- 0,\ytop)+(1.0,0)$) -- (d1);
\node[above right=0.0cm and -0.51cm of d1] {$n_1$};

\draw[line] (U.east |- 0,\ymid) -- ++(1.2,0);
\node[detector] (d2) at ($(U.east |- 0,\ymid)+(1.0,0)$) {};
\draw[line] ($(U.east |- 0,\ymid)+(1.0,0)$) -- (d2);
\node[above right=0.0cm and -0.52cm of d2] {$n_a$};

\draw[line] (U.east |- 0,\ybot) -- ++(1.2,0);
\node[detector] (d3) at ($(U.east |- 0,\ybot)+(1.0,0)$) {};
\draw[line] ($(U.east |- 0,\ybot)+(1.0,0)$) -- (d3);
\node[above right=0.0cm and -0.53cm of d3] {$n_m$};

\node at (5.0,1.0) {$\vdots$};
\node at (5.0,-0.8) {$\vdots$};

\node at (0,2.65) {product input};

\node at (0,-2.95) {$\rho=\bigotimes_{a=1}^m \sigma^{(a)}$};

\node at (5.2,2.90) {particle-number \\ resolving detectors};
\node at (5.2,-3.00) {$\bm n={(n_1,\ldots,n_m)} \in \mathcal{S}_{m,n}$};

\node at (6.5,0.8) {output\\probability};

\node at (6.4,-0.1) {\large $\sim p_{\rho, U}$};

\node[box, text width=7.5cm, minimum height=1.15cm] (eval) at (11.8,1.9)
{ \vspace*{-3mm}
\begin{flushleft}
\hspace*{2mm} \textbf{Result 1.}
    For product inputs $\rho=\bigotimes_{a=1}^m \sigma^{(a)}$
    \\[1pt]
 \hspace*{2mm}   the LXEB reference value
    \\[-0.5pt]
    \begin{center}
    $\mathrm{LXE}^{(n)}_{\mathrm{ref}}(\rho) = \underset{U\sim \mathrm{Haar}(m)}{\mathbb E}\underset{\bm n \in \mathcal S_{m,n}}{\sum}  p_{\rho,U}(\bm n)^2$\\[-5mm]
    \end{center}
   \vspace{-2mm} 
  \hspace*{2mm}  can be classically computed in time $\operatorname{poly}(m,n)$.\\
    \end{flushleft}
};

\node[box, text width=7.5cm, minimum height=1.20cm] (eval) at (11.8,-1.6)
{\vspace*{-3mm}
\begin{flushleft}
\hspace*{2mm} \textbf{Result 2.}
   The probability distribution $p_U$ of the\\ 
\hspace*{2mm} standard Boson Sampling  in the saturated regime\\
\hspace*{2mm} \(m=\Theta(n)\) satisfies average-case anticoncentration.\\
\hspace*{2mm} There exists \(C>0\) such
 that for every \(\tau\in[0,1]\)
    \vspace{-1mm}
    \[
        \underset{\substack{U\sim\mathrm{Haar}(m)\\ \bm n\sim \mathcal U_{m,n}}}{\Pr}
        \left(
            p_U(\bm n)\ge \frac{\tau}{|\mathcal S_{m,n}|}
        \right)
        \ge
        \frac{(1-\tau)^2}{C}.
    \]\\
 \end{flushleft}
 \vspace*{-3mm}
 \hspace{2mm}
};

\end{tikzpicture}
        \caption{
            Setting and key results of the paper. Product input states are evolved via  a Haar-random \(m\)-mode interferometer and measured with particle-number-resolving detectors, resulting in an outcome $\bm{n}$. The set of allowed outcomes is denoted by $\mathcal{S}_{m,n}$ and the uniform distribution over by $\mathcal{U}_{m,n}$. We develop tools to classically compute the reference value $\mathrm{LXE}^{(n)}_{\mathrm{ref}}(\rho) = \mathbb{E}_{U\sim \mathrm{Haar}(m)}\sum_{\bm n \in \mathcal S_{m,n}}  p_{\rho,U}(\bm n)^2$ for $n$-particle sector. The same tools also establish anticoncentration for Boson Sampling in the saturated regime \(m=\Theta(n)\). 
        }
        \label{fig:general_setup_lxeb}
    \end{figure*}
    
    Photonic platforms have emerged as leading candidates for demonstrating quantum computational advantage, owing to their low decoherence, high-fidelity linear optical components, and scalability to high number of modes~\cite{O_Brien_2009,Flamini_2017,Bao2023}. This progress has enabled increasingly sophisticated photonic experiments, most notably Boson Sampling~\cite{aaronson:2010} and Gaussian Boson Sampling~\cite{hamilton2017,Kruse_2019}, which are widely regarded as one of the strongest candidates for attaining quantum computational advantage with near-term hardware. 
    The scale of conventional Boson Sampling experiments has increased steadily over the years~\cite{wang2017,brod2019,wang2019}, reaching up to 20 particles in 60 modes~\cite{wang2019}. A non-photonic realization was demonstrated in Ref.~\cite{Young_2024}, in which analog Boson Sampling was performed with approximately 180 atoms.
    At the same time, Gaussian Boson Sampling experiments have rapidly advanced to regimes with hundreds~\cite{Zhong_2020, borealis} or even thousands~\cite{liu2025robustquantumcomputationaladvantage} of optical modes and photon-number detection events.

    In practice, however, photonic quantum devices inevitably suffer from imperfections such as particle losses~\cite{saleh2016,Oszmaniec_2018,Garc_a_Patr_n_2019,LossGBS2020}, partial distinguishability~\cite{Shchesnovich_2015,Renema_2018}, detector imperfections~\cite{resta2023}, inaccuracies in circuit implementations~\cite{arkhipov2015}.
    Experimental imperfections therefore lead to deviations from the ideal theoretical model and, in sufficiently noisy regimes, can even move the device into a parameter range that admits efficient classical simulation, thereby putting constraints on claims of quantum advantage~\cite{changhunclass2024, {dodd2025fastfrugalgaussianboson}}. Consequently, the development of reliable certification and benchmarking methods has become central in the study of photonic quantum advantage proposals~\cite{Spagnolo_2014,Eisert_2020,Hangleiter_2023}. However, certifying from experimental data that the samples produced by a quantum device are close in total variation (TV) distance to the target output distribution requires an exponential number of repetitions~\cite{Hangleiter_2019,Hangleiter_2023}. This is particularly significant because closeness in TV distance is the notion of approximation that underlies current hardness results in quantum advantage proposals based on sampling~\cite{aaronson:2010, movassagh2020quantumsupremacyrandomcircuits,Bouland2019ComplexityVerificationRCS,bouland2025complexitytheoreticfoundationsbosonsamplinglinear,Oszmaniec_2022}. 
    
    Since certification in TV distance is infeasible for larger experiments, a variety of benchmarking and validation protocols have been developed in recent years for Boson Sampling and other quantum advantage schemes.
    One of the most important methods for validating quantum advantage experiments is cross-entropy benchmarking (XEB) and its linear variant (LXEB) ~\cite{boixo2018characterizing,Arute_2019,AaronsonGunn2020XEB}. These tests quantify the correlation between experimental data and ideal distribution generated by a quantum circuit of interest, and so far have been developed and realized mostly for random quantum circuits (RQC) implemented on superconducting qubits~\cite{Arute_2019, Heifei2021RQC, PhaseTransitionsRCSRQC}. 
    Importantly, both XEB and LXEB require classical computation of probabilities associated with the noiseless implementation of the quantum circuit generating the samples. The LXEB test for a given random quantum circuit amounts to (i) collecting a number of samples  from a device realizing the circuit, and (ii) computing the empirical mean of the probabilities of the samples. A device ``passes'' the LXEB test if the empirical mean is close to the reference LXEB score, i.e., the expectation value corresponding to sampling from an ideal circuit.
    Although LXEB can be efficiently spoofed~\cite{BarakChouGao2021SpoofingXEB,Aharonov23,Gao_2024}, it remains a widely used tool for benchmarking and validating quantum advantage proposals based on sampling from random quantum circuits~\cite{Hangleiter_2023}.

    Despite significant experimental progress in photonic quantum advantage proposals, the behavior of LXEB in these architectures remains poorly understood. Nevertheless, cross-entropy-based measures have already appeared in recent photonic experiments, including comparisons against adversarial classical samplers~\cite{borealis}. For Boson Sampling and Gaussian Boson Sampling, existing analytical results for average LXEB scores are asymptotic and largely restricted to the dilute regime, i.e., where the number of modes $m$ scales at least quadratically with the number of photons~\cite{aaronson:2010,martinezcifuentes2024linearcrossentropy,Ehrenberg_2025}. 
    Beyond these special cases, characterized by no collisions of photons at the output and Gaussian approximations of interferometers describing the scattering in linear optical network, no general analytical or efficiently computable expression for the reference LXEB score is currently known in photonic settings, and the available formulas are mostly restricted to the dilute regime; see \Cref{tab:bs_regimes}. This significantly restricts the practical applicability of LXEB as a benchmark for photonic quantum advantage experiments, particularly in the saturated regime, where the (expected) number of photons is comparable to the number of modes. As virtually all experiments to date operate in the saturated regime, this is both experimentally relevant and well-motivated by recent complexity-theoretic hardness results~\cite{bouland2025complexitytheoreticfoundationsbosonsamplinglinear}.
    
    \subsection*{Main results and their significance}

    In this work, we address the problem of evaluating the average LXEB score for a class of quantum advantage proposals in which an initial product photonic state \(\rho\) on \(m\) optical modes is transformed by a Haar-random passive unitary transformation \(U\in\mathrm{U}(m)\) and subsequently measured using particle-number detectors. This results in a string of occupation numbers for each mode $\bm{n}=(n_1,\ldots,n_m)$ distributed according to a  probability distribution $p_{\rho,U}$.
    Our methods apply to a broad class of photonic sampling proposals including Boson Sampling~\cite{aaronson:2010}, Gaussian Boson Sampling~\cite{hamilton2017}, and several related variants~\cite{Bentivegna_2015,Li2025,Renema_2020}. Our approach is valid for general relations between the number of photons $n$ and the number of modes $m$, as summarized in \Cref{tab:bs_regimes}. Importantly, our methods cover the saturated regime and accommodate optical loss, nonuniform squeezing, and displacements, thereby extending LXEB analysis for photonic quantum advantage proposals far beyond the current state of the art.

    \begin{table}[h]
        \centering
        \caption{Reference LXEB values available in the literature. }
        \label{tab:bs_regimes}
        \renewcommand{\arraystretch}{1.4}
        \begin{tabular}{lccc}
        \hline
        \textbf{Experiment} 
        & \textbf{$m = \Omega(n^2)$} 
        & \textbf{Arbitrary $n,m$} 
        & \textbf{Losses} \\
        \hline
        Boson Sampling (BS)  & \cite{aaronson:2010} & \color{ForestGreen} This work  & \color{ForestGreen} This work  \\
        Gaussian BS (GBS)      & \cite{martinezcifuentes2024linearcrossentropy} & \color{ForestGreen} This work & \cite{martinezcifuentes2024linearcrossentropy} \\
        Displaced GBS~\cite{Li2025} & \color{ForestGreen} This work & \color{ForestGreen} This work & \color{ForestGreen} This work \\
        Superposition BS~\cite{Renema_2020} & \color{ForestGreen} This work  & \color{ForestGreen} This work  & \color{ForestGreen} This work  \\
        Scattershot BS~\cite{Bentivegna_2015} & \cite{aaronson:2010} & \color{ForestGreen} This work  & \color{ForestGreen} This work  \\
              \hline
        \end{tabular}
    \end{table}

    On the technical side, our approach is based on decomposing two-copy Hilbert spaces at a fixed particle number $n$, denoted by $\mathrm{Sym}^n(\mathbb{C}^m)^{\otimes 2}$, onto irreducible components of the $m$-mode unitary group $\mathrm{U}(m)$. The expected LXEB score can be written as a sum (over possible outcomes $\bm{n})$ of second moments of the output distribution, $\mathbb{E}_{U\sim \mathrm{Haar}(m)}[p_{\rho,U}(\bm{n})^2]$. These objects are quadratic in the probabilities and therefore naturally admit a two-copy formulation, which brings the tensor-square representation of the interferometer on the symmetric 
    $n$-particle sector into play, allowing the relevant space to be decomposed into irreducible components. The Haar average then reduces to a finite sum of elementary terms, which can be rewritten in terms of purities of partial traces (over particles, not modes) of the initial quantum state $\rho$. This group-theoretic and first-quantization approach allows us to circumvent previous limitations in the computation of expected LXEB scores \cite{martinezcifuentes2024linearcrossentropy,Ehrenberg_2025,Ehrenberg_2025_second_moment,aaronson:2010} that resulted from (i) the use of specific forms of probability distributions involved in quantum advantage experiments (Permanents, Hafnians, etc.),  (ii) utilizing Gaussian approximations of submatrices of Haar random unitaries $U$, and (iii) the necessity of the hiding property, i.e., equivalence between different outputs $\bm{n},\bm{n}'$. Importantly, properties (ii) and (iii) hold only in the (experimentally less realistic) dilute regime $m = \Omega(n^2)$.

    Beyond benchmarking, we apply the second-moment quantities appearing in $\mathrm{LXEB}$ to study anticoncentration of the output distributions generated by these quantum advantage proposals. At an intuitive level, anticoncentration means that the output distribution is not excessively peaked: rather than placing almost all of its weight on a very small subset of outcomes, a non-negligible fraction of outputs should occur on the natural probability scale set by the inverse sample-space size~\cite{aaronson:2010}. Anticoncentration is highly relevant in complexity-theoretic arguments for the classical hardness of approximately sampling from output probability distributions in quantum advantage experiments. This is because Stockmeyer’s approximate counting algorithm~\cite{stockmeyer} produces, given an approximate classical sampler, additive estimates of output probabilities $p_{\rho,U}(\bm{n})$, and the anticoncentration property is precisely what allows such additive estimates to be upgraded to multiplicative ones on a constant fraction of instances. Anticoncentration has been established for several quantum advantage proposals, but in the photonic setting it has so far remained unproven, although substantial effort has been made in this direction \cite{nezami2021permanentrandommatricesrepresentation, PhaseTransitionsRCSRQC,Ehrenberg_2025_second_moment} since the explicit conjecture was stated in the original work of Aaronson and Arkhipov~\cite{aaronson:2010}.

    We show that anticoncentration for typical interferometers can be directly controlled by LXEB reference values. Crucially, for regular Boson Sampling (initialized with a collision-free Fock state) in the saturated regime $m=\Theta(n)$, our second moment analysis is strong enough to prove a meaningful anticoncentration result comparable to the earlier results obtained using second moment techniques for IQP~\cite{Bremner_2017}, Random Circuit Sampling~\cite{Hangleiter2018anticoncentration,dalzell_2022}, and Fermion Sampling~\cite{Oszmaniec_2022}. We emphasize that this is the first proof of anticoncentration for a photonic quantum advantage scheme. Interestingly, for other schemes, such as Gaussian Boson Sampling or Scattershot Boson Sampling, second moment techniques allow only to establish \emph{weak anticoncentration} (for arbitrary relations between $m$ and $n$), in accordance with earlier works \cite{Ehrenberg_2025_second_moment,PhaseTransitionsRCSRQC} that obtained analogous results for the dilute regime ($m = \Omega(n^2)$). We also note that our Boson Sampling result \emph{does not} prove the \emph{Permanent of Gaussians Anticoncentration Conjecture} from \cite{aaronson:2010}, which is a statement that bears significance only in the dilute regime, and requires analysis beyond second moments.

    As a byproduct of our findings and computational techniques, we establish a few additional results that may be of independent interest. First, following \cite{Hangleiter_2019}, we use formulas for expected LXEB scores to establish an exponential lower bound on the sample complexity of device-independent  certification for Boson Sampling, Gaussian Boson Sampling, and Scattershot Boson Sampling. These bounds hold for arbitrary $m,n$, going beyond the specific parameter regime studied in \cite{Hangleiter_2019}. 
    Second, we use our formulas for the second moments of output probabilities in Boson Sampling to establish a variant of Hunter-Jones conjecture about moments of Permanents of Haar random matrices (formulated in \cite{nezami2021permanentrandommatricesrepresentation}).   
    Finally, through the computation of expectation values of bosonic swap operators \(\mathbb{S}_q\), we develop an algorithmic method for computing and bounding the particle-entanglement Rényi-2 entropy of bosonic Fock states \(\ket{\bm{n}}\), also known as multimode Dicke states~\cite{Oszmaniec_2016}.
    In particular, we show that typical Fock states from the bosonic Hilbert space $\mathrm{Sym}^n(\mathbb{C}^m)$ are characterized by volume-law entanglement when $m=\Theta(n)$, complementing earlier results~\cite{Popkov_2005,Carrasco_2016} that established area-law scaling in the regime of fixed number of modes $m$.

\subsection*{Relation to prior work}

    \subsubsection*{Existing benchmarking and validation strategies for photonic quantum advantage}

    Since the hardness assumptions  of quantum advantage experiments is not directly verifiable at experimental scale, a suite of tailored and restricted methods to certify Boson Sampling-like experiments have been proposed in recent years. First, the Aaronson-Arkhipov row-norm estimator, used to rule out uniform-sampling hypotheses~\cite{aaronson2013bosonsamplingfarfromuniform,spagnolo2014experimental}; likelihood-ratio tests comparing indistinguishable-boson and distinguishable-particle models~\cite{Bentivegna_2015}; bosonic randomized benchmarking protocols~\cite{Arienzo_2025,Valahu_2024}; and Bayesian approaches~\cite{bentivegna2014,Mart_nez_Cifuentes_2023}.
    Additional certification strategies exploit device-level information, including interferometer tomography from one- and two-photon data \cite{laing2012superstabletomography} and correlation-based benchmarks capturing many-body interference, ranging from Boson Sampling signatures \cite{walschaers2016statistical} to two-point correlators for Gaussian Boson Sampling that remain informative under loss and noise \cite{phillips2019benchmarkinggbs}.
    These checks are often complemented by validation against physically motivated imperfection models and by identifying regimes where loss/noise render the experiment efficiently classically simulable \cite{rahimikeshari2016sufficientconditions}. More recently, data-driven validators based on unsupervised learning \cite{agresti2019patternrecognition}, coarse-grained ``binned'' photon-number tests \cite{seron2024binnedvalidation}, and efficient verification protocols leveraging additional single-mode Gaussian measurements \cite{chabaud2021efficientverification} have been developed to scale certification without the need full probability estimation.

\subsubsection*{Relevance of XEB and LXEB for\\ certification of quantum advantage}
    Under certain assumptions regarding experimental noise \cite{Bouland2019ComplexityVerificationRCS} XEB  bounds TV distance between ideal and experimental probability distributions. With regards to more commonly used LXEB the situation is more subtle---first, for random quantum circuits it was used \cite{Arute_2019,Heifei2021RQC} as the diagnostic of fidelity $f$ to of the the ideal probability distribution $p_{\text{ideal}}$ to its experimental realization $p_{\text{exp}}=f p_{\text{ideal}} +(1-f)p_{\text{unif}}$, with $p_{\text{unif}}$ being a uniform distribution on $N$ bits (this simple model of experimental imperfections noise can be proved rigorously in the regime of low unital noise \cite{Dalzell2024WhiteNoise}, but breaks down for larger errors \cite{PhaseTransitionsRCSRQC} or in the presence of non-unital noise  \cite{Fefferman2024NonunitalRCS}). Second, inspired by the original proposal of RQC by Google \cite{boixo2018characterizing}, Aaronson and Gunn proposed \cite{AaronsonGunn2020XEB} a complexity-theoretic conjecture named XQUATH (``Linear Cross-Entropy Quantum Threshold Assumption'') and used it to justify classical hardness of  spoofing LXEB for RCS. However, it was shown subsequently that linear cross entropy benchmarking can be classically spoofed efficiently in experimentally relevant regimes ~\cite{BarakChouGao2021SpoofingXEB,Aharonov23,Gao_2024}. Despite this, LXEB remains a key benchmark for assessing quality of quantum advantage experiments based on RQC, which is substantiated further by the fact that it can be regarded as a form of randomized benchmarking~\cite{Helsen2022GeneralFrameworkRB}.

\subsubsection*{Prior studies on anticoncentration and\\  LXEB in the photonic context}

   The key conceptual reason why LXE can be computed for RQC lays in the fact that they form, even for shallow depths approximate 2-designs \cite{BrandaoHarrowHorodecki2016PolynomialDesigns,SchusterHaferkampHuang2024ExtremelyLowDepth}, which allows to analytically compute (approximate) second moments of output probabilities of RQC  $\mathbb{E}_C [p_C(\bm{x})^2] \approx 2/4^{N}$. Furthermore, for RQCs all outcomes are equivalent (due to so-called \emph{hiding} property) and as a result the expected LXEB score can expressed solely in terms $\mathbb{E}_C [p_C(\bm{x})^2]$, for specific outcome $\bm{x}$. In contrast, analogous  properties do not hold and for random passive linear optical circuits underlying quantum advantage experiments \cite{aaronson:2010,hamilton2017}, which makes computation of expected LXEB challenging. More specifically, until recently analysis of hardness of sampling, as well as LXEB for photonic proposals was limited to the dilute regime in which number of optical modes greatly exceeds squared number of photons, $m =\Omega(n^2)$. This condition ensured that (i) with high probability collisions of photons in the output do not occur and all events are equivalent up to permutation of modes (hiding), (ii) output statistics are approximately given by suitably normalized complex-Gaussian matrices. 

   A systematics study of LXEB for Gaussian Boson Sampling was carried out in \cite{martinezcifuentes2024linearcrossentropy}. Therein, authors derived efficient numerical tools to estimate reference values for a suite of Gaussian Boson Sampling experiments. However, the results relied on asymptotic analysis on complicated group-theoretics objects (Weingarten functions)  valid only in the extreme diluted scenarios. In contrast, our methods allow to compute LXEB scores exactly for arbitrary relation between $m$ and $n$, for the same class of input states.
   
    In Ref.~\cite{Ehrenberg_2025}, the authors employ methods introduced earlier in ~\cite{Ehrenberg_2025_second_moment} to analyze second moments of the Gaussian Boson Sampling distribution. Within this work, the authors show that Gaussian Boson Sampling exhibits a transition in anticoncentration depending on how the number of initially squeezed modes scales relative to the number of detected photons. If the number of squeezed modes grows too slowly with the number of modes, anticoncentration fails. On the other hand, when it grows sufficiently rapidly, the output probabilities satisfy a weak form of anticoncentration. Our approach is able to recover this transition in the saturated regime.

    
    Finally, significant progress in understanding anticoncentration and  high moments of of Permanents of Gaussian matrices and submatrices of Haar random matrices  Ref.~\cite{nezami2021permanentrandommatricesrepresentation}. Operationally, these correspond only to collision-free probabilities in Boson Sampling experiments.  Our setting, however, differs in an essential way: we establish anticoncentration in the \emph{saturated} regime in which (i)  matrix entries contributing to the Permanent are highly correlated, (ii) we have many repeated rows (collisions).

    \vspace{2em}
    \paragraph*{Organization of the paper.}
        \Cref{sec:basics} introduces the notation and basic concepts used throughout the paper, including first- and second-quantized descriptions of bosonic systems, the Fock-space formalism, Boson Sampling and Gaussian Boson Sampling.
        In \Cref{sec:main_results}, we present the main results of this work, and in \Cref{sec:discussion} we discuss their implications together with related open problems.
        The general framework for LXEB in photonic quantum advantage experiments is developed in \Cref{sec:framework}, where we derive exact formulas for Haar-averaged LXEB reference values.
        In \Cref{sec:specific_experiments}, we apply this framework to concrete photonic quantum advantage proposals, specializing it to both Fock-state and Gaussian Boson Sampling and analyzing estimator variances as well as the effect of uniform losses.
        Finally, in \Cref{sec:anticoncentration}, we revisit the framework from a second perspective, using it to study anticoncentration, and provide a proof of average anticoncentration for saturated-regime Boson Sampling, which in turn provides support for a Stockmeyer-based hardness argument.
       The technical details and more involved proofs are given in the Appendix.

\section{Basics and notation}\label{sec:basics}

    We start by presenting basic notations and terminology necessary to describe systems of bosons that can occupy $m$ modes. Our exposition incorporates  both first  and second quantization. These perspectives are complementary---first quantization emphasizes the particle picture, which is useful for studying higher-order invariants in bosonic systems, while second quantization is familiar for quantum optics community and is often more useful to perform concrete computations.
    
    An individual particle that can be found in one of $m$ levels (modes) can be described by a single-particle Hilbert space $\mathbb{C}^{m}$, on which we have a natural orthonormal basis of modes $\{\ket{i} \}_{i=1}^m$.  Let $\pi \in S_n $, where \( S_n \) is the permutation group of \( n \) elements, and consider the unitary permutation of factors of $(\mathbb{C}^{m})^{\otimes n}$,  \( U_\pi : (\mathbb{C}^m)^{\otimes n} \to (\mathbb{C}^m)^{\otimes n} \) defined by
    $
        U_\pi \psi = U_\pi \, \phi_1 \otimes \dots \otimes \phi_n = \phi_{\pi^{-1}(1)} \otimes \dots \otimes \phi_{\pi^{-1}(n)}.
    $
    For bosons, the \(n\)-particle Hilbert space is the subspace of $(\mathbb{C}^m)^{\otimes n}$ defined by
    \begin{equation}
        \Sym^n(\mathbb{C}^m) \!  \coloneqq \!  \big \{ \psi \in (\mathbb{C}^m)^{\otimes n}: U_\pi \psi = \psi ,\;\forall \,\pi \in S_n \big \}.
    \end{equation}
    In other words, bosonic wavefunctions are the ones that are totally symmetric upon permutation of particles. The natural orthonormal basis of $\Sym^n(\mathbb{C}^m)$ is given by Fock states (equivalent to generalized Dicke states~\cite{Popkov_2005,Carrasco_2016,Nepomechie_2024}, which are parametrized by a tuple of occupation numbers $\bm{n}=(n_1,\ldots,n_m)$ (satisfying $\sum_{i=1}^m n_i=n$)  and are given by
    \begin{equation}
        \ket{\bm{n}}\!= \binom{n}{\bm{n}}^{1/2} \mathbb{P}_{\mathrm{sym}} \left(\ket{1}^{\otimes n_1} \otimes \ldots \otimes\ket{m}^{\otimes n_m}\right)\ ,
    \end{equation}
    where $\mathbb{P}_{\mathrm{sym}}\coloneqq \frac{1}{n!}\sum_{\pi\in S_n} U_\pi$ is the orthonormal projector onto $\Sym^n(\mathbb{C}^m)$ and $\binom{n}{\bm{n}}= \frac{n!}{n_1! \ldots n_m!}$. Throughout the work, we will use $\mathcal{D}(\Sym^n(\mathbb{C}^m))$ to refer to the set of  quantum states on $\Sym^n(\mathbb{C}^m)$. 

    When the number of bosons is not fixed, the appropriate Hilbert space is the \textit{bosonic Fock space}, which is defined as the direct sum of the \(n\)-particle Hilbert spaces:
    \(
        \mathcal{F}(\mathbb{C}^m)=\bigoplus_{n=0}^{\infty}\Sym^n(\mathbb{C}^m).
    \)
    By convention, \(\Sym^0(\mathbb{C}^m)\cong \mathbb{C}\ket{0}\), where \(\ket{0}\) denotes the vacuum state containing no particles. On $\mathcal{F}(\mathbb{C}^m)$ we define the creation and annihilation operators by their action on the Fock basis states as
    \begin{align} \begin{split} \label{eq:ladder_operators_action}
        a^\dagger_j \ket{ n_1, \dots, n_j, \dots, n_m} &\!=\! \sqrt{n_j + 1} \ket{n_1, \dots, n_j + 1, \dots, n_m},\\
        a_j \ket{n_1, \dots, n_j, \dots, n_m} &\!=\! \sqrt{n_j} \ket{n_1, \dots, n_j -1, \dots, n_m},
    \end{split} \end{align}
    which satisfy the canonical commutation relations $[a_i, a_j^\dagger] = \delta_{ij} \mathbbm{1}, [a_i, a_j] = [a_i^\dagger, a_j^\dagger] = 0$.

   A passive linear optical transformation is described by a unitary matrix $U \in \U(m)$, which represents how a interferometer acts on a single-particle Hilbert space $\mathbb{C}^m$. The interferometer acts on the space of $n$ bosons  $\Sym^n(\mathbb{C}^m)$ by 
    \begin{equation}\label{eq:interferometer_action}
        \varphi_n (U) \coloneqq U^{\otimes n} |_{\Sym^n(\mathbb{C}^m)}
    \end{equation}
    Importantly, $\varphi_n$ is an irreducible representation of $\mathrm{U}(m)$~\cite{fultonharris}. Additionally, an interferometer acts on the full Fock space as $\varphi(U) = \bigoplus_{n=0}^\infty \varphi_n(U)$. Note that $\varphi(U)$ preserves the total numer of particles.
    It is useful to recall that $\varphi(U)$ induces the following transformation on the annihilation and creation operators:
    \begin{equation}
        \varphi(U)\, a_j^{\dagger} \,\varphi(U^{\dagger}) = \sum_{k=1}^m U_{kj}\, a_k^\dagger.
    \end{equation}
    Throughout this work, a photonic quantum advantage experiment is specified by an input state $\rho$ on $\mathcal{F}(\mathbb{C}^m)$, followed by a passive linear optical transformation described by  $U \in \mathrm{U}(m)$. The resulting state is measured by the photon-number-resolving POVM
    $\{ \ketbra{\bm{n}} \}_{\bm{n} \in \mathbb{N}_0^m}$, 
    each outcome corresponds to an occupation pattern
    $\bm{n} = (n_1,\ldots,n_m)$ and $n_i$ denotes the number of detected photons
    in mode $i$ (note that $n_i$ is unconstrained).  
    The probability of observing a given occupation pattern $\bm{n}$ is then given by the Born rule  
    \begin{equation}\label{eq:particle_detection_probability} 
        p_{\rho, U}(\bm{n})
        \coloneqq
        \Tr\!\left[
            \varphi_n(U)\, \rho \,\varphi_n(U)^\dagger
            \ketbra{\bm{n}}
        \right].
    \end{equation}
    Sampling from $p_{\rho, U}$, for different initial states $\rho$ constitutes the computational task underlying photonic quantum advantage experiments  such as Boson Sampling ~\cite{aaronson:2010,bouland2025complexitytheoreticfoundationsbosonsamplinglinear}, Gaussian Boson Sampling \cite{hamilton2017}, and their variants \cite{Bentivegna_2015,Grier_2022,Li2025}. Throughout this work we will be often interested in the subset of outputs $\bm{n}$ that correspond to a fixed particle number $|\bm{n}|\coloneqq \sum_{i=1}^m n_i =n$:
    \begin{equation} \label{eq:Smn}
        \mathcal S_{m,n}
        \coloneqq
        \left\{
            \bm n\in\mathbb N_0^m:\ |\bm n|=n
        \right\}\ .
    \end{equation}
    For the regular Boson Sampling \cite{aaronson:2010} the input state is a collision-free Fock state 
    \begin{equation}\label{eq:fockin}
        \ket{\bm n_0}
        \coloneqq
        |\!\underbrace{1, \dots, 1}_{n}, \underbrace{0, \dots, 0}_{m-n}\rangle\ .
    \end{equation}
    The corresponding probability distribution has support precisely on $\mathcal{S}_{m,n}$ and individual probabilities are expressed in terms of Permanents of matrices constructed from the matrix $U$ describing the interferometer. 
    
    In the Gaussian Boson Sampling scheme~\cite{hamilton2017}, the input is instead a multimode Gaussian state; in the most common variant, this is a product of single-mode squeezed vacuum states characterized by a squeezing parameter $r\in \mathbb{R}_+)$ and  defined as
    \begin{equation}\label{eq:squeezed_state}
        \ket{r} \coloneqq \exp\left(\frac{r}{2} [ (a^\dagger)^2 - a^2] \right) \ket{0} .
    \end{equation}
    In this case output probability has support on $\mathcal{S}_{m,n}$ corresponding to different even occupation numbers $n$, and the probabilities themselves are expressed in terms of Hafnians constructed from submatrices of $U$.

    In our work, we will make extensive use of Haar measure on the on \(\mathrm{U}(m)\), denote by $\mu_m$. Recall that $\mathrm{Haar}(m)$ is a unique probability measure on Lie group  \(\mathrm{U}(m)\) that  both left and right group multiplication, known as the \emph{Haar measure}. Accordingly, the Haar measure provides the natural notion of sampling an \(m\times m\) unitary matrix uniformly at random which is used in formulation of  quantum advantage experiments. For more details on the Haar measure, see Ref.~\cite{Mele_2024}.
    The expectation values over unitaries drawn from $\mathrm{Haar}(m)$ are can be written as integration with respect to the Haar measure:
    \begin{equation}
        \underset{U \sim \mu_m }{\mathbb{E}}\!\left[f(U)\right]
        \coloneqq \int_{U \in \mathrm{U}(m)} f(U)\,\mathrm{d}\mu_{\mathrm{Haar}}(U) . 
    \end{equation}
    where \(f:\mathrm{U}(m)\to\mathbb{C}\) is an  arbitrary integrable function on \(\mathrm{U}(m)\).
    Similarly, let \(\mathcal{U}_{m,n}\) denote the uniform distribution on \(\mathcal{S}_{m,n}\). Then, for any function $X:\mathcal{S}_{m,n}\rightarrow \mathbb{C}$ we have
    \begin{equation}
        \underset{\bm{n}\sim \mathcal{U}_{m,n}}{\mathbb{E}}\!\left[X_{\bm{n}}\right]
        \coloneqq
        \frac{1}{|\mathcal{S}_{m,n}|}\sum_{\bm{n}\in\mathcal{S}_{m,n}} X_{\bm{n}}.
    \end{equation}

    We conclude this part with a formal formulation on linear cross-entropy  benchmarking (LXEB) for photonic quantum advantage proposals (cf. \cite{martinezcifuentes2024linearcrossentropy} for extended discussion of this topic). Consider a probability distribution $p_{\rho,U}$ corresponding to outputs statistics at the output of the interferometer $U$, for initial state $\rho$ (cf. Eq. \eqref{eq:particle_detection_probability}). Assume that $\rho$ is characterized by the fixed particle number\footnote{Since $\varphi(U)$ preserves particle number, this condition can be enforced by considering ``postselected'' probability distribution $p_{\rho,U}$ conditioned on observing specific total photon number $n$.} $n$ and thus $\bm{n}\in\mathcal{S}_{m,n}$.  
    Let $q$ be another probability distribution corresponding to experimental realization of the sampling scheme corresponding to $\rho$ and $U$.
    The linear cross entropy between $p_{\rho,U}$ and $q$ is defined by
    \begin{equation}\label{LEXdef} 
        \mathrm{LXE}(p_{\rho,U}, q) \coloneqq \underset{\bm{n}\sim q}{\mathbb{E}}\  p_{\rho,U}(\bm{n}) = \sum_{\bm{n}\in\mathcal{S}_{m,n}} q(\bm{n}) p_{\rho,U}(\bm{n})\ .
    \end{equation}
    This measure indicates correlations between \emph{ideal} probabilistic distribution $p_{\rho,U}, q)$ (that may be hard to sample from classically) and experimental distribution $q$. If $q = p_{\rho,U}$ (i.e., $q$ perfectly implements ideal probability distribution) the LXEB score equals the $2$-norm of $p_{\rho,U}$, $\mathrm{LXE}(p_{\rho,U}, p_{\rho,U})= \sum_{\bm{n}\in\mathcal{S}_{m,n}} p_{\rho,U}(\bm{n})^2$. 
    The \emph{reference value} for LXEB for a given input state $\rho$ is obtained by taking $\mathrm{LXE}(p_{\rho,U}, p_{\rho,U})$ and averaging it over $U\sim\mathrm{Haar}(m)$, yielding
    \begin{equation}\label{eq:ref_def}
        \mathrm{LXE}_{\mathrm{ref}}^{(n)}(\rho)
        \coloneqq
        \underset{U\sim\mathrm{Haar}(m)}{\mathbb E} \sum_{\bm n\in\mathcal S_{m,n}} p_{\rho,U}(\bm n)^2.
    \end{equation}
    This quantity is central in our analysis. From the perspective of LXEB it sets the normalization for the associated LXEB fidelity (see \cref{sec:numerical_experiments}) and therefore plays the role of the fundamental benchmark reference value.
    
    In the  LXEB test, which is defined for a given input state $\rho$ and unitary $U\in\mathrm{U}(m)$, we collect now $s$ independent samples $\mathcal{X}=\{\bm{n}_1,\ldots,\bm{n}_s\}$, where $\bm{n}_i \sim q$. We define empirical LXEB score by
    \begin{equation}\label{eq:empiricalLXE}
        \widehat{\mathrm{LXE}}(p_{\rho,U}, \mathcal{X})\coloneq \frac{1}{s} \sum_{i=1}^s p_{\rho,U}(\bm{n}_i)\ .
    \end{equation}
    Clearly, $ \widehat{\mathrm{LXE}}(p_{\rho,U}, \mathcal{X})$ is an unbiased estimator of $\mathrm{LXE}(p_{\rho,U}, q)$ for $\bm{n}_i \sim q$.   
    We say that a collection of samples $\mathcal{X}$ passes the LXEB test if
    \begin{equation}
        \frac{|\,\widehat{\mathrm{LXE}}(p_{\rho, U}, \mathcal{X}) - \mathrm{LXE}_{\mathrm{ref}}^{(n)}(\rho)\,|}{\mathrm{LXE}_{\mathrm{ref}}^{(n)}(\rho)} < \Delta
    \end{equation}
    for some $\Delta \in[0,1)$, i.e., when $\widehat{\mathrm{LXE}}(p_{\rho,U}, \mathcal{X})$ approximates $\mathrm{LXE}_{\mathrm{ref}}^{(n)}(\rho)$ up to relative error $\Delta$ (see \Cref{sec:discussion} for a discussion of well-posedness of this task). 

    Throughout the paper, we use standard asymptotic notation: for $f,g$ nonnegative functions, $f(x)=O(g(x))$ means that $f$ grows no faster than $g$ up to a constant factor, while $f(x)=\Omega(g(x))$ means that $f$ grows at least as fast as $g$. Finally, $f(x)=\Theta(g(x))$ means that both statements hold, i.e., $f$ and $g$ have the same asymptotic scaling up to constant factors.

    Finally, we adopt the notation of the rising Pochhammer symbol:
    \begin{equation}
        (a)_b \coloneqq
        \begin{cases}
            1, & b=0,\\
            a(a+1)\cdots(a+b-1), & b\ge 1,
        \end{cases}
    \end{equation}
    where $b\in\mathbb{N}_0$.

\section{Overview of Results}
    \label{sec:main_results}

    Below, we state the main results of our work on LXEB and anticoncentration in photonic quantum advantage experiments. We also discuss the consequences and scope of these findings, and indicate where in the manuscript further details can be found.

    \subsubsection*{Linear cross-entropy benchmarking}

    \begin{result}[Efficient LXEB reference value computation]\label{res:LXEB}
        Let
        \(
        \rho=\bigotimes_{a=1}^m \sigma^{(a)}
        \)
        be an \(m\)-mode product state, where each \(\sigma^{(a)}\) is a single-mode state.  Assume that for each mode \(a\), the matrix elements in the single-mode Fock basis
        \(
        \langle k \vert \sigma^{(a)} \vert \ell \rangle\) for 
        \( 0 \le k,\,\ell \le n,
        \)
        are given. Then the LXEB reference value in the $n$ particle sector
        \[
        \mathrm{LXE}_{\mathrm{ref}}^{(n)}(\rho)
        \coloneqq
        \sum_{\bm n \in \mathcal S_{m,n}}
        \underset{U \sim \mathrm{Haar}(m)}{\mathbb E}
        [p_{\rho,U}(\bm n)^2],
        \]
        can be computed classically in time $\mathrm{poly}(m,n)$.
    \end{result}

    The technical ingredients behind this statement are developed mostly in \Cref{sec:framework}. There, we show that the Haar-averaged second moments of the output probabilities, and hence the reference value $\mathrm{LXE}_{\mathrm{ref}}^{(n)}(\rho)$, can be written as a finite linear combination of bosonic-swap $\mathbb{S}_q$ expectation values on $\rho^{\otimes 2}$. Equivalently, for two identical copies this becomes a linear combination of purities of particle-reduced density matrices of the relevant fixed-$n$ component of the input state, which makes explicit the role played by particle entanglement in the LXEB problem.

    The scope of Result \ref{res:LXEB} is very broad. It covers standard Boson Sampling~\cite{aaronson:2010}, Scattershot Boson Sampling~\cite{Bentivegna_2015}, Gaussian Boson Sampling~\cite{hamilton2017}, Superposition Boson Sampling~\cite{Renema_2020}, and displaced variants such as Displaced Gaussian Boson Sampling~\cite{Li2025}. The common reason is that in these models the input state can be assumed to be a product across modes. The same reasoning also applies to lossy variants, including nonuniform losses whenever they can be commuted onto the input state~\cite{Oszmaniec_2018}, since this only modifies the initial state while preserving the product structure. Importantly, the procedure works for arbitrary relations between the number of modes $m$ and the detected particle number $n$. In summary, Result \ref{res:LXEB} provides a tractable classical procedure for evaluating benchmark reference values in experimentally relevant regimes across a wide range of photonic quantum advantage proposals.

    The Fock basis coefficients of the constituent single-mode states truncated to a fixed particle number $n$ can be obtained, for example, using \textsc{Piquasso}~\cite{Kolarovszki_2025}. In \Cref{sec:numerical_experiments}, we test how empirical estimators converge to the corresponding reference values for Boson Sampling and Gaussian Boson Sampling in the saturated regime, and we further explore how losses affect the LXEB scores.

    In several important cases, we go beyond efficient computation and derive explicit closed formulas for $\mathrm{LXE}_{\mathrm{ref}}$, which is important for the asymptotic analysis of anticoncentration. The scenarios for which $\mathrm{LXE}_{\mathrm{ref}}$ can be computed exactly include Boson Sampling (including the case of uniform losses), Scattershot Boson Sampling, and Gaussian Boson Sampling; these are presented in \Cref{sec:specific_experiments}.

    \subsubsection*{Anticoncentration in the saturated regime}

    Recent work by \cite{bouland2025complexitytheoreticfoundationsbosonsamplinglinear} substantially narrowed the gap between complexity-theoretic hardness arguments and realistic photonic experiments by extending evidence for
    hardness from the traditional dilute regime~\cite{aaronson:2010} to the experimentally relevant saturated regime $m = \Theta(n)$ for Boson Sampling and a variant of Gaussian Boson Sampling. Our results complement this development from a different angle.
    The explicit expression for the LXEB reference value, combined with the Paley-Zygmund inequality~\cite{PETROV20072703}, yields an average-case anticoncentration theorem over output configurations $\bm{n}\in\mathcal{S}_{m,n}$, showing that a constant fraction of output probabilities remains of order inverse of the size of the sample space.

    \begin{result}[Average anticoncentration in saturated-regime Boson Sampling]\label{res:anticon}
        Let $p_U$ be a probability distribution in a Boson Sampling experiment initialized with a standard collision-free input $\bm{n}_0=(1^n 0^{m-n})$. For every \(\tau\in[0,1]\) and every \(n\ge 1\), we have the inequality
        \[
            \underset{\substack{U\sim\mathrm{Haar}(m)\\ \bm n\sim \mathcal U_{m,n}}}{\Pr}
            \left(
                p_U(\bm n)\ge \frac{\tau}{|\mathcal S_{m,n}|}
            \right)
            \ge
            \frac{(1-\tau)^2}{C_{m, n}},
        \]
        where $C_{m, n} = \const$.
        In particular, when \(m=\Theta(n)\), this lower bound remains bounded away from zero uniformly in \(n\). 
    \end{result}

    The appearance of the uniform distribution \(\mathcal{U}_{m,n}\) over the possible outcomes \(\bm{n}\in\mathcal{S}_{m,n}\) reflects the absence of the hiding property in the saturated regime~\cite{bouland2025complexitytheoreticfoundationsbosonsamplinglinear}. In particular, when $m = \Theta(n)$, the outcomes \(\bm{n}, \bm{n}' \in \mathcal{S}_{m, n}\) are inequivalent unless are related by a permutation of modes. 

    Although the notion of anticoncentration in Result \ref{res:anticon} is weaker than the one in the original Permanent of Gaussians anticoncentration conjecture of Ref.~\cite{aaronson:2010}, it is still sufficient for the relevant complexity-theoretic application. As we explain in \Cref{sec:anticoncentration}, average anticoncentration upgrades the additive estimate produced by Stockmeyer's approximate counting algorithm~\cite{stockmeyer} into a multiplicative estimate on a constant fraction of instances. Consequently, an efficient classical sampler approximating the Boson Sampling distribution to total variation error \(\beta=1/\poly(n)\) would yield an algorithm for solving \( |\mathrm{SUPER}|^2\)---the problem introduced in Ref.~\cite{bouland2025complexitytheoreticfoundationsbosonsamplinglinear} to justify hardness of Boson Sampling in the saturated regime---up to multiplicative error on a non-negligible fraction of instances. Thus, if one conjectures that this multiplicative approximation task for \( |\mathrm{SUPER}|^2\) is \(\#\mathsf{P}\)-hard on a non-negligible fraction of instances, then the existence of such a sampler would imply, by the standard arguments based on Toda's theorem, a collapse of the polynomial hierarchy. In this sense, average anticoncentration places saturated Boson Sampling on a footing comparable to RCS, IQP and Fermion Sampling, where analogous multiplicative average-case hardness assumptions are made ~\cite{Hangleiter2018anticoncentration,Bremner_2017,Oszmaniec_2022}.

    The formal proof of Result \ref{res:anticon} is given in \Cref{sec:anticoncentration}, which uses the Paley-Zygmund inequality and a careful upper bound on LXEB reference values for regular Boson Sampling. Importantly, we get bounds on anticoncentration for an arbitrary $n$-photon input state\footnote{If input state does not have a fixed number of photons we consider part of the statistics characterized by a fixed photon number $n$, c.f. Section \ref{sec:anticoncentration}. } $\rho$:
    \begin{equation}\label{eq:acMAIN}
    \underset{\substack{U\sim\mathrm{Haar}(m)\\ \bm n\sim \mathcal U_{m,n}}}{\Pr}
        \left(
            p_{\rho,U}(\bm n)\ge \frac{\tau}{|\mathcal S_{m,n}|}
        \right)
        \ge
        \frac{(1-\tau)^2}{\mathrm{AC}^{(n)}_\rho},
    \end{equation}
    where $ \mathrm{AC}^{(n)}_\rho
            = |\mathcal{S}_{m,n}| \times\mathrm{LXE}_{\mathrm{ref}}^{(n)}(\rho)
    $.
    The bounds based on second moments turn out to be too weak to establish strong (i.e., corresponding to a constant denominator in the right hand side of \eqref{eq:acMAIN}) anticoncentration for experimentally relevant Gaussian Boson Sampling and Scattershot Boson Samplings schemes in the saturated regime. However, for these models it is possible to prove \textit{weak anticoncentration}~\cite{aaronson:2010}, which corresponds to $\mathrm{AC}^{(n)}_\rho=\mathrm{poly}(m,n)$, at least for the case where the number of sources $d$ of squeezed states (the remaining $m-d$ modes are initialized with vacuum) is comparable the total number of  modes $m$. On the other hand, when we only have a single source ($d=1$), then $\mathrm{AC}^{(n)}_\rho$ exhibits exponential scaling and even weak anticoncentration fails. We note, that these observations are reminiscent of the anticoncentration phenomena previously observed for Gaussian Boson Sampling in the dilute regime~\cite{PhaseTransitionsRCSRQC}. The details are given in Proposition \ref{prop:scatter_asymptotic_bounds} and \ref{prop:gaussian_asymptotic_bounds} formulated in Section \ref{sec:anticoncentration}.

   \subsubsection*{Further results and consequences}

    \noindent The framework developed in this work has also further consequences beyond already discussed results.

    \medskip
    
    \noindent\textbf{Sample complexity of certification tests.}
        Following a line of thought of Ref.~\cite{Hangleiter_2019}, we can use the second moment of the output distribution to give a bound on the sample complexity of certifying different Boson Sampling schemes. As described in detail in \Cref{app:sample_complexity}, the certification of the output distribution requires at least 
        \begin{equation}
            \Omega\!\left(|\mathcal E|^{1/4}/ \sqrt{\mathrm{AC}_\rho^{(n)}}\right)
        \end{equation}
        samples for a typical Haar-random interferometer, where \(|\mathcal E|\) denotes the number of possible outcomes\footnote{Typically $|\mathcal E|=|\mathcal{S}_{m,n}|$, but for Scattershot Boson Sampling, that slightly diverges from the setup in main part, we have $|\mathcal E|=|\mathcal{S}_{m,n}| |\mathcal{S}_{d,n}|$, where $d$ is the number of two-mode squeezed states in the input.}. Since, in the parameter regimes relevant to quantum advantage, \(\mathrm{AC}_\rho^{(n)}\) grows at most polynomially in the system size, whereas \(|\mathcal E|\) grows exponentially in 
        all the considered photonic quantum advantage schemes (Boson Sampling and its Scattershot and Gaussian variants), sample-efficient non-interactive device-independent certification is impossible for these schemes. 

    \medskip

    \noindent \textbf{Proof of Hunter-Jones conjecture for $t=2$.}
        As a further application of our framework, we establish the following asymptotic moment statement, which can be viewed as a rigorous confirmation of the \(t=2\) case of the Hunter-Jones conjecture. In Ref.~\cite{nezami2021permanentrandommatricesrepresentation}, the author conjectured an explicit asymptotic formula for Haar moments of \(|\!\per(U)|^2\), where \(U\sim\mathrm{Haar}(n)\). While the case \(t=1\) is immediate, the case \(t=2\) had only been supported by numerical evidence. A consequence of our results is a proof of this first nontrivial instance.
        Let \(U\sim \mathrm{Haar}(n)\), and define
        $
            X_U \coloneqq |\! \per(U)|^2
        $. Then
        \begin{equation}
            \lim_{n\to\infty}
            \frac{
                \underset{U\sim \mathrm{Haar}(n)}{\mathbb{E}}\!\left[X_U^2\right]
            }{
                \left(
                \underset{U\sim \mathrm{Haar}(n)}{\mathbb{E}}\!\left[X_U\right]
                \right)^2
            }
            = 2.
        \end{equation}
        The proof of this fact can be found in \Cref{app:proof_hunter_jones}.
        To the best of our knowledge, this is the first rigorous proof of the \(t=2\) case of the Hunter-Jones conjecture.

    \medskip
   
    \noindent \textbf{R\'enyi-2 entropies.}
        A central object in our analysis is the bosonic swap operator $\mathbb{S}_q$, which implements a permutation-symmetric swap of $q$ out of $n$ particles. Its expectation value in $\rho \otimes \rho$ directly yields the purity of $q$-particle reduced density matrix
        $\Tr\left[\mathbb{S}_q\,\rho^{\otimes 2}\right]
        =
        \Tr\left[\Tr_{n-q}(\rho)^2\right]$,
        the negative logarithm of which is the R\'enyi-2 entropy. 
        Our framework provides an efficient method to evaluate this quantity, making it a practical numerical tool for studying particle entanglement in highly symmetric bosonic states, including generalized Dicke states. On the analytical side, we derive a closed-form expression for the purity of the particle-reduced density matrix averaged uniformly over all Fock states with fixed particle number $n$
        and mode number $m$, from which a lower bound on the average Rényi-2 entropy follows. In the regime $m = \Theta(n)$, this lower bound grows as $\Theta(n)$, i.e., a volume law holds. For details, see \Cref{app:entanglement}.

\section{Outlook and open problems}\label{sec:discussion}
    In this work, we have developed representation-theoretic methods for computing expected LXEB scores across a broad class of photonic quantum advantage schemes. The same second-moment techniques enabled us to prove anticoncentration for Fock-state Boson Sampling in the saturated regime ($m = \Theta (n)$), thereby strengthening the hardness guarantees for this quantum advantage proposal~\cite{bouland2025complexitytheoreticfoundationsbosonsamplinglinear}. Beyond this progress, several interesting open problems remain.

    A natural direction for future work is the extension of our results to scenarios in which photon-number resolution is replaced by \emph{threshold detectors}, which distinguish only between vacuum and non-vacuum outcomes~\cite{Quesada_2018} In this setting, measurement outcomes no longer correspond to fixed photon-number subspaces. Our results also do not cover \emph{Bipartite Gaussian Boson Sampling}~\cite{Grier_2022}, though we expect that this case can be addressed with closely related techniques. 
    
    Another natural generalization is the treatment of \emph{nonuniform optical losses}~\cite{Oszmaniec_2018,Renema2026mitig}: while uniform losses are naturally accommodated by the present framework, spatially varying loss breaks the symmetries underlying the uniform-case derivations. Extending our LXEB reference value results to \textit{partially distinguishable photons} is a further open direction of direct experimental relevance~\cite{Shchesnovich_2015,Renema_2018}.
    
    Several foundational questions concerning LXEB in photonic quantum computing also remain open. We did not rigorously analyze the convergence of the empirical estimator 
    $\widehat{\mathrm{LXE}}(p_{\rho,U},\mathcal{X})$ to $\mathrm{LXE}(p_{\rho,U}, p_{\rho,U})$. Similarly, we did not prove, that with high probability over Haar random $U$, the LXEB score
    concentrates around its reference value, although this is strongly supported by our numerical experiments (see \Cref{sec:numerical_experiments}).
    Establishing results of this kind would require control over higher moments of \(p_{\rho,U}(\bm{n})\) which we believe is within reach using similar techniques---note that a formal proof of analogous concentration behavior for RCS-based LXEB tests has only recently appeared~\cite{NickJonas2026}. It would also be worthwhile to explore the photonic analogue of the XQUATH assumption of~\cite{AaronsonGunn2020XEB} and examine its implications for spoofing LXEB tests in photonic quantum advantage schemes.

    Furthermore, techniques introduced in this work could be useful in the context of \emph{Bosonic RB}~\cite{Arienzo_2025,Valahu_2024}, a method which is currently developed mostly in the dilute regime ($m = \Omega(n^2)$) and suffer from computational overheads resulting from complicated group-theoretics structures involved, i.e., summations over multiple Clebsch-Gordan coefficients. We believe that the application of the first-quantization-inspired techniques developed in this paper could make Bosonic RB more efficient and more amenable to the analysis of noise in realistic quantum advantage experiments operating in the saturated regime.

    Finally, from a structural perspective, it would be interesting to understand higher-order invariants of passive linear optics in the language of second quantization, without explicitly relying of first quantization techniques and Schur-Weyl duality.

\section*{Acknowledgement}
    The authors used AI tools, in particular ChatGPT 5.2 PRO and 5.4 Pro, in the preparation of this manuscript to assist with symbolic manipulations, simplification of derivations, tightening of bounds for the AC scores, and formatting. Notably, the formula stated in \Cref{prop:mean_purity_simple} was suggested by AI; this identity substantially simplified the proof of anticoncentration for Boson Sampling in the saturated regime. AI assistance also helped identify several useful simplifications in combinatorial expressions, including the applicability of the Pfaff-Saalschütz formula (cf. \Cref{app:useful_comb}). All such suggestions and derivations were subsequently checked and validated by the authors, who are fully responsible for the final content.

    We thank Michał Kotowski, Marcin Kotowski, Bartosz Naskr\k{e}cki, Marcel Hinsche, Ingo Roth and Markus Heinrich for interesting discussions. We also thank Yan Schmidt for bringing minor mistakes in an earlier version of the Appendix to our attention.
    The authors would also like to thank the support of the Hungarian National Research, Development and Innovation Office (NKFIH) through the Quantum Information National Laboratory of Hungary and the grants TKP-2021-NVA-04, TKP-2021-NVA-29 and FK 135220. ZZ was partially supported by the Horizon Europe programme HORIZON-CL4-2022-QUANTUM-01-SGA via the project 101113946 OpenSuperQPlus100 and the QuantERA II project HQCC-101017733. 
    MO acknowledges the support from the European Union’s Horizon Europe research and innovation program under EPIQUE Project GA No. 101135288, and from National Science Center, Poland within the QuantERA III Programme (No 2023/05/Y/ST2/00140 acronym Tuquan). 
    The C4QEC project is carried out within the IRAP of the Foundation for Polish Science co-financed by the European Union.
    The authors also acknowledge the computational resources provided by the Wigner Scientific Computational Laboratory (WSCLAB).

    \emph{Note added.} Upon completion of this work, we became aware of independent work by Mhiri et al.~\cite{mhiri2026} that studies anticoncentration and LXEB in Boson Sampling using different group-theoretic techniques and proves anticoncentration in the saturated regime.

\section{Framework}\label{sec:framework}
   
   In this section, we develop a general framework to compute second moments of output probability distributions $p_{\rho,U}$ and LXEB reference values across a broad class of photonic quantum advantage experiments. The tools developed here can also be used to compute general averages involving two copies $\varphi_n (U)\otimes \varphi_n(U)$ of the bosonic representation of $\mathrm{U}(m)$.

   In our derivations, we will adapt the first-quantization perspective and the Schur-Weyl duality, linking representation theory of unitary and symmetric groups acting on tensor product space $(\mathbb{C}^m)^{\otimes t}$ ~\cite{fultonharris}. In particular, irreducible representations of $\mathrm{U}(m)$ in this space are indexed by Young diagrams $\lambda$ of size $t$ and height $h(\lambda)\leq m$ \footnote{A Young diagram $\lambda$ of size $t$ is a partition $\lambda\vdash t$, i.e., a weakly decreasing sequence of nonnegative integers $\lambda=(\lambda_1,\lambda_2,\ldots)$ satisfying $\sum_i \lambda_i=t$. It is represented graphically as a left-justified array of boxes with $\lambda_i$ boxes in the $i$-th row. The height of $\lambda$, $h(\lambda)$ is a number of its nonzero rows.}. For instance, the fully symmetric subspace $\Sym^t(\mathbb{C}^m)\subset (\mathbb{C}^m)^{\otimes t}$ corresponds to the one-row Young diagram $(t)$, while the fully antisymmetric subspace $\bigwedge^t(\mathbb{C}^m)\subset (\mathbb{C}^m)^{\otimes t}$  corresponds to the one-column Young diagram $(1,\ldots,1)$.
    
    The Haar averages we need to evaluate contain the average with respect to two-copy $\varphi_n^{\otimes 2}$ representation of $\mathrm{U}(m)$. We will evaluate these averages by using Schur's lemma and decomposing $\Sym^n(\mathbb{C}^m)^{\otimes 2}$ into irreducible subspaces.
    According to the
    Pieri rule~\cite{serafini}, we have
    \begin{equation} \label{eq:irrep_decomp}
        \Sym^n(\mathbb{C}^m) \otimes \Sym^n(\mathbb{C}^m) \cong \bigoplus_{k=0}^n \mathbb{S}_{(2n-k, k)}(\mathbb{C}^m),
    \end{equation}
    where $\mathbb{S}_{(2n-k,\,k)}(\mathbb{C}^m)$ denotes the irreducible
    $\mathrm{U}(m)$-representation associated with the Young diagram having two rows
    of lengths $2n-k$ and $k$.
    Consequently, the representation $\varphi_n\otimes\varphi_n$
    decomposes as
    \begin{equation}
        \varphi_n(U)\otimes\varphi_n(U)
        \;\cong\;
        \bigoplus_{k=0}^{n}\Pi_k (U),
    \end{equation}
    where each irreducible representation $\Pi_k$ is supported on the subspace
    $\mathbb{S}_{(2n-k,\,k)}(\mathbb{C}^m)$, and are mutually
    inequivalent.

    Let us denote the projector onto the subspace $\mathbb{S}_{(2n-k,\,k)}(\mathbb{C}^m)$ by $\mathbb{P}_k$. Our key contribution is to make these projectors explicit and to express them in a form adapted to bosonic systems. To this end, we construct them from restrictions of  sums over Young symmetrizes corresponding to diagrams $\lambda_k=(2n-k,k)$ acting on $(\mathbb{C}^m)^{\otimes 2n}$, then rewrite them as linear combinations of a family of operators that we call \emph{bosonic swap operators} (the formal proof is provided in \Cref{app:expressing_the_projector}). 
  
    \begin{restatable}[Explicit expression for $\mathbb{P}_k$]{lemma}{lemmapk}\label{lem:explicit_P_k}
        We can write the projector $\mathbb{P}_k$ as the following sum:
        \begin{equation}\label{eq:swap_expansion}
            \mathbb{P}_k
            =
            \sum_{q=0}^{n} c_{k,q}\,\mathbb{S}_q,
        \end{equation}
        where the coefficients are
        \begin{equation}\label{eq:coefficients}
            c_{k,q}\!
            =\!
            \frac{2n-2k+1}{2n-k+1}\binom{n}{k}\binom{n}{q}\!\!
            \sum_{l=\max(0,k-q)}^{\min(k,n-q)}
            \!\!\!\!\!(-1)^{k-l}
            \frac{\binom{k}{l}\binom{n-k}{l+q-k}}{\binom{2n-k}{n-l}}.
        \end{equation}
        and $\mathbb{S}_q:  \Sym^n(\mathbb{C}^m)^{\otimes 2} \rightarrow   \Sym^n(\mathbb{C}^m)^{\otimes 2} $ is the bosonic swap operator defined by
        \begin{equation}\label{eq:bosonic_swap_to_trace}
            \Tr [\mathbb{S}_q (\rho \otimes \sigma) ]
            = \Tr [ \Tr_{n-q}(\rho) \Tr_{n-q}(\sigma)],
        \end{equation}
        where $\Tr_{n-q}$ denotes the partial trace over $(n-q)$ particles.
    \end{restatable}
    
    This decomposition is particularly useful because the bosonic swap operator \(\mathbb{S}_q\) admits a transparent interpretation in first quantization (as operator acting on $(\mathbb{C}^m)^{\otimes 2n}$), namely
    \begin{equation}
        \mathbb{S}_q = \mathbb{P}_{\mathrm{sym}}^{\otimes 2} \left( \prod_{i=1}^q \mathbb{F}_{i, i+n} \right) \mathbb{P}_{\mathrm{sym}}^{\otimes 2} \ ,
    \end{equation}
    where $\mathbb{P}_{\mathrm{sym}}$ is the projector onto $\Sym^n(\mathbb{C}^m)$ and $\mathbb{F}_{i,j}$ is the swap operator between $i$-th and $j$-th component of $(\mathbb{C}^m)^{\otimes 2n}$. In other words, $\mathbb{S}_q$ exchanges \(q\) particles between the two copies of symmetric subspace uniformly. We emphasize that the decomposition \eqref{eq:swap_expansion} formally does not depend on the number of modes $m$. However, this parameter enters in the dimension of $\mathbb{S}_{(2n-k, k)}(\mathbb{C}^m)$ as
    \begin{equation}\label{eq:dimensions}
        \Tr \mathbb{P}_k = \frac{2n-2k+1}{2n-k+1}
        \binom{2n+m-k-1}{m-1}\binom{m-2+k}{m-2}\ .
    \end{equation}

    Let us start by deriving a closed-form expression for the second moment $\mathbb{E}_{U\sim\mathrm{Haar}(m)} [p_{\rho,U}(\bm{n})]^2$. Using \eqref{eq:particle_detection_probability} and the ``replica trick'' $\Tr[AB]^2=\Tr[A^{\otimes 2} B^{\otimes 2}]$, we get 
    \begin{multline}\label{eq:expansion}
    \underset{U\sim\mathrm{Haar}(m)}{\mathbb{E}} [p_{\rho,U}(\bm{n})^2] \\= \underset{U\sim\mathrm{Haar}(m)}{\mathbb{E}} \Tr\!\left[
                \varphi_n(U)\, \rho \,\varphi_n(U)^\dagger
                \ketbra{\bm{n}}\right]^2 
              \\=  \Tr\!\left[ \mathcal{T}(\rho^{\otimes 2})
                \ketbra{\bm{n}}^{\otimes 2}\right]\ ,
    \end{multline}
    where for a linear operator $X$ on $\Sym^n(\mathbb{C}^m)^{\otimes 2}$ we defined the twirling operation
    \begin{equation}
        \mathcal{T}(X) \coloneqq \underset{U \sim \mathrm{Haar}(m)}{\mathbb{E}} 
        [\varphi_n(U)^{\otimes 2} X (\varphi_n(U)^{\otimes 2})^\dagger] ,
    \end{equation}
    Since for every $U\in \mathrm{U}(m)$ we have  $[\mathcal{T}(\rho^{\otimes 2}),\varphi_n(U)^{\otimes 2}]=0$  (i.e., $\mathcal{T}(\rho^{\otimes 2})$ lies in the commutant of $\varphi_n^{\otimes 2}$) and $\varphi_n^{\otimes 2}$ is multiplicity-free, by Schur's lemma, we have
    \begin{equation}\label{eq:twirl}
         \mathcal{T}(\rho^{\otimes 2})
        = \sum_{k=0}^n \frac{\Tr [\mathbb{P}_k \rho^{\otimes 2}]}{\Tr \mathbb{P}_k} \mathbb{P}_k.
    \end{equation}
    By inserting the above to \eqref{eq:expansion} and by using $\Tr[\mathbb{P}_k  \ketbra{\bm{n}}^{\otimes 2}]=0$ for odd $k$ (this is because $\mathbb{P}_k$ has support in the antisymmetric subspace of two copies of bosonic space, $\bigwedge^2 (\Sym^n(\mathbb{C}^m))$, c.f. \Cref{remark:symmetry} in \Cref{app:expressing_the_projector})) we obtain
    
     \begin{restatable}{proposition}{secondMoment}\label{prop:secondmoment}
       For arbitrary state $\rho$ supported on $\mathrm{Sym}^n(\mathbb{C}^m)$ and $\bm{n}\in \mathcal{S}_{m,n}$ we have
        \begin{equation}\label{eq:SECONDmom}
            \underset{U\sim\mathrm{Haar}(m)}{\mathbb{E}} [p_{\rho,U}(\bm{n})^2]\!=\!
            \sum_{r=0 }^{\lfloor n/2 \rfloor}  
            \frac{\Tr [\mathbb{P}_{2r} \rho^{\otimes 2}] }{\Tr \mathbb{P}_{2r}} \Tr [\mathbb{P}_{2r} \ketbra{\bm{n}}^{\otimes 2}].
        \end{equation}
    \end{restatable}
    \noindent By summing \eqref{eq:lxe_ref_boxed} over $\bm{n}\in \mathcal{S}_{m,n}$ and using the definition of LXEB reference score from \eqref{eq:ref_def} we obtain 
   \begin{equation}\label{eq:lxe_ref_boxed}
        \boxed{
            \mathrm{LXE}_{\text{ref}}^{(n)}(\rho) \!=\!| \mathcal{S}_{m, n} |   \sum_{r=0 }^{\lfloor n/2 \rfloor}  
            \frac{\Tr [\mathbb{P}_{2r} \rho^{\otimes 2}] }{\Tr \mathbb{P}_{2r}} \Tr [\mathbb{P}_{2r}  \mathbb{D}_{m, n}],
        }
    \end{equation}
    where 
    \begin{equation}
        \mathbb{D}_{m, n} \coloneqq \underset{\bm{n} \sim \mathcal{U}_{m,n}}{\mathbb{E}} [\ketbra{\bm{n}}^{\otimes 2}]
    \end{equation}
    represents the average of the projectors corresponding to all outcomes in $\mathcal{S}_{m,n}$.

    To proceed, we express $\mathbb{P}_{2r}$ in terms of bosonic swaps using \eqref{eq:coefficients}: $\mathbb{P}_{2r}=
            \sum_{q=0}^{n} c_{2r,q}\,\mathbb{S}_q$. For the case of $\mathbb{D}_{m,n}$ this strategy yields the following Proposition, proven in \Cref{app:simplified_formulas}, which gives us the exact value of $\Tr[\mathbb{P}_{2r} \mathbb{D}_{m,n}]$.

     \begin{restatable}{lemma}{trpkdmn}\label{lem:tr_p2r_dmn}
        For $0\leq r \leq \lfloor n/2\rfloor$ we can write $\Tr\!\left[\mathbb P_{2r} \mathbb{D}_{m, n} \right]$ as
        \begin{equation}
            \Tr [\mathbb{P}_{2r}\mathbb{D}_{m, n}] \!= \!
            \frac{2n-4r+1}{2n-2r+1}
            \frac{\binom{n-r}{r}}{\binom{2n-2r}{n}}
            \frac{(m+n)_{n-2r}\left(\frac{m-1}{2}\right)_r}
            {\left(\frac{m+1}{2}\right)_{n-r}}\ .
        \end{equation}
    \end{restatable}

    At this point, we should point out that the problem of computing $\mathrm{LXE}_{\text{ref}}^{(n)}(\rho)$ boils down to computation of expectation values of bosonic swap operators $\Tr[\mathbb{S}_q \rho^{\otimes 2}]$ needed to evaluate \eqref{eq:lxe_ref_boxed}. To this end, we developed a fast algorithm which enables to compute it efficiently for arbitrary product states (in modes). Specifically, we have the following result:
   
    \begin{restatable}{proposition}{swapexpval}
        \label{prop:efficient_swap_product}
        Let
        $\rho=\bigotimes_{a=1}^m \sigma^{(a)}$ be a product (across optical modes) of single-mode states, and let
        \begin{equation}
            \rho_{(n)} \coloneqq \frac{\rho|_{n}}{\Tr[\rho|_{n}]}
        \end{equation}
        denote its normalized $n$-particle component. Given the matrix elements of individual components in the single mode Fock basis $\bra{l} \sigma^{(a)} \ket{k}$ for $k,l\in\{0,\ldots,n\}$, $a\in\{1,\ldots,m\}$, the expectation value
        \begin{equation}
            \Tr\!\left[\mathbb{S}_q\,\rho_{(n)}^{\otimes 2}\right]
        \end{equation}
        can be computed in
        $
            O\!\left(m\,(n-q+1)^2(q+1)^2\log\bigl((n-q+1)(q+1)\bigr)\right)
        $
        time using
        $
            O\!\left((n-q+1)^2(q+1)^2\right)
        $
        memory.
    \end{restatable}
    \noindent The proof of this statement is given in \Cref{app:proof_bosonic_swap_efficient}, which uses coefficient extraction via truncated polynomial arithmetic through \Cref{prop:polynomial_machinery}, also given therein.
    We also note that similar result works for the products of non-single-mode states as well, and is efficient when the locality of the components can be bounded by a constant.

    For completeness, we present the algorithmic workflow for computing $\mathrm{LXE}_{\text{ref}}^{(n)}(\rho)$ in \Cref{fig:lxeb_algorithmic_workflow}, combining all the techniques presented in this section.
    We note that the reference values are not only efficiently computable, but often admit closed-form expressions, as we show in the following \Cref{sec:specific_experiments} for Fock-state Boson Sampling, Scattershot Boson Sampling and Gaussian Boson Sampling. These expressions also make the asymptotic behavior transparent for these models, which is especially useful for investigating anticoncentration in  quantum advantage experiments; see \Cref{sec:anticoncentration}.

    \begin{figure*}[t]
    \centering
    \begin{tikzpicture}[
    font=\small,
    node distance=6mm and 10mm,
    box/.style={draw, align=center, inner sep=5pt, text width=0.20\linewidth},
    branch/.style={draw, align=center, inner sep=5pt, text width=0.35\linewidth},
    final/.style={draw, align=center, inner sep=6pt, text width=0.55\linewidth},
    emph/.style={box, fill=blue!8},
    proc/.style={box, fill=gray!10},
    arrow/.style={-Latex, thick}
]

    \node[box, text width=0.25\linewidth] (nd) {\textbf{Total photon number} $n$ \\ \textbf{Number of modes} $\d$};
    \node[box, right=of nd, xshift=30mm] (rho) {\textbf{Input state} $\rho$};

    \node[box, below=20mm of nd, xshift=-10mm, text width=0.20\linewidth, xshift=-25mm] (trpk) {
        \textbf{Dimensions}\\ Compute $\Tr\,\mathbb{P}_{2r}$ from \Cref{eq:dimensions}.
    };

    \node[branch, below=22mm of rho, text width=5cm] (stateswap) {
        \textbf{$\mathbb{S}_q$ expectation values}\\[2pt]
        Compute $\Tr[\mathbb{S}_q\,\rho^{\otimes 2}]$\\
        using \Cref{prop:efficient_swap_product}.
    };
    \node[box, below=10mm of nd, xshift=32mm] (coeffs) {
        \textbf{Coefficients $c_{k,q}$}\\[2pt]
        Compute swap expansion coefficients from $\mathbb{P}_k=\sum_q c_{k,q}\mathbb{S}_q$\\
        using \Cref{eq:coefficients}.
    };

    \node[branch, below=45mm of nd, xshift=0mm, text width=5cm] (det) {
        \textbf{Detector-dependent factors}\\[2pt]
        Compute $
            \Tr [\mathbb{P}_{2r}  \; \mathbb{D}_{m, n}]
        $ \\
        directly using \Cref{lem:tr_p2r_dmn}.
    };

    \node[branch, below=10mm of stateswap] (statepk) {
        \textbf{Input state-dependent factors}\\[2pt]
        Combine $c_{k,q}$ and $\mathbb{S}_q$ expectation values by $\rho$ to get
        $ \displaystyle \Tr[\mathbb{P}_{2r}\,\rho^{\otimes 2}] $ using \Cref{prop:efficient_swap_product}.
    };
    
    \node[final, below=20mm of det, xshift=20mm, text width=7cm] (combine) {
        \textbf{LXEB reference value}\\[2pt]
        Combine detector-dependent and input state-dependent factors to get
        $$
            \displaystyle
            \mathrm{LXE}_{\mathrm{ref}}(\rho)
            =
            |\mathcal{S}_{\d,n}|
            \sum_{r=0}^{\lfloor n / 2 \rfloor }
            \frac{\Tr[\mathbb{P}_{2r}\,\rho^{\otimes 2}]}{\Tr\,\mathbb{P}_{2r}}
            \Tr [\mathbb{P}_{2r}  \; \mathbb{D}_{m, n}].
        $$
    };

    \draw[arrow] (nd) -| (trpk);
    \draw[arrow] (nd) -- (det);
    \draw[arrow] (nd) -| (coeffs);

    \draw[arrow] (rho) -- (stateswap);

    \draw[arrow] (stateswap) -- (statepk);
    \draw[arrow] (coeffs.south) |- (statepk);
    \draw[arrow] (trpk) |- (combine);
    \draw[arrow] (statepk) |- (combine.east);
    \coordinate (mid) at ($(det)+(0,-1.5cm)$);
    \draw[arrow] (det) |- (mid) -| (combine);
\end{tikzpicture}
    \caption{Algorithmic workflow for computing the LXEB reference value.
    Starting from the problem parameters $(n,m,\rho)$, the computation reduces the Haar
    average of $p_{\rho,U}(\bm{n})^2$ to expectation values of bosonic swap operators.
    The resulting expression separates into a detector-dependent part, involving uniform
    averages over $\mathcal{S}_{m,n}$, and an input-state-dependent part, involving
    $\Tr[\mathbb{S}_q\,\rho^{\otimes 2}]$, assembled using the coefficients $c_{k,q}$.
    This reduction enables an efficient classical evaluation of $\mathrm{LXE}_{\mathrm{ref}}^{(n)}(\rho)$.}
    \label{fig:lxeb_algorithmic_workflow}
    \end{figure*}

\section{Linear cross-entropy benchmarking in specific experiments}\label{sec:specific_experiments}
    In this section, we specialize the general LXEB framework to a number of concrete photonic quantum advantage experiments. For each setting, we identify the relevant output distribution, formulate the associated reference value, and show how it can be computed efficiently using the polynomial techniques introduced above. This demonstrates that the framework is sufficiently general to cover a broad range of experimentally relevant regimes while remaining amenable to explicit calculation.

    \subsection{Boson Sampling}\label{sec:lxeb_boson_sampling}
        As discussed in \Cref{sec:basics}, in the (Fock-state) Boson Sampling setting we consider a $m$-mode, $n$-photon collision-free (i.e., up to one particle per mode) Fock input state $\ket{\bm{n}_0}$ injected into an $m$-mode Haar-random linear interferometer $U \sim \mathrm{Haar}(m)$, terminated by particle number detectors.  

        In the following, starting from \Cref{eq:lxe_ref_boxed}, we will give a simple formula for $\mathrm{LXE}_{\mathrm{ref}}^{(n)}(\ketbra{\bm{n}_0})$, which is the relevant reference value for LXEB in Boson Sampling experiments. Before this, we analyze the bosonic swap expectation values of collision-free Fock states, using
        \begin{equation}
            \Tr\!\left[\mathbb S_q \ketbra{\bm{n}_0}^{\otimes 2}\right]
            =
            \Tr\!\left[\Tr_{n-q}(\ketbra{\bm{n}_0})^2\right],
        \end{equation}
        where we can write that
        \begin{equation}
            \Tr_{n-q}(\ketbra{\bm{n}_0})
            =
            \frac{1}{\binom{n}{q}}
                \sum_{\substack{\bm{0} \leq \bm{q} \leq \bm{n}_0 \\ |\bm{q}| = q }}
                \ketbra{\bm{q}}.
        \end{equation}
        Direct evaluation yields that the bosonic swap expectation value is
        \begin{equation}\label{eq:swapexp_collfree}
            \Tr\!\left[\mathbb{S}_q\,\ketbra{\bm{n}_0}^{\otimes 2}\right]
            = \frac{1}{\binom{n}{q}}.
        \end{equation}
        We can use this formula to get the following simplified formula for the LXEB reference value:
        \begin{proposition}\label{prop:bs_lxeref_simplified}
            We can write the reference value $\mathrm{LXE}_{\mathrm{ref}}^{(n)}(\ketbra{\bm{n}_0})$ in the form of
            \begin{multline}
                \mathrm{LXE}_{\mathrm{ref}}^{(n)}(\ketbra{\bm{n}_0}) \\= 2^{n} n! \sum_{r=0}^{\lfloor n/2\rfloor}
                \frac{2n-4r+1}{2n-2r+1}
                \frac{
                    16^{-r} \binom{2r}{r}
                }{
                    \binom{2n-2r}{\,n-r\,}
                    \left(\frac m2\right)_r
                    \left(\frac{m+1}{2}\right)_{n-r}
                }.
            \end{multline}
        \end{proposition}
        \begin{proof}
            We begin with \Cref{prop:trace_pk_n0}, which for even \(k=2r\) and using \Cref{eq:swapexp_collfree}, gives the compact closed-form expression
            \begin{equation}\label{eq:tr_pk_n0}
                \Tr\!\left[\mathbb{P}_{2r}\,\ketbra{\bm{n}_0}^{\otimes 2}\right]
                = 2^{\,n-2r}\,\frac{2n-4r+1}{2n-2r+1}\,
                \frac{\binom{n}{r}}{\binom{2n-2r}{\,n-r\,}}.
            \end{equation}
            Next, for \(\Tr[\mathbb{P}_{2r}\mathbb{D}_{m,n}]\), we invoke \Cref{lem:tr_p2r_dmn}, together with the simplified formula for the dimensions of the irreducible subspaces given in \Cref{eq:dimensions}.
            Combining these expressions and using the elementary identities
            \begin{subequations}
            \begin{align}
                \binom{m+n-1}{n}\,
                \frac{(m+n)_{n-2r}}{\binom{2n+m-2r-1}{m-1}}
                &=
                \frac{(2n-2r)!}{n!},
                \label{eq:id1_pf}
                \\[0.5ex]
                \frac{\left(\frac{m-1}{2}\right)_r}{\binom{m+2r-2}{m-2}}
                &=
                \frac{(2r)!}{4^r\left(\frac m2\right)_r},
                \label{eq:id2_pf}
                \\[0.5ex]
                \frac{\binom{n}{r}\binom{n-r}{r}}{\binom{2n-2r}{n}}
                \frac{(2n-2r)!}{n!}
                &=
                \frac{n!}{r!^2},
                \label{eq:id3_pf}
            \end{align}
            \end{subequations}
            we obtain the claimed formula.
        \end{proof}
        \noindent Although Boson Sampling was originally proposed in the dilute regime \(m=\Omega(n^2)\), our formula is not confined to this case. Rather, it extends naturally to other scaling regimes, including the saturated regime \(m=\Theta(n)\), whose classical hardness has also been studied in Ref.~\cite{bouland2025complexitytheoreticfoundationsbosonsamplinglinear}.
        
        \subsubsection{Uniform losses}\label{sssec:uniform_losses_bs}
            Here, we aim to show that the LXEB reference value in Boson Sampling can also be efficiently computed in the presence of uniform losses. The uniform loss channel $\Lambda_\eta$ parametrized by transmissivity $\eta$ acts on a density operator $\rho$ defined on the $n$-particle subspace as~\cite{Oszmaniec_2018}
            \begin{equation}\label{eq:loss_channel}
              \Lambda_\eta(\rho) \coloneqq \sum_{\ell = 0}^n \binom{n}{\ell} \eta^\ell (1-\eta)^{n-\ell} \Tr_{n-\ell} ( \rho ),
            \end{equation}
            where the probability that $\ell$ particles remain in the system is $\binom{n}{\ell} \eta^\ell (1-\eta)^{n-\ell}$.
            Importantly, in the presence of uniform losses, the loss channel can be commuted to the input state, and the output probabilities can be written as
            \begin{multline}
                p_{\Lambda_\eta(\ketbra{\bm{n}_0}), U}(\bm{n})
                \\=
                \Tr\!\left[
                    \varphi(U)\, \Lambda_\eta(\ketbra{\bm{n}_0}) \,\varphi(U)^\dagger
                    \ketbra{\bm{n}}
                \right]
            \end{multline}
            Suppose that we restrict attention to outcomes with total photon number \(\ell\), which we will signal in the superscript of the reference value as $\mathrm{LXE}_{\text{ref}}^{(\ell)}$ in this discussion.
            We can write the reference value as
            \begin{multline}
                \mathrm{LXE}_{\text{ref}}^{(\ell)}(\Lambda_\eta(\ketbra{\bm{n}_0})) \\= \sum_{\bm{n} \in \mathcal{S}_{m,\ell}} 
                \underset{U \sim \mathrm{Haar}(m)}{\mathbb{E}}
                \left[
                     p_{\Lambda_\eta(\ketbra{\bm{n}_0}), U}(\bm{n} \,|\, |\bm{n}| = \ell)^2
                \right].
            \end{multline}
            We know that $\Lambda_\eta(\ketbra{\bm{n}_0})_{(\ell)}$ (i.e., its restriction to $\ell$ particles) is independent of the transmissivity $\eta$ and can be written as a uniform mixture
            \begin{equation}
                \Lambda_\eta(\ketbra{\bm{n}_0})_{(\ell)}
                = \Tr_{n-\ell} (\ketbra{\bm{n}_0})
                .
            \end{equation}
            Since this state is supported on the \(\ell\)-particle sector, we may apply \Cref{eq:bosonic_swap_to_trace} with \(n\) replaced by \(\ell\). Thus,
            \begin{multline}
                \Tr\!\left[\mathbb S_q^{(\ell)}\Tr_{n-\ell} (\ketbra{\bm{n}_0})^{\otimes 2}\right]
                \\=
                \Tr\!\left[\Tr_{\ell-q}(\Tr_{n-\ell} (\ketbra{\bm{n}_0}))^2\right]
                \\= \Tr\!\left[\Tr_{n-q}(\ketbra{\bm{n_0}})^2\right]
            \end{multline}
            where $\mathbb S_q^{(\ell)}$ is the bosonic swap operator acting on two copies of the $\ell$-particle subspace.
            Therefore,
            \begin{equation}
                \Tr\!\left[\mathbb S_q^{(\ell)}\,\Lambda_\eta(\ketbra{\bm n_0})_{(\ell)}^{\otimes 2}\right]
                =
                \frac{1}{\binom{n}{q}}.
            \end{equation}
            As expected, this formula independent of both \(\ell\) and \(\eta\).

            Using this formula, we can write that
            \begin{equation}
                \Tr\!\left[\mathbb{P}^{(\ell)}_k\, \Lambda_\eta(\ketbra{\bm n_0})_{(\ell)}^{\otimes 2} \right] 
                =
                \sum_{q=0}^\ell \frac{ c_{k,q}^{(\ell)}}{\binom{n}{q}},
            \end{equation}
            where \(\mathbb{P}^{(\ell)}_k\) denotes the projector onto the subspace
            \(\mathbb{S}_{(2\ell-k,\,k)}(\mathbb{C}^m)\) and $c_{k,q}^{(\ell)}$ are the corresponding bosonic swap coefficients according to \Cref{lem:explicit_P_k}.
            Collecting every term, we can write the final formula as
            \begin{multline}
                \mathrm{LXE}_{\mathrm{ref}}^{(\ell)}
                \!\left(\Lambda_\eta(\ketbra{\bm{n}_0})\right)
                \\
                =
                |\mathcal{S}_{m,\ell}|
                \sum_{r = 0}^{\lfloor \ell/2 \rfloor}
                \sum_{q=0}^\ell
                \frac{
                   c_{2r,q}^{(\ell)}
                }{
                    \binom{n}{q}\operatorname{Tr}\mathbb{P}^{(\ell)}_{2r}
                }
                \operatorname{Tr}\!\left[
                    \mathbb{P}^{(\ell)}_{2r}
                    \mathbb{D}_{m,\ell}
                \right],
            \end{multline}
            where, following the discussion of \Cref{sec:framework}, each term in the sum can be efficiently evaluated analytically.
            Moreover, we can give an explicit formula in the form of

            \begin{equation}
                \mathrm{LXE}_{\mathrm{ref}}^{(\ell)}
                \!\left(\Lambda_\eta(\ketbra{\bm{n}_0})\right)
                =
                \sum_{r = 0}^{\lfloor \ell/2 \rfloor}
                \sum_{q=0}^{\ell}
                \frac{
                    c_{2r,q}^{(\ell)}
                }{
                    \binom{n}{q}
                }
                \frac{
                    (\ell-r)!\left(\frac{1}{2}\right)_r
                }{
                    \left(\frac{m}{2}\right)_r
                    \left(\frac{m+1}{2}\right)_{\ell-r}
                }.
            \end{equation}
            We note that in the special case \(\ell=n\), the
            reference value reduces to the ideal, lossless expression,
            \begin{equation}
                \mathrm{LXE}_{\mathrm{ref}}^{(n)}
                \!\left(\Lambda_\eta(\ketbra{\bm{n}_0})\right)
                =
                \mathrm{LXE}_{\mathrm{ref}}^{(n)}(\ketbra{\bm{n}_0}).
            \end{equation}

\subsection{Scattershot Boson Sampling}\label{ssec:lxeb_scattershot_boson_sampling}
    
    Scattershot Boson Sampling~\cite{Bentivegna_2015} and Twofold Scattershot Boson Sampling~\cite{Chakhmakhchyan_2017} give rise to output distributions whose probabilities are expressed in terms of Permanents of submatrices of the same interferometer unitary as in standard Fock-state Boson Sampling. Since the reference values depend only on second moments of these probabilities and are invariant under the random choice of occupied input modes, the corresponding reference values can be computed efficiently in these settings as well.
    
    More concretely, in Scattershot Boson Sampling, one considers a product of \(1 \leq d \leq m\) two-mode squeezed states with a uniform squeezing parameter \(r\) and $(m-d)$ pairs of vacuum modes as initial state.
    One mode from each pair is injected into a linear interferometer, while the other is immediately measured using photon-number-resolving detectors. The joint probability of observing output pattern \(\bm{n}\) in the interferometer modes and heralding pattern \(\bm{m}\) in the measured modes is
    \begin{equation}\label{eq:scattershot_joint_prob_dist}
        p_U(\bm{n}, \bm{m}) =
        \frac{\tanh(r)^{2|\bm{n}|}}{\cosh(r)^{2d}}
        \frac{\left|\mathrm{per}(U_{\bm{n},\bm{m}})\right|^2}{\bm{n}!\,\bm{m}!},
    \end{equation}
    where \(\bm{n}\) and \(\bm{m}\) denote the photon-number patterns in the interferometer and heralding modes, respectively. Here, $U_{\bm{n}, \bm{m}}$ corresponds to the matrix formed from $U$ by repeating its rows according to $\bm{n}$ and its columns according to $\bm{m}$.
    
    This distribution admits a natural time-unfolded interpretation: conditioned on the heralding outcome \(\bm{m}\), the interferometer effectively receives the Fock input state \(\ket{\bm{m}}\). For our framework, we also postselect for outcomes with $n$ particles in each leg. The corresponding joint output distribution is therefore
    \begin{equation}\label{eq:scattershot_cond_prob_dist}
        p_U(\bm{n}, \bm{m}) = \frac{1}{|\mathcal{S}_{d,n}|}
        \frac{\left|\mathrm{per}(U_{\bm{n},\bm{m}})\right|^2}{\bm{n}!\,\bm{m}!}.
    \end{equation}
    Accordingly, the LXEB reference value for Scattershot Boson Sampling is obtained by summing \(\mathrm{LXE}_{\mathrm{ref}}^{(n)}(\ketbra{\bm{m}})\) over all heralding patterns:
    \begin{equation}\label{eq:lxe_ref_sbs_def}
        \mathrm{LXE}_{\mathrm{ref,SBS}}(m, n, d)
        \coloneqq
        \frac{1}{|\mathcal{S}_{d,n}|^2}
        \sum_{\bm{m}\in\mathcal{S}_{d,n}}
        \mathrm{LXE}_{\mathrm{ref}}^{(n)}(\ketbra{\bm{m}}).
    \end{equation}
    Here, we sum over $\mathcal{S}_{d,n}$, because the initial state on the second leg (which is immediately measured) is supported on $d$ modes.
    Using \Cref{eq:lxe_ref_boxed}, this can be rewritten as
    \begin{equation}\label{eq:lxe_ref_sbs_sum}
        \mathrm{LXE}_{\mathrm{ref,SBS}}(m,n, d)
        \!=\!
        \frac{|\mathcal{S}_{m, n}|}{|\mathcal{S}_{d, n}|}
        \sum_{r = 0}^{\lfloor n / 2 \rfloor}
        \frac{\Tr\!\left[\mathbb{P}_{2r}\mathbb{D}_{m,n}\right] \Tr\!\left[\mathbb{P}_{2r}\mathbb{D}_{d,n}\right]}{\Tr \mathbb{P}_{2r}}.
    \end{equation}
    This expression is independent of the squeezing parameter \(r\), and depends only on the number of modes \(m\), the number of photons \(n\) under consideration, and the number of input two-mode squeezed states $d$.
    As discussed in \Cref{sec:framework}, this expression can be evaluated efficiently, but we can also give a further simplification as follows:
    \begin{proposition}\label{prop:scattershot_final_form}
        We can write the reference value for Scattershot Boson Sampling
        in the form
        \begin{multline}
            \mathrm{LXE}_{\mathrm{ref,SBS}}(m,n,d)
            \\=
            \frac{2^{2n}}{\binom{d+n-1}{n}^2}
            \sum_{r=0}^{\lfloor n/2\rfloor}
            \frac{2n-4r+1}{2n-2r+1}\,
            \frac{2^{-4r}\binom{2r}{r}}{\binom{2n-2r}{n-r}}\,
            \frac{ \left(\frac d2\right)_{n-r}\left(\frac{d-1}{2}\right)_{r}}
            {\left(\frac m2\right)_{r}\left(\frac{m+1}{2}\right)_{n-r}}.
        \end{multline}
    \end{proposition}

    \begin{proof}
        Starting from
        \begin{equation}
            \mathrm{LXE}_{\mathrm{ref,SBS}}(m,n,d)
            \!=\!
            \frac{|\mathcal{S}_{m,n}|}{|\mathcal{S}_{d,n}|}
            \!\sum_{r=0}^{\lfloor n/2\rfloor}
            \!\frac{\Tr\!\left[\mathbb{P}_{2r}\mathbb{D}_{m,n}\right]
                  \!\Tr\!\left[\mathbb{P}_{2r}\mathbb{D}_{d,n}\right]}
                 {\Tr \mathbb{P}_{2r}},
        \end{equation}
        we substitute \Cref{lem:tr_p2r_dmn}, \Cref{eq:dimensions}, and
        \(
        |\mathcal{S}_{x,n}|=\binom{x+n-1}{n}=\frac{(x)_n}{n!}.
        \)
        After using similar identities as for the proof of \Cref{prop:bs_lxeref_simplified}, and canceling the common factors, we get the desired formula.
    \end{proof}
    \noindent We emphasize, that---as for the case of Boson Sampling from \Cref{prop:bs_lxeref_simplified}---these formulas are not limited to any regime, and work beyond the dilute regime $m = \Omega(n^2)$.
    
    Finally, when $m = d$, Haar invariance implies that the same expressions remain valid for Twofold Scattershot Boson Sampling~\cite{Chakhmakhchyan_2017}, where the second leg is also transformed by a linear interferometer.

    \subsection{Gaussian Boson Sampling}\label{sec:lxeb_gaussan_boson_sampling}
        Our framework naturally extends to Gaussian Boson Sampling. Unlike standard Boson Sampling with Fock-state inputs, Gaussian Boson Sampling has support on all total photon-number sectors, with sector weights determined by the squeezing parameters. To fit this setting into the framework of \Cref{sec:framework}, we fix a total photon number $n$ and consider the normalized $n$-particle component $\rho_{(n)}$ of the input density operator $\rho$. Combining \Cref{eq:swap_expval_algorithm_form} with the method of \Cref{app:proof_bosonic_swap_efficient} then yields an efficient algorithm for computing the bosonic swap expectation value on two copies, and hence the reference $\mathrm{LXE}$ value. The necessary matrix elements of single-mode squeezed states are given in Refs.~\cite{yuen_1976,Weedbrook_2012}.

        We can obtain a particularly compact expression for the reference value in the case of uniform squeezing. Specifically, consider an input state consisting of $d \leq m$ single-mode squeezed states with identical squeezing parameters and $(m-d)$ vacuum modes. We restrict attention to even total photon number, writing $n=2N$ where $N$ is the number of photon pairs, since odd total photon numbers are absent from the ideal Gaussian Boson Sampling distribution. 
        In this setup, we can get the following formula for the reference value:
        \begin{proposition}\label{prop:LXE_ref_uniform_squeezed_closed_binom}
            Let \(n=2N\), and let \(\rho\) be the normalized \(n\)-particle restriction
            of a product state consisting of \(1 \leq d \leq m\) single-mode squeezed states and
            \((m-d)\) vacuum modes. Then
            \begin{multline}
                \label{eq:LXE_ref_uniform_squeezed_closed_binom}
                \mathrm{LXE}_{\mathrm{ref}}^{(2N)}(\rho)
                =\\
                N!^2
                \!\sum_{j=0}^{N}\!
                \frac{4N-4j+1}{4N-2j+1}
                \!
                \frac{\binom{2j}{j}
                \binom{2N-2j}{N-j}^{2} \left(\frac{d-1}{2}\right)_j\left(\frac d2+N\right)_{N-j}}
                {\binom{4N-2j}{2N-j}\left(\frac d2\right)_N \left(\frac m2\right)_j\left(\frac{m+1}{2}\right)_{2N-j}}.
            \end{multline}
        \end{proposition}
        \noindent This formula directly follows from \Cref{thm:P_k_uniform_squeezed_closed} using similar techniques as for \Cref{prop:bs_lxeref_simplified}.
        As expected, this expression is independent of the squeezing parameter \(r\), and only depends on the number of pairs $N$, the number of modes $m$ and the number of squeezed input states $d$.

        \subsubsection*{Uniform losses}
            Similarly to the Boson Sampling setting discussed in \Cref{sssec:uniform_losses_bs}, uniform losses can be commuted onto the input state, so that the problem reduces to evaluating expectation values with respect to lossy single-mode squeezed vacuum states. The matrix elements of these states in the Fock basis can be computed efficiently, for instance using the \textsc{Piquasso} library~\cite{Kolarovszki_2025}. Moreover, since the resulting state remains a product of single-mode states, \Cref{prop:efficient_swap_product} applies directly in this setting. Consequently, the relevant bosonic-swap expectation values can also be computed efficiently for uniformly lossy Gaussian Boson Sampling, yielding an efficient algorithm for the corresponding LXEB reference value as well.

    \subsection{Numerical experiments}
        \label{sec:numerical_experiments}

        In this section, we complement our analytical results for the LXEB reference values with Monte Carlo simulations, with the aim of providing empirical evidence for the following phenomena that remain unproven in this work: the convergence of the estimator \(\widehat{\mathrm{LXE}}(p_{\rho,U},\mathcal{X})\) to \(\mathrm{LXE}(p_{\rho,U},p_{\rho,U})\) as the number of collected samples increases, and the concentration of the ideal LXEB score \(\mathrm{LXE}(p_{\rho,U},p_{\rho,U})\) around its reference value with high probability over Haar-random \(U\).
        Moreover, we investigate the impact of optical losses on the LXEB values, showing that losses reduce the LXEB score and that the gap between the lossless and lossy cases increases with system size.
        All simulations were performed using the \textsc{Piquasso} framework~\cite{Kolarovszki_2025}, which enables exact evaluation of Gaussian and Fock-space matrix elements up to a chosen cutoff.
        All numerical simulations were executed on an Intel Xeon E5-2650 processor platform.

        Before we delve into the numerical experiments, we define the \textit{linear cross-entropy benchmarking fidelity} as
        \begin{equation}\label{eq:lxeb_fidelity}
            \mathcal{F}_{\mathrm{LXE}}(p_{\rho, U}, q) \coloneqq \frac{
                | \mathcal{S}_{m,n} | \, \mathrm{LXE}(p_{\rho, U}, q) - 1
            }{
                | \mathcal{S}_{m,n} | \, \mathrm{LXE}_{\text{ref}}^{(n)}(\rho) - 1
            },
        \end{equation}
        which is analog to the version introduced for RCS in Ref.~\cite{Arute_2019}.
        We can easily verify that $\mathbb{E}_{U \sim \mathrm{Haar}(m)} [\mathcal{F}_{\mathrm{LXE}}(p_{\rho, U}, p_{\rho, U})] = 1$ and that $\mathbb{E}_{U \sim \mathrm{Haar}(m)}[\mathcal{F}_{\mathrm{LXE}}(p_\rho, u)] = 0$, for $u$ being the uniform distribution over $\mathcal{S}_{m, n}$. 
        This choice of fidelity is justified by the assumption that the experimental probability distribution $q$ can be modeled as a mixture of the ideal distribution $p_{\rho, U}$ and the uniform distribution $u$, modeling the random noise. In this sense, the LXEB fidelity determines the amount of mixing between the ideal $\rho$ and the maximally mixed state over $n$ particles.
    
        We begin our numerical study by analyzing the convergence of the LXEB fidelity estimator in standard Boson Sampling. We evaluate the LXEB reference value predicted by the formula in \Cref{prop:bs_lxeref_simplified}.
        As discussed in \Cref{sec:basics}, considering an interferometer modeled \(U \in \mathrm{U}(m)\), the fixed-interferometer estimate is given by
        $
            \widehat{\mathrm{LXE}}(p_U, \mathcal{X}) = (1/s)\sum_{i=1}^s p_U(\bm{n}_i)
        $,
        where the outcomes \(\mathcal{X} = (\bm{n}_1, \dots, \bm{n}_s)\) are sampled from the Boson Sampling distribution denoted by \(p_U\). The corresponding fidelity estimator is simply just
        \begin{equation}\label{eq:lxeb_fidelity_estimator}
            \widehat{\mathcal{F}}_{\mathrm{LXE}}(p_{\rho,U}, \mathcal{X}) =
            \frac{
                |\mathcal{S}_{m,n}|\, \widehat{\mathrm{LXE}}(p_U, \mathcal{X}) - 1
            }{
                |\mathcal{S}_{m,n}|\, \mathrm{LXE}_{\mathrm{ref}}^{(n)}(\rho) - 1
            }.
        \end{equation}
        To illustrate numerically the self-averaging of the fidelity over \(U \sim \mathrm{Haar}(m)\), we generated \(N=10\) independent Haar-random interferometers. For each interferometer, we computed the exact output probabilities corresponding to \(s=1000\) samples. The resulting estimates are shown in \Cref{fig:lxeb_experiment}(a).
        In all tested instances, the fixed-interferometer estimates lie within \(2\sigma\) of the corresponding analytical predictions, where \(\sigma\) denotes the standard error across the ensemble of interferometers:
        \begin{equation}
            \sigma =
            \frac{
                |\mathcal{S}_{m,n}|
            }{
                |\mathcal{S}_{m,n}|\, \mathrm{LXE}_{\mathrm{ref}}^{(n)}(\rho) - 1
            }
            \frac{s_X}{\sqrt{N}},
        \end{equation}
        with \(s_X\) being the standard deviation of \(\widehat{\mathrm{LXE}}(p_U, \mathcal{X})\) across interferometers.
        Numerically, we observe that the variance of the estimator decreases slightly with increasing system size, suggesting that the LXEB fidelity becomes increasingly concentrated over \(U \sim \mathrm{Haar}(m)\) as the system size grows.

        \begin{figure*}[t]
            \centering
            \begin{minipage}[t]{0.7\linewidth}
                \centering
                \includegraphics[width=\linewidth]{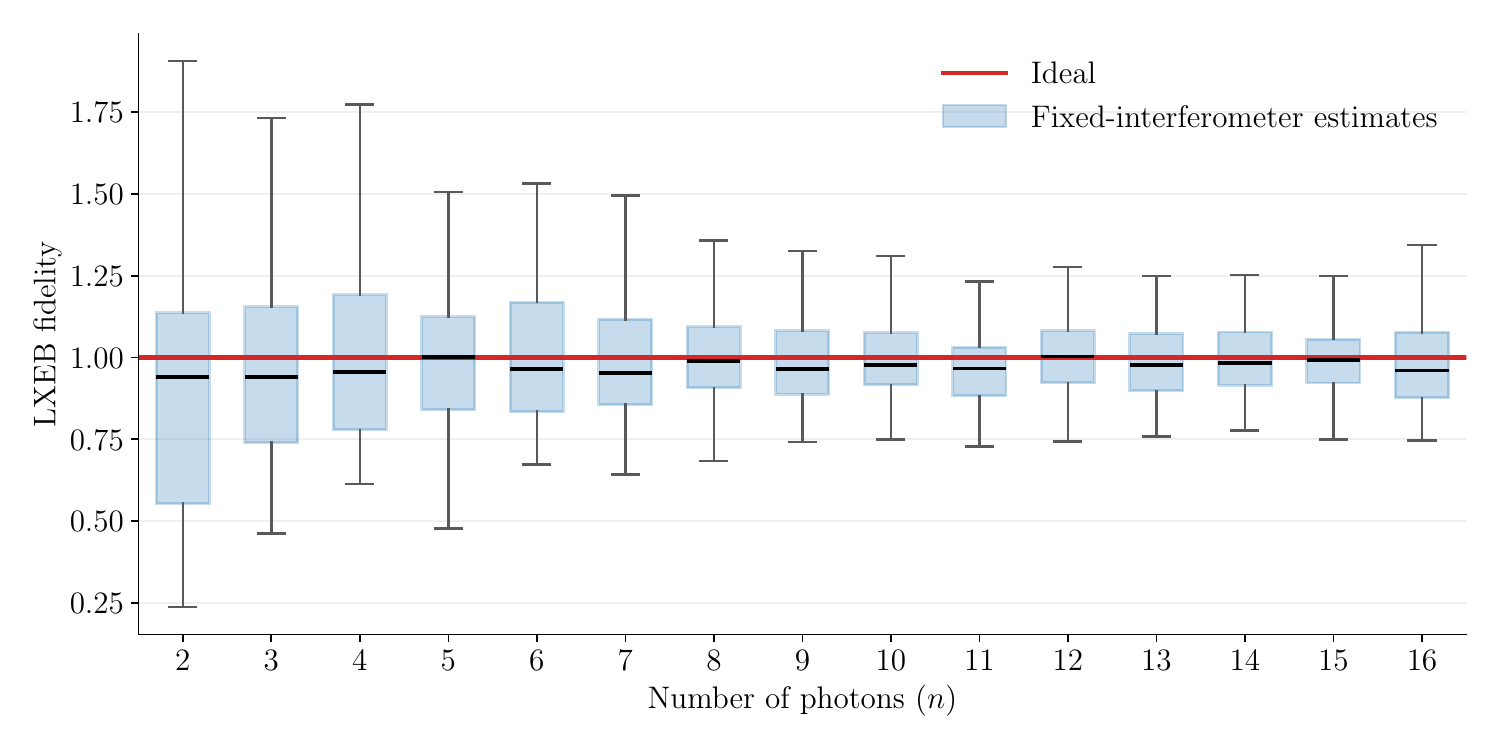}
                
                \vspace{0.5em}
                \footnotesize\textbf{(a)} Boson Sampling with particle-number-resolving detectors with $m = 2n$; see \Cref{sec:basics} for details.
            \end{minipage}
            
            \vspace{1em}
            
            \begin{minipage}[t]{0.7\linewidth}
                \centering
                \includegraphics[width=\linewidth]{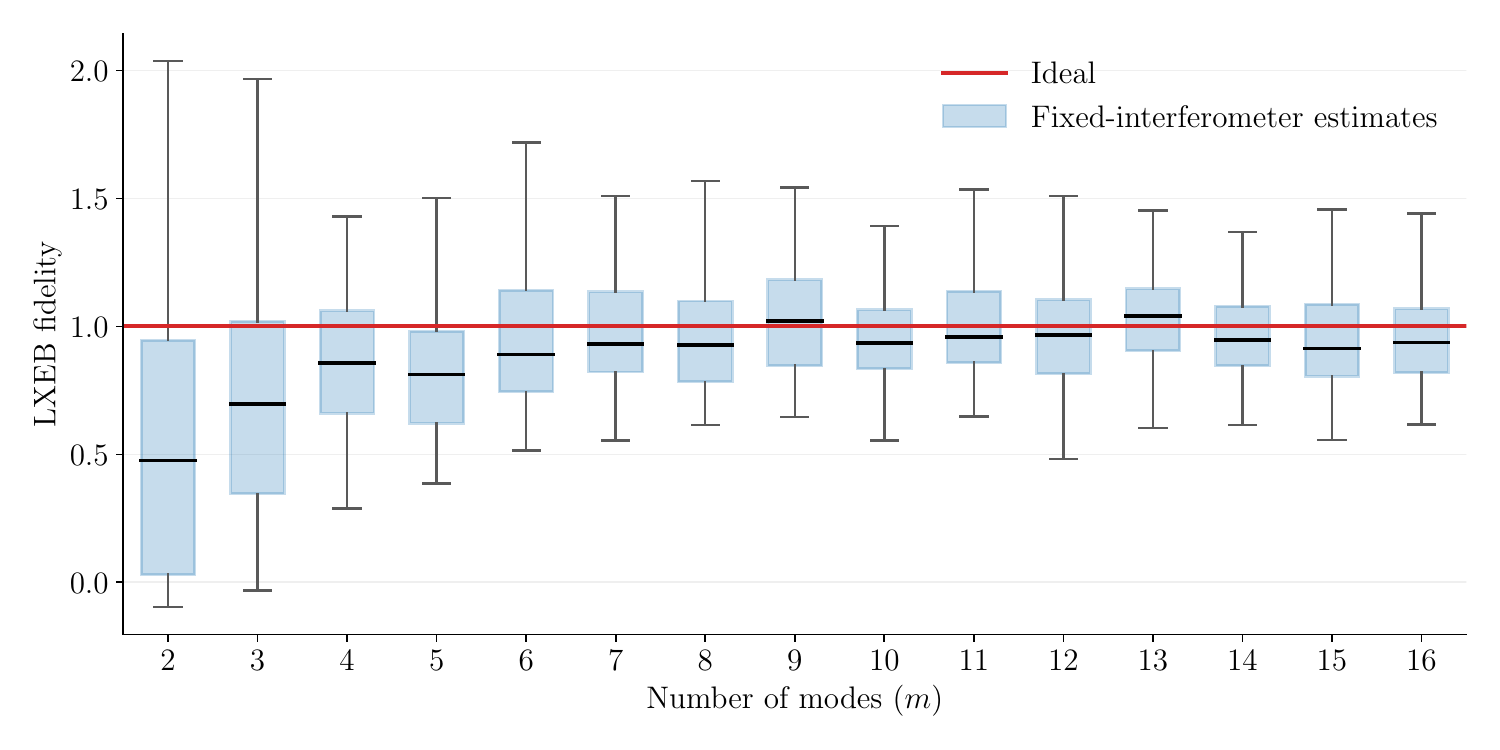}
                
                \vspace{0.5em}
                \footnotesize\textbf{(b)} Gaussian Boson Sampling with uniformly squeezed input states; see \Cref{sec:basics} for details. The squeezing parameter is fixed to \(r=\operatorname{arcsinh}(1)\), so that the mean photon number per mode equals \(1\), and hence the total mean photon number equals the number of modes.
            \end{minipage}
            
            \caption{Comparison of analytical and estimated LXEB reference values. For each value of \(m\), the box plot summarizes the distribution of LXEB fidelity estimates obtained from individual interferometers (see \Cref{eq:lxeb_fidelity_estimator}). For each interferometer, we estimated the LXEB fidelity from \( s=1000\) samples and averaged the results over \(N=100\) independently drawn Haar-random interferometers. The central horizontal line inside each box indicates the median, while the lower and upper edges of the box mark the first and third quartiles, respectively. The whiskers extend to the range of non-outlying values. The red line shows the corresponding ground-truth fidelity. }
            \label{fig:lxeb_experiment}
        \end{figure*}

        We performed an analogous experiment for Gaussian Boson Sampling with uniformly squeezed input states.
        The squeezing parameter $r$ was chosen so that the expected total number of photons equals the number of modes.
        Since the total number of photons is not fixed in Gaussian Boson Sampling, all statistics were conditioned for a total number of $n$ particles.
        Due to the even-parity constraint of squeezed-vacuum states, when the number of modes $m$ was odd, we imposed the postselection condition $n=m-1$ instead of $n = m$.
        An additional requirement in Gaussian Boson Sampling, compared with Boson Sampling, is the evaluation of conditional output probabilities, which are obtained by normalizing the detection probabilities by the total probability of observing the postselected particle number. This normalization does not introduce any additional computational difficulty: for uniformly squeezed states, the required factor is available in closed form, and for arbitrary product states it can be computed using \Cref{prop:polynomial_machinery} from \Cref{app:proof_bosonic_swap_efficient}.
        The estimated LXEB fidelities are compared with the analytical reference values obtained from \Cref{sec:framework}.
        As illustrated in \Cref{fig:lxeb_experiment}(b), the experimental estimates agree with the theoretical predictions. As in the Boson Sampling case, the variance of the estimator exhibits a decreasing trend with increasing system size.
        Finally, to illustrate the effect of optical loss on the LXEB benchmark in Boson Sampling and Gaussian Boson Sampling. In the case of Boson Sampling with uniform losses, the number of surviving
        photons is itself a random variable. If the input contains \(n\) photons and the
        loss channel has transmissivity \(\eta\), then the probability that exactly
        \(\ell\) photons survive is
        $
            p_{\eta}(n,\ell)
            \coloneqq
            \binom{n}{\ell}
            \eta^\ell
            (1-\eta)^{n-\ell}.
        $
        For each possible value of \(\ell\), we evaluate the LXEB fidelity using the
        corresponding \(\ell\)-photon output space \(\mathcal{S}_{m,\ell}\). We then
        average these values with respect to the probability distribution
        \(p_{\eta}(n,\ell)\). Thus, the plotted quantity describes the LXEB fidelity
        averaged over the possible numbers of photons that remain after the loss
        channel.
        In \Cref{fig:lxeb_fidelity_in_lossy_sampling}(a), we plot the case of saturated regime Boson Sampling with scaling
        \(m= 2 n\) and evaluate this expression for several fixed values of \(\eta\).
        In the case of Gaussian Boson Sampling, we plot in \Cref{fig:lxeb_fidelity_in_lossy_sampling}(b) the exact LXEB fidelity as a function of the number of modes $m$ for several values of the loss parameter, in the regime $n=m$. As shown in the figure, the LXEB fidelity decreases monotonically with increasing loss, and this degradation becomes more pronounced as the system size grows.
        Overall, these numerical experiments provide a quantitative illustration of the sensitivity of the LXEB benchmark to loss in a Boson Sampling and Gaussian Boson Sampling experiments.

        \begin{figure*}[t]
            \centering
            \begin{minipage}[t]{0.7\linewidth}
                \centering
                \includegraphics[width=\linewidth]{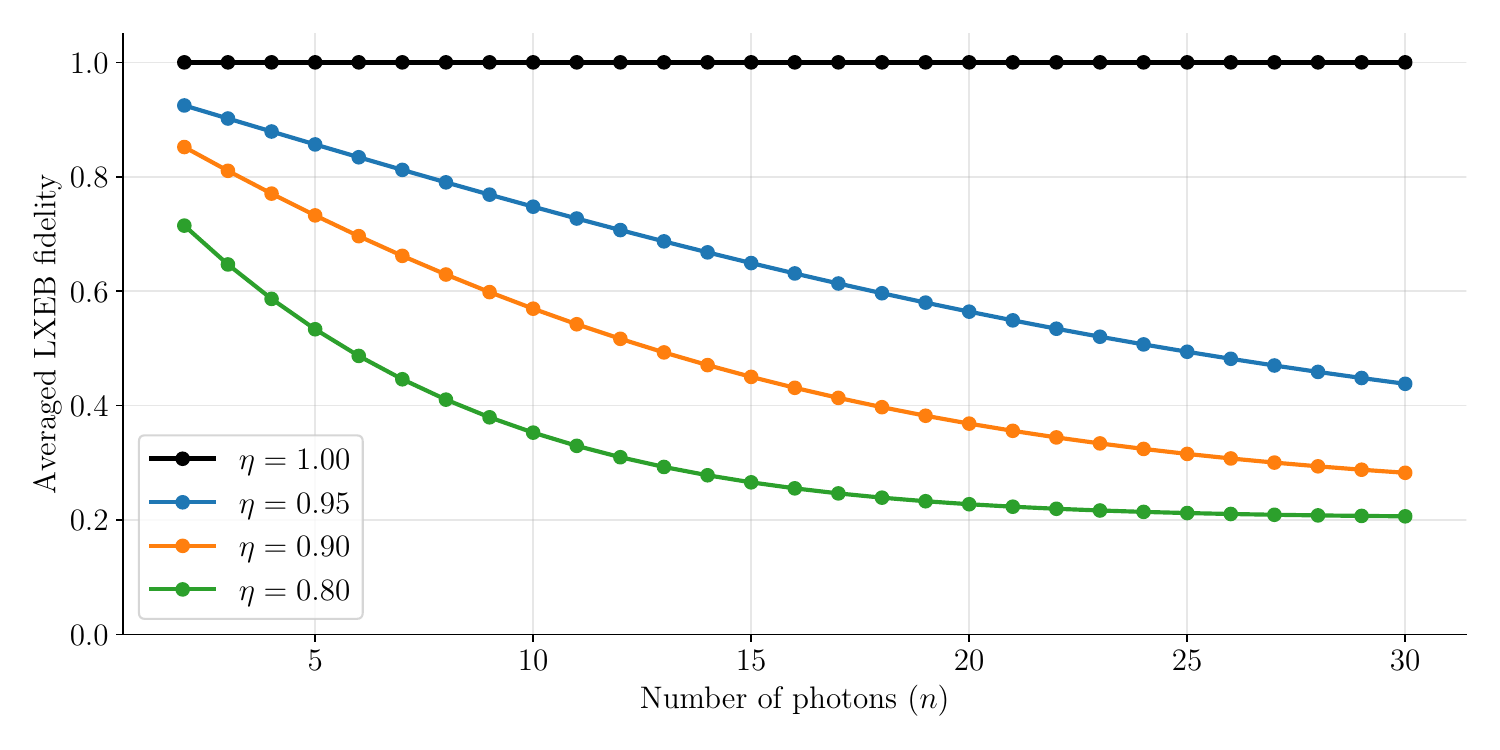}
                
                \vspace{0.5em}
                \footnotesize\textbf{(a)} Lossy Boson Sampling in the saturated regime \(m=2n\).
            \end{minipage}
            
            \vspace{1em}
            
            \begin{minipage}[t]{0.7\linewidth}
                \centering
                \includegraphics[width=\linewidth]{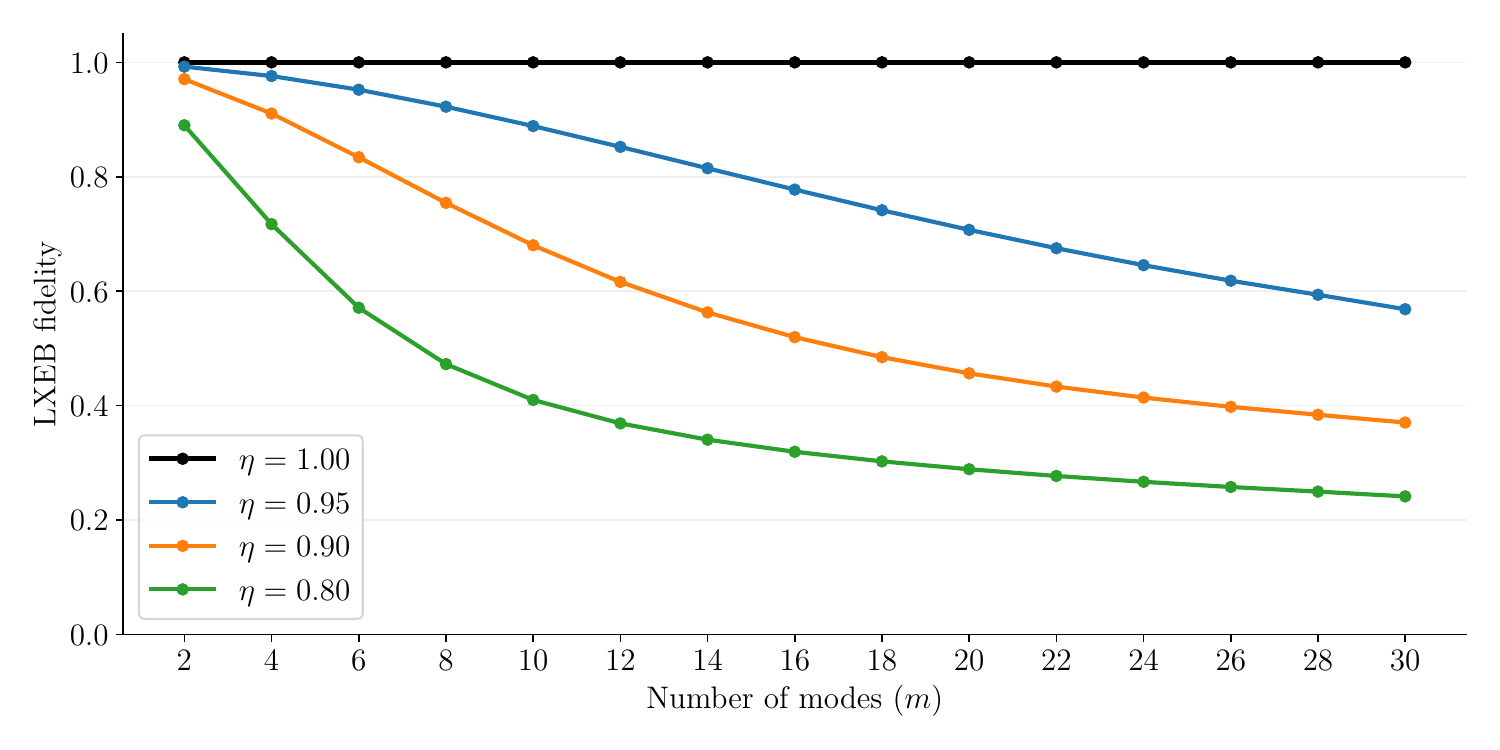}
                
                \vspace{0.5em}
                \footnotesize\textbf{(b)} Lossy Gaussian Boson Sampling in the regime \(n=m\), with squeezing parameter fixed to \(r=\operatorname{arcsinh}(1)\).
            \end{minipage}
            
            \caption{LXEB fidelity in the presence of optical loss. 
            In both panels, the LXEB fidelity from \Cref{eq:lxeb_fidelity} is shown for several values of the transmissivity \(\eta\) (see \Cref{eq:loss_channel}). 
            Panel~\textbf{(a)} shows the loss-averaged LXEB fidelity for Boson Sampling as a function of the input photon number \(n\), in the saturated regime \(m=2n\). 
            For each value of \(n\), the LXEB reference values are evaluated for all possible numbers \(\ell\) of surviving photons using the corresponding \(\ell\)-photon output space \(\mathcal{S}_{m,\ell}\), averaged with respect to the binomial distribution of surviving photons, and normalized with respect to the ideal lossless \(n\)-photon reference value. 
            Panel~\textbf{(b)} shows the LXEB fidelity for Gaussian Boson Sampling as a function of the number of modes \(m\), in the regime \(n=m\). 
            The squeezing parameter is fixed to \(r=\operatorname{arcsinh}(1)\), so that the mean photon number per mode equals \(1\), and hence the total mean photon number equals the number of modes. 
            In this case, each curve is obtained by evaluating the lossy LXEB reference value and normalizing it with respect to the ideal lossless case; the density matrices of the corresponding single-mode lossy squeezed states are computed numerically using \textsc{Piquasso}~\cite{Kolarovszki_2025}. 
            In both settings, optical loss suppresses the LXEB fidelity, with a more pronounced separation between the lossless and lossy curves as the system size increases.}
            \label{fig:lxeb_fidelity_in_lossy_sampling}
        \end{figure*}

\section{Anticoncentration from second moments}\label{sec:anticoncentration}

In this part we connect LXEB scores to anticoncentration results based on second moments of output probabilities $\mathbb{E}_{U\sim\mathrm{Haar}(m)}[p_{\rho,U}(\bm{n})^2]$. To establish these, we use Paley-Zygmund inequality (a standard tool used in this context in the literature \cite{Bremner_2017,dalzell_2022,Oszmaniec_2022}) adjusted accordingly to accommodate the fact that in outside of the dilute regime we do not have hiding and therefore the we cannot reduce our analysis to computation of moments of a single outcome probability $p_{\rho,U}(\bm{n}_0)$. 
Our key finding is the meaningful anticoncentration statement for Boson Sampling in the saturated regime ($m=\Theta(n)$); phrased in \Cref{thm:ac}. This strengthens Stockmeyer reduction by level of approximation of output probability from additive to multiplicative (cf. \Cref{thm:stockmeyer_average_anticoncentration}), improving hardness guarantees of Boson Sampling in the saturated regime \cite{bouland2025complexitytheoreticfoundationsbosonsamplinglinear}. Interestingly, analogous analysis proves too weak to establish meaningful anticoncentration for Scattershot Boson Sampling and Gaussian Boson Sampling, and only weaker forms of this property can be obtained; see \Cref{prop:scatter_asymptotic_bounds,prop:gaussian_asymptotic_bounds}.

    As discussed earlier in the paper, the anticoncentration of the probability distribution $p_{\rho,U}$ captures its flatness, i.e., the degree to which it is not concentrated on a small fraction of of outputs.  In what follows, we will be concerned with the following definition of anticoncentration which capture average-case case behavior of $p_{\rho,U}$ (for a state $\rho$ with definite photon number $n$) both over $U\sim\mathrm{Haar}(m)$ and $\bm{n}\sim \mathcal{U}_{m,n}$.

    \begin{definition}[Average-case anticoncentration]\label{def:anticoncentration}
        Let $\rho\in \mathcal{D}(\Sym^n(\mathbb{C}^m))$. We say that $\rho$  exhibits anticoncentration\footnote{Technically speaking, we should be defining anticoncentration as the property of a \emph{sequence} of quantum states $\{\rho_k\}_{k=1}^\infty$ and such that $\rho_k \in \mathcal{D}(\Sym^k(\mathbb{C}^{m_k}))$ (for $(m_k)_{k=1}^\infty$ being a sequence of number of modes), such that \Cref{eq:anticoncentrationDEF} is satisfied for universal constant $C>0$. We decide to follow this simplified definition to avoid unnecessary clutter.} if there exists a constant \(C>0\)  (independent on $n$ and $m$)  such that, for every \(\tau\in[0,1]\) and every \(n\ge 1\) ,
        \begin{equation}\label{eq:anticoncentrationDEF}
            \underset{\substack{U \sim \mathrm{Haar}(m) \\ \bm{n}\sim \mathcal{U}_{m,n}} }{\mathrm{Pr}}
            \left(
            p_{\rho_n,U}(\bm{n})\ge \frac{\tau}{|\mathcal{S}_{m,n}|}
            \right)
            \ge \frac{(1-\tau)^2}{C}.
        \end{equation}
    \end{definition}

    Note that this definition is weaker than the one appearing in the original Permanent anticoncentration conjecture from \cite{aaronson:2010}. This notion is however sufficient for our purposes, as will become apparent form proofs of \Cref{thm:ac} and \Cref{thm:stockmeyer_average_anticoncentration}) below. 

   Recall that from the irreducibility of $\varphi_n$, we have that $\mathbb{E}_{U\sim\mathrm{Haar}(m)} [p_{\rho_n,U}(\bm{n})]= 1 / |\mathcal{S}_{m,n}|$. This observation and the form of  \eqref{eq:anticoncentrationDEF} suggests the use Paley-Zygmund inequality which states that nonnegative random variable \(X\) and any \(\tau\in[0,1]\)
    \begin{equation}
        \mathbb{P}\!\left(X\ge \tau\,\mathbb{E}[X]\right)
        \ge
        (1-\tau)^2\frac{\mathbb{E}[X]^2}{\mathbb{E}[X^2]}.
    \end{equation}
Thus, intuitively, in order to prove anticoncentration, it is enough to control the ratio \(\mathbb{E}[X^2]/\mathbb{E}[X]^2\). This motivates to define, for a state $\rho\in\mathcal{D}(\Sym^n(\mathbb{C}^m))$ and fixed output $\bm{n}\in\mathcal{S}_{m,n}$, ``anticoncentration score''
  \begin{align}\begin{split}
            \mathrm{AC}_{\rho}(\bm{n})
            &\coloneqq
            \frac{
            \underset{U\sim \mathrm{Haar}(m)}{\mathbb{E}}\!\left[p_{\rho,U}(\bm{n})^2\right]
            }{
            \underset{U\sim \mathrm{Haar}(m)}{\mathbb{E}}\!\left[p_{\rho,U}(\bm{n})\right]^2  
            }
            \\&=  |\mathcal{S}_{m,n}|^2 \underset{U\sim \mathrm{Haar}(m)}{\mathbb{E}}\!\left[p_{\rho,U}(\bm{n})^2\right] ,
     \end{split}\end{align}
     where in the second equality we used, again, irreducibility of  $\varphi_n(U)$. We further define the average anticoncentration score by
        \begin{equation}\label{eq:ACaverage}
            \mathrm{AC}_\rho^{(n)} 
            \coloneqq 
           \underset{\bm{n}\sim \mathcal{U}_{m,n}}{\mathbb{E}} \!\left[\mathrm{AC}_\rho(\bm{n})\right] = |\mathcal{S}_{m,n}| \mathrm{LXE}_{\text{ref}}^{(n)}(\rho)\ ,
        \end{equation}
    with the second equality stemming form the definition of $\mathrm{LXE}_{\text{ref}}^{(n)}(\rho)$. The following Lemma shows that average AC score can put bounds on average-case anticoncentration.
        
    \begin{lemma}[Bounds for anticoncentration from average AC score]
    \label{lem:anticoncFROMacSCORE}
       Let $\rho\in\mathcal{D}(\Sym^n(\mathbb{C}^m))$. Then, for every \(\tau\in[0,1]\) and every \(n\ge 1\), we have
        \begin{equation}
            \underset{\substack{U\sim\mathrm{Haar}(m)\\ \bm n\sim \mathcal U_{m,n}}}{\Pr}
            \left(
                p_{\rho,U}(\bm n)\ge \frac{\tau}{|\mathcal S_{m,n}|}
            \right)
            \ge
            \frac{(1-\tau)^2}{\mathrm{AC}_\rho^{(n)}}.
        \end{equation}
     \end{lemma}

     \begin{proof}
        Applying the Paley-Zygmund inequality to the nonnegative random variable \(p_{\rho,U}(\bm{n})\), and recalling $\underset{U\sim\mathrm{Haar}(m)}{\mathbb{E}}[p_{\rho,U}(\bm{n})]= 1/|\mathcal{S}_{m,n}|$ we obtain
        \begin{equation}
            \underset{U \sim \mathrm{Haar}(m) }{\mathrm{Pr}}
            \left(
            p_{\rho,U}(\bm{n})\ge \frac{\tau}{|\mathcal{S}_{m,n}|}
            \right) \geq
            \frac{(1-\tau)^2}{\mathrm{AC}_\rho(\bm{n})}.
        \end{equation}
        Taking the expectation value over \(\bm{n}\sim \mathcal{U}_{m,n}\) yields
        \begin{multline}
            \underset{\bm{n}\sim \mathcal{U}_{m,n}}{\mathbb{E}}
            \left[
            \underset{U \sim \mathrm{Haar}(m) }{\mathrm{Pr}}
            \left(
            p_{\rho,U}(\bm{n})\ge \frac{\tau}{|\mathcal{S}_{m,n}|}
            \right)
            \right]
            \\ \ge
            (1-\tau)^2
            \underset{\bm{n}\sim \mathcal{U}_{m,n}}{\mathbb{E}}
            \left[
            \frac{1}{\mathrm{AC}_\rho(\bm{n})}
            \right].
        \end{multline}
        On the other hand, denoting the anticoncentration event by
        \begin{equation}
            \mathcal{A}  = \left\{ (U, \bm{n}) : p_{\rho,U}(\bm{n})\ge \frac{\tau}{|\mathcal{S}_{m,n}|} \right\},
        \end{equation}
        and using the independence of $\bm{n} \sim \mathcal{U}_{m,n}$ and $U \sim \mathrm{Haar}(m)$ we can write
        \begin{align}\begin{split}
            \underset{\bm{n}\sim \mathcal{U}_{m,n}}{\mathbb{E}}
            \left[
            \underset{U \sim \mathrm{Haar}(m) }{\mathrm{Pr}}
            \left(
                \mathcal{A}
            \right)
            \right]
            &= 
            \underset{\bm{n}\sim \mathcal{U}_{m,n}}{\mathbb{E}}
            \left[
            \underset{U \sim \mathrm{Haar}(m) }{\mathbb{E}}
            \left[
                \mathbbm{1}_{\mathcal{A}}
            \right]
            \right]
            \\&=
            \underset{\substack{U \sim \mathrm{Haar}(m) \\ \bm{n}\sim \mathcal{U}_{m,n}} }{\mathbb{E}}
            [ \mathbbm{1}_{\mathcal{A}}]
            \\&=
            \underset{\substack{U \sim \mathrm{Haar}(m) \\ \bm{n}\sim \mathcal{U}_{m,n}} }{\mathrm{Pr}}
            \left(
            \mathcal{A}
            \right),
        \end{split}\end{align}
        where $\mathbbm{1}_{\mathcal{A}}$ denotes the indicator function corresponding to the event $\mathcal{A}$.
        Hence, we know that
        \begin{multline}
          \underset{\substack{U \sim \mathrm{Haar}(m) \\ \bm{n}\sim \mathcal{U}_{m,n}} }{\mathrm{Pr}}
        \left(
        p_U(\bm{n})\ge \frac{\tau}{|\mathcal{S}_{m,n}|}
        \right) \\ \geq 
        (1-\tau)^2
            \underset{\bm{n}\sim \mathcal{U}_{m,n}}{\mathbb{E}}
            \left[
            \frac{1}{\mathrm{AC}_\rho(\bm{n})}
            \right].
        \end{multline}
        Since \(x\mapsto 1/x\) is convex on \((0,\infty)\), we lower bound the right-hand side as
        \begin{equation}
            \underset{\bm{n}\sim \mathcal{U}_{m,n}}{\mathbb{E}}
            \left[
            \frac{1}{\mathrm{AC}_\rho(\bm{n})}
            \right]
            \ge
            \frac{1}{
            \underset{\bm{n}\sim \mathcal{U}_{m,n}}{\mathbb{E}}[\mathrm{AC}_\rho(\bm{n})]
            }
            =
            \frac{1}{\mathrm{AC}_\rho^{(n)}}.
        \end{equation}
    \end{proof}

The following theorem provides a strong bound on AC score for regular Boson Sampling in the saturated regime, which in the light of proceeding lemma guarantees average-case anticoncentration.

    \begin{theorem}[Anticoncentration in saturated Boson Sampling]
        \label{thm:ac}
        Let $\ketbra{\bm{n}_0}$ be a collision-free Fock state supporting $n$ photons in $m$ modes (cf. \eqref{eq:fockin}). Define $p_{U}(\bm{n})\coloneq p_{\ketbra{\bm{n}_0}{\bm{n}_0},U}(\bm n)$ and $\mathrm{AC}_{\ketbra{\bm{n}_0}{\bm{n}_0}} \coloneq\mathrm{AC}(m, n)$. We then have
        \begin{equation}\label{eq:ACaverage1}
            \mathrm{AC}(m, n)\le \const
        \end{equation} 
        As a result, for every \(\tau\in[0,1]\) and every \(n\ge 1\), we have
        \begin{equation}
            \underset{\substack{U\sim\mathrm{Haar}(m)\\ \bm n\sim \mathcal U_{m,n}}}{\Pr}
            \left(
                p_{U}(\bm n)\ge \frac{\tau}{|\mathcal S_{m,n}|}
            \right)
            \ge
            \frac{(1-\tau)^2}{\const}.
        \end{equation}
        In particular, when \(m=\Theta(n)\), this lower bound remains bounded away from zero uniformly in \(n\). Thus, Boson Sampling in the saturated regime \(m=\Theta(n)\) satisfies average anticoncentration according to Definition \ref{def:anticoncentration}.
    \end{theorem}

     \begin{figure*}
        \centering
        \includegraphics[width=0.8\linewidth]{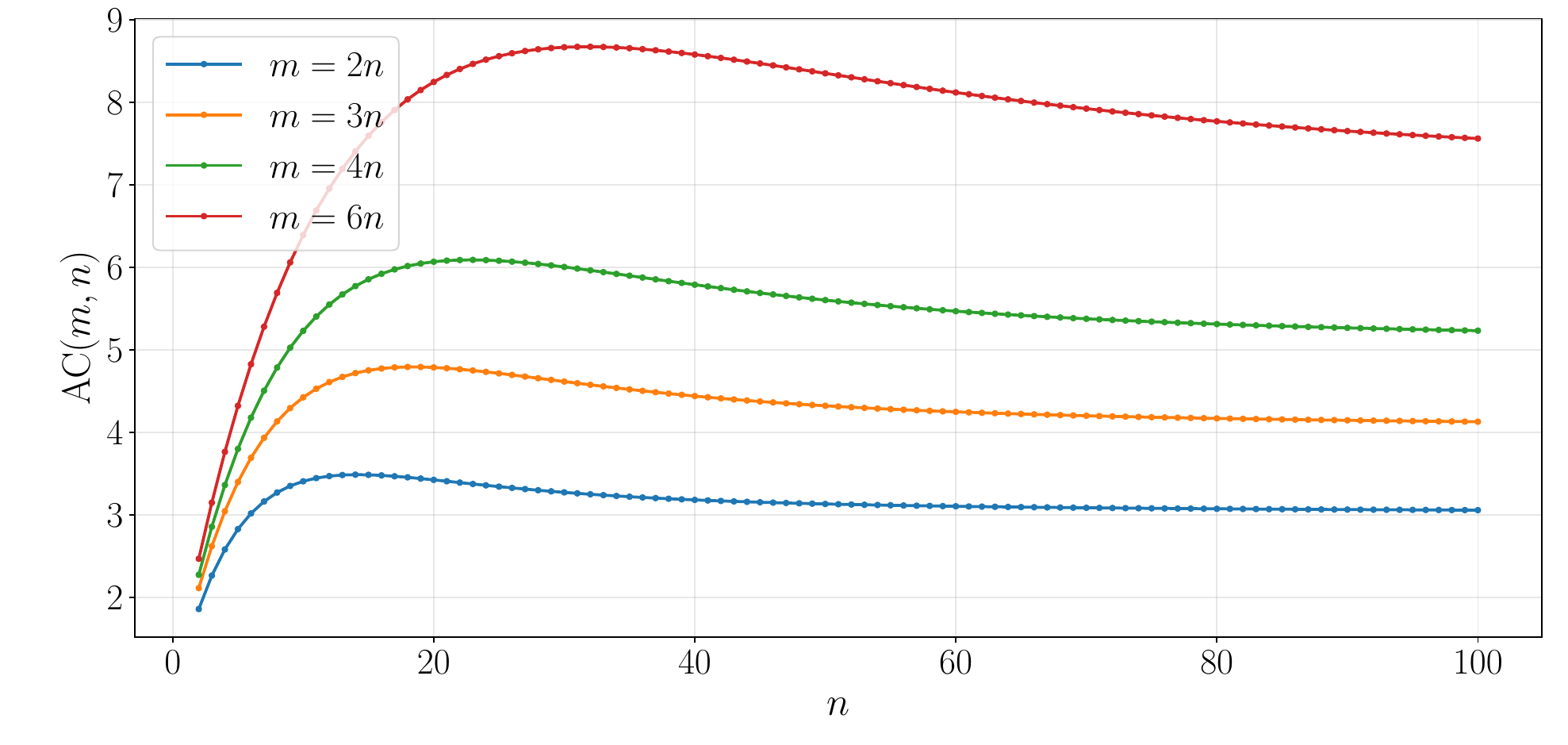}
        \caption{Numerical values of $\mathrm{AC}(m,n)$ in the saturated regime for $m=2n,3n,4n,6n$ and $2 \le n \le 100$. The data indicate that the general upper bound proven in \Cref{lem:ac} can be made tighter, and are consistent with the empirical estimate $\mathrm{AC}(m,n)\le 1+1.285\,\frac{m}{n}$.}
        \label{fig:ac_scores_saturated_bs}
    \end{figure*}

    We stress that the bound in \eqref{eq:ACaverage1} is unlikely to be optimal. The numerics suggests that as stronger estimate should hold (see \Cref{fig:ac_scores_saturated_bs} for illustration), 
    \(
    \mathrm{AC}(m,n) \le 1 + 1.285\,\frac{m}{n},
    \)
    In fact, it is possible to show
    \begin{equation}
        \lim_{n \to \infty} \mathrm{AC}(m,n) = 1 + \frac{m}{n},
    \end{equation}
    whenever the ratio $\frac{m}{n}$ converges. We provide a proof for this statement in \Cref{lemma:ac_asymptotic} from \Cref{app:ac_asymptotic}.

    This shows that the output distribution does not concentrate excessively near zero on average over occupation patterns, and provides the quantitative control needed in hardness argument that follows.     More concretely, this property can be used to turn an additive estimate into a multiplicative estimate in the  Stockmeyer reduction adapted to saturated regime Boson Sampling from Ref~\cite{bouland2025complexitytheoreticfoundationsbosonsamplinglinear}.

    Let us recall that weak classical sampler is a probabilistic polynomial-time classical algorithm which, on input an \(m\times m\) interferometer matrix \(U\) and an error tolerance \(\beta>0\), outputs samples from a distribution \(q_U\) such that $TV(p_U,q_U)\leq \beta$. The existence of such a sampler can be leveraged to approximately compute output probabilities. The key tool is Stockmeyer's approximate counting algorithm~\cite{stockmeyer}, which runs in \(\mathsf{BPP}^{\mathsf{NP}}\). In our setting, average anticoncentration can be used to turn an additive estimate into a multiplicative estimate in the Stockmeyer reduction adapted to saturated regime Boson Sampling from Ref~\cite{bouland2025complexitytheoreticfoundationsbosonsamplinglinear}.

    \begin{restatable}[Stockmeyer reduction using average anticoncentration]{theorem}{stockmeyer}
    \label{thm:stockmeyer_average_anticoncentration}
        Assume that \(m=\Theta(n)\), and suppose that the Boson Sampling problem admits a classical sampler \(q_U\) satisfying
        \begin{equation}
            \|p_U-q_U\|_{\mathrm{TV}}\le \beta
        \end{equation}
        for some $\beta \geq 0$ and every \(U\in \mathrm U(m) \). Then, for every \(\tau\in(0,1)\) and every \(k_1,k_2>0\), there exists a \(\mathsf{BPP}^{\mathsf{NP}}\) algorithm which, given \(U\sim \mathrm{Haar}(m)\) and \(\bm n \sim \mathcal U_{m,n}\), outputs an estimate \(\widetilde p_U(\bm n)\) such that
        \begin{equation}
            \frac{\left|\widetilde p_U(\bm n)-p_U(\bm n)\right|}{p_U(\bm n)}
            \leq
            \frac{2\beta k_1+\frac{k_2}{\mathrm{poly}(n)}}{\tau}
        \end{equation}
        with probability at least
        \begin{equation}
            \frac{(1-\tau)^2}{\const}
            -\frac{1}{k_1}
            -\frac{1}{k_2}.
        \end{equation}
    \end{restatable}
    \begin{proof}
        Let us start by fixing \(U\in \mathrm U(m)\). Since \(\|q_U-p_U\|_{\mathrm{TV}}\le \beta\), averaging over the uniform distribution on \(\mathcal S_{m,n}\) gives
        \begin{align}\begin{split}
            \mathbb E_{\bm n\sim \mathcal U_{m,n}}
            \left[
                |p_U(\bm n)-q_U(\bm n)|
            \right]
            \le
            \frac{2\beta}{|\mathcal S_{m,n}|}.
        \end{split}
        \end{align}
        Therefore, by Markov's inequality, for every \(k_1>0\),
        \begin{equation}
        \label{eq:stock_avg_1}
            \underset{\bm n\sim \mathcal U_{m,n}}{\mathrm{Pr}}
            \left(
                |p_U(\bm n)-q_U(\bm n)|
                \ge
                \frac{2\beta k_1}{|\mathcal S_{m,n}|}
            \right)
            \le \frac{1}{k_1}.
        \end{equation}
        Next, Stockmeyer's approximate counting algorithm~\cite{stockmeyer} implies that in \(\mathsf{BPP}^{\mathsf{NP}}\) one can estimate \(q_U(\bm n)\) by some \(\widetilde q_U(\bm n)\) satisfying
        \begin{equation}
        \label{eq:stockmeyer_rel}
            |\widetilde q_U(\bm n)-q_U(\bm n)|
            \le
            \frac{1}{\mathrm{poly}(n)}\, q_U(\bm n).
        \end{equation}
        Since
        $
            \mathbb E_{\bm n\sim \mathcal U_{m,n}}[q_U(\bm n)]
            =
            \frac{1}{|\mathcal S_{m,n}|},
        $
        another application of Markov's inequality yields, for every \(k_2>0\),
        \begin{equation}
        \label{eq:stock_avg_2}
            \underset{\bm n\sim \mathcal U_{m,n}}{\mathrm{Pr}}
            \left(
                q_U(\bm n)\ge \frac{k_2}{|\mathcal S_{m,n}|}
            \right)
            \le \frac{1}{k_2}.
        \end{equation}
        Combining \Cref{eq:stockmeyer_rel,eq:stock_avg_2}, we get
        \begin{equation}
            \underset{\bm n\sim \mathcal U_{m,n}}{\mathrm{Pr}}
            \left(
                |\widetilde q_U(\bm n)-q_U(\bm n)|
                \ge
                \frac{k_2}{\poly(n)\,|\mathcal S_{m,n}|}
            \right)
            \le \frac{1}{k_2}.
        \end{equation}
        Hence, by the union bound, for every fixed \(U\), we can write
        \begin{equation}
            \underset{\bm n\sim \mathcal U_{m,n}}{\mathrm{Pr}}
            \left(
                |p_U(\bm n)-\widetilde q_U(\bm n)|
                \ge
                \frac{2\beta k_1+\frac{k_2}{\poly(n)}}{|\mathcal S_{m,n}|}
            \right)
            \le
            \frac{1}{k_1}+\frac{1}{k_2}.
        \end{equation}
        Since this holds for every fixed \(U\), it also holds jointly over independent \(U\sim \mathrm{Haar}(m)\) and \(\bm n\sim \mathcal U_{m,n}\):
        \begin{equation}
        \label{eq:joint_additive}
            \underset{\substack{\bm n\sim \mathcal U_{m,n}\\ U\sim \mathrm{Haar}(m)}}{{\mathrm{Pr}}} \!\!\!
            \left(
                |p_U(\bm n)-\widetilde q_U(\bm n)|
                \ge
                \frac{2\beta k_1+\frac{k_2}{\mathrm{poly}(n)}}{|\mathcal S_{m,n}|}
            \right)
            \le
            \frac{1}{k_1}+\frac{1}{k_2}.
        \end{equation}

        Now we need to turn this additive estimate into a multiplicative one. By \Cref{thm:ac}, we know that in the $m = \Theta(n)$ regime, Boson Sampling exhibits average anticoncentration in the sense of \Cref{def:anticoncentration}, i.e., for every \(\tau\in[0,1]\),
        \begin{equation}
            \underset{\substack{\bm n\sim \mathcal U_{m,n}\\ U\sim \mathrm{Haar}(m)}}{\mathrm{Pr}}
            \left(
                p_U(\bm n)\ge \frac{\tau}{|\mathcal S_{m,n}|}
            \right)
            \ge \frac{(1-\tau)^2}{\const}.
        \end{equation}
        Let us denote the anticoncentration event by
        \begin{equation}
            \mathcal A
            \coloneqq
            \left\{
                (U,\bm n):
                p_U(\bm n)\ge \frac{\tau}{|\mathcal S_{m,n}|}
            \right\},
        \end{equation}
        and the event corresponding to unsuccessful approximation as
        \begin{equation}
            \mathcal{B} \coloneqq \left\{ (U, \bm{n} ) : |p_U(\bm n)-\widetilde q_U(\bm n)|
                \ge
                \frac{2\beta k_1+\frac{k_2}{\mathrm{poly}(n)}}{|\mathcal S_{m,n}|} \right\}
        \end{equation}
        The anticoncentration event $\mathcal{A}$ with a successful additive-approximation event $\overline{\mathcal{B}}$, we obtain
        \begin{equation}
            \underset{\substack{\bm n\sim \mathcal U_{m,n}\\ U\sim \mathrm{Haar}(m)}}{{\mathrm{Pr}}}
            \left[
                \; \mathcal{A} \; \cap\; \overline{\mathcal{B}} \;
            \right]
            \ge
            \frac{(1-\tau)^2}{\const}
            -\frac{1}{k_1}-\frac{1}{k_2}.
        \end{equation}
        On this event, we get an approximation in the form of
        \begin{equation}
            \frac{
                |\widetilde q_U(\bm n)-p_U(\bm n)|
            }{
                p_U(\bm n)
            }
            \le
            \frac{2\beta k_1+\frac{k_2}{\mathrm{poly}(n)}}{\tau},
        \end{equation}
        which proves the claim.
    \end{proof}
    In contrast to the standard Stockmeyer reduction, one must here work with a uniformly random output pattern  $\bm{n}$, a subtlety that is absent in the dilute regime of Boson Sampling, where all outcomes are equivalent up to permutations of the modes. The key additional ingredient is average anticoncentration, which guarantees that a non-negligible fraction of Haar-random output probabilities are of order \(1/|\mathcal S_{m,n}|\). On this set of instances, the additive estimate furnished by Stockmeyer's algorithm is automatically promoted to an inverse-polynomial multiplicative approximation.

   Let \(\mathcal{D}_{m,n}\) denote the probability distribution on \(n\times n\) matrices \(U_S\) induced by drawing \(U\sim \mathrm{Haar}(m)\) and \(S\) uniformly from \(\mathcal{S}_{m,n}\).   The following problem is a multiplicative version of $|\mathrm{SUPER}|^2_{\pm}|$ used in arguments computational hardness of Boson Sampling in the saturated regime ~\cite{bouland2025complexitytheoreticfoundationsbosonsamplinglinear}.

    \begin{problem}[Sub-Unitary Permanent Estimation with Repetitions (multiplicative error version), $|\mathrm{SUPER}|^2_{\times}$]
        Given $V \sim \mathcal{D}_{m, n}$, output $z \in \mathbb{C}$ such that
        \begin{equation}
            \left| z - |\per(V)|^2 \right| \leq \frac{1}{\poly(m,n)} | \per(V)|^2,
        \end{equation}
        with probability $P_{\mathrm{succ}}$ at least $1- \delta$ over choice of $V$.
    \end{problem}

   In analogy to  ~\cite{bouland2025complexitytheoreticfoundationsbosonsamplinglinear} we can conjecture that $|\mathrm{SUPER}|^2_{\times}$ is hard $\# \mathsf{P}$-hard in the saturated regime.
    \begin{conjecture}\label{conj:SUPERmult}
        $|\mathrm{SUPER}|^2_{\times}$ is hard $\# \mathsf{P}$ for $m=\theta(n)$  with success probability $P_{\mathrm{succ}}$ greater than an absolute constant.
    \end{conjecture}

    \noindent As a consequence of \Cref{thm:stockmeyer_average_anticoncentration}, one can prove the following statement:
    \begin{theorem}\label{th:MultiplicativeAPPROX}
         The existence of a classical sampler that samples from a distribution approximating the saturated regime Boson Sampling distribution from \Cref{eq:particle_detection_probability} to accuracy $\beta = 1 / \poly(m)$ implies a  \(\mathsf{BPP}^{\mathsf{NP}}\) algorithm for \( |\mathrm{SUPER}|^2_{\times} \) with
        \begin{equation}\label{eq:boundONP}
            P_{\mathrm{succ}} \geq \frac{1}{8 \left(\const\right)}\ .
        \end{equation}
    \end{theorem}
    \begin{proof}
        Let \(\beta=1/\mathrm{poly}(n)\), and set
        \begin{equation}
            k_1=k_2=\frac{4 \left(\const\right)}{(1-\tau)^2}.
        \end{equation}
        Then, according to \Cref{thm:stockmeyer_average_anticoncentration}, existence of a classical sampler with probability distribution \(q_U\) satisfying 
        $
            \|p_U-q_U\|_{\mathrm{TV}}\le \beta
        $, implies the existence of a \(\mathsf{BPP}^{\mathsf{NP}}\) algorithm which approximates the output probabilities \(p_U(\bm n)\) to inverse-polynomial multiplicative error 
        \begin{equation}
            \frac{4\left(\const\right)}{\tau (1-\tau)^2} \left(2\beta +\frac{1}{\mathrm{poly}(n)}\right)
        \end{equation}
        with probability at least
        \begin{equation}
            \frac{(1-\tau)^2}{2\left( \const\right)}.
        \end{equation}
        By choosing $\tau = 1/2$, we get the probability from \eqref{eq:boundONP}.

        Now, we need to use the Reverse Embedding lemma~\cite{bouland2025complexitytheoreticfoundationsbosonsamplinglinear} in order to relate approximate Boson Sampling to the problem \(|\mathrm{SUPER}|^2_{\times}\). This comes from a mismatch in input types: Stockmeyer's algorithm applies to Haar-random unitaries \(U\), whereas \(|\mathrm{SUPER}|^2_{\times}\) is formulated in terms of \(n\times n\) matrices drawn from the correlated distribution of submatrices of Haar-random unitaries, allowing repeated rows. The lemma resolves this mismatch by showing that, given such a matrix \(V\sim\mathcal{D}_{m,n}\), one can generate in $\mathsf{BPP}^{\mathsf{NP}}$ a Haar-random unitary \(U\) together with an output pattern \(\bm n\) such $U_{\bm{n}}=V$.

        Once this correspondence is established, an approximation to the Boson Sampling output probability \(p_U(\bm n)=\frac{|\mathrm{Per}(U_{\bm{n}})|^2}{\bm{n}!}\) immediately yields an approximation to the associated squared Permanent, since these quantities are related by the usual normalization factor corresponding to the output collisions. Therefore, the multiplicative approximation to \(p_U(\bm n)\) furnished by \Cref{thm:stockmeyer_average_anticoncentration} translates into a solution of \(|\mathrm{SUPER}|^2_{\times}\).
    \end{proof}

    Combining \Cref{th:MultiplicativeAPPROX} with Conjecture \ref{conj:SUPERmult}, we obtain the usual hardness consequence for sampling from Boson Sampling circuits in the saturated regime. Indeed, an week classical sampler for saturated Boson Sampling up to $\beta=1/\mathrm{poly}(n)$ TV error would, by \ref{th:MultiplicativeAPPROX}, yield a \(\mathsf{BPP}^{\mathsf{NP}}\) algorithm for \( |\mathrm{SUPER}|^2_{\times} \) on a constant fraction of instances. Assuming Conjecture \ref{conj:SUPERmult}, this puts a \(\#\mathsf{P}\)-hard multiplicative approximation problem in \(\mathsf{BPP}^{\mathsf{NP}}\). By Toda's theorem \cite{AroraBarak2009}, this implies a collapse of the polynomial hierarchy. The assumption here is the same multiplicative average-case hardness assumption used previously for RCS \cite{Hangleiter2018anticoncentration}, IQP \cite{Bremner_2017}, and Fermion Sampling \cite{Oszmaniec_2022}.

    \subsection*{Anticoncentration for other photonic schemes}

    We also apply our methods to study anticoncentration in further quantum advantage schemes, for example, in Scattershot Boson Sampling~\cite{Bentivegna_2015}. 
    Here, the correct definition of anticoncentration used in the Paley-Zygmund inequality is
    \begin{equation}\label{eq:ac_sbs_lxeb}
        \mathrm{AC}_{\mathrm{SBS}}(m, n, d)=|\mathcal{S}_{m,n}|\,|\mathcal{S}_{d,n}|\,\mathrm{LXEB}_{\mathrm{ref},\mathrm{SBS}}(m, n, d),
    \end{equation}
    where $\mathrm{LXEB}_{\mathrm{ref},\mathrm{SBS}}(m, n)$ is defined by \Cref{eq:lxe_ref_sbs_def}.
    Using the expression in \Cref{prop:scattershot_final_form}, we can also compute the corresponding anticoncentration scores and derive their asymptotic bounds.
    \begin{proposition}\label{prop:scatter_asymptotic_bounds}
        Fixing the number of uniform two-mode squeezing sources to be equal to the number of modes $m$, the anticoncentration score for Scattershot Boson Sampling satisfies
        \begin{subequations}
        \begin{align}
            \mathrm{AC}_{\mathrm{SBS}}(m,n,m)
            &= \Omega\!\left(\frac{mn}{m+n}\right), \\
            \mathrm{AC}_{\mathrm{SBS}}(m,n,m)
            &= O\!\left(\sqrt{mn}\right).
        \end{align}
        \end{subequations}
    \end{proposition}
    \noindent The proof is given in \Cref{app:scattershot_bounds}. In particular, these bounds imply that the Scattershot Boson Sampling distribution weakly anticoncentrates when $d=m$. By contrast, when $d=1$, we obtain
    \begin{equation}
        \mathrm{AC}_{\mathrm{SBS}}(m,n,1)
        =
        \frac{n! \binom{n+m-1}{n}}{\left(\frac{m+1}{2}\right)_n}
        \geq \exp(\Theta(n)),
    \end{equation}
    whenever $m = \Omega(n)$. Hence, in this regime, the distribution fails to anticoncentrate, which is also indicated by \Cref{fig:sbs_ac_plot}.

    \begin{figure}[t]
         \centering
         \includegraphics[width=\linewidth]{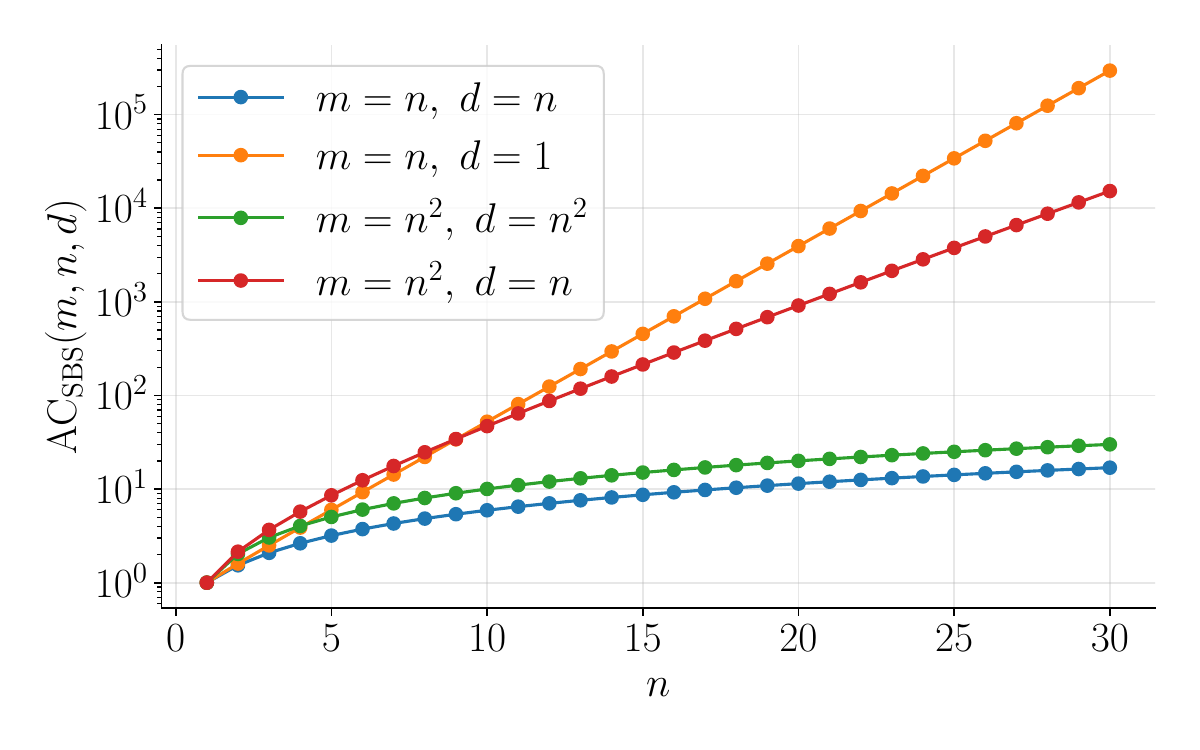}
         \caption{
            Anticoncentration scores for Scattershot Boson Sampling under several scalings of $m$ and $d$. Weak anticoncentration is observed in the regime $m = d$, while the remaining scalings show no comparable anticoncentration behavior over the plotted range.
         }
        \label{fig:sbs_ac_plot}
    \end{figure}

    Turning to Gaussian Boson Sampling, our framework also enables a numerical investigation of anticoncentration in this setting as well.  In  Refs.~\cite{Ehrenberg_2025_second_moment,Ehrenberg_2025}, the authors developed a graph-theoretic framework for analyzing moments of the Gaussian Boson Sampling distribution in the dilute regime and showed that, similarly to the case of Scattershot Boson Sampling discussed above, the model exhibits a transition in anticoncentration depending on how the number of initially squeezed modes scales relative to the number of detected photons. When the number of squeezed input modes grows too slowly, anticoncentration fails, whereas sufficiently rapid growth leads to weak anticoncentration.
    
    We get insight into such phenomena from the formula in \Cref{prop:LXE_ref_uniform_squeezed_closed_binom} together with the identity
    \begin{equation}
        \mathrm{AC}_{\mathrm{GBS}}(m,n,d)
        =
        |\mathcal S_{m,n}|\,\mathrm{LXE}_{\mathrm{ref}}^{(n)}(\rho),
    \end{equation}
    where $d$ denotes the number of uniformly squeezed single-mode input states.
    Importantly, the analysis from \cite{Ehrenberg_2025_second_moment} only applies to the dilute regime, while our work extends to any relation between $n$ and $m$.
    In the case $d=1$, that is, when only a single input mode is squeezed, the anticoncentration score scales as
    \begin{equation}
        \mathrm{AC}_{\mathrm{GBS}}(m,n,1)
        =
        \frac{(m)_n}{\left(\frac m2+1\right)_n} \geq \exp(\Theta(n))
    \end{equation}
    whenever $m = \Omega(n)$,
    which shows that anticoncentration does not occur in this regime.     At the opposite extreme, namely when $d=m$, the anticoncentration score satisfies the following asymptotic bounds.
    \begin{proposition}\label{prop:gaussian_asymptotic_bounds}
        If the number of uniform squeezing sources equals the number of modes $m$, then the anticoncentration score for Gaussian Boson Sampling satisfies
        \begin{align}\begin{split}
            \mathrm{AC}_{\mathrm{GBS}}(m,n,m)
            &= \Omega\!\left(\sqrt{\frac{nm}{m+n}}\right),\\
            \mathrm{AC}_{\mathrm{GBS}}(m,n,m)
            &= O\!\left(\sqrt{\frac{n(m+n)}{m}}\right).
        \end{split}\end{align}
    \end{proposition}
    \noindent A proof is provided in \Cref{app:gbs_bounds}. These bounds show that, similarly to Scattershot Boson Sampling, Gaussian Boson Sampling weakly anticoncentrates when $d=m$. Moreover, these results are consistent with the numerically observed behavior from \Cref{fig:gbs_ac_plot}.

    \begin{figure}[t]
         \centering
         \includegraphics[width=\linewidth]{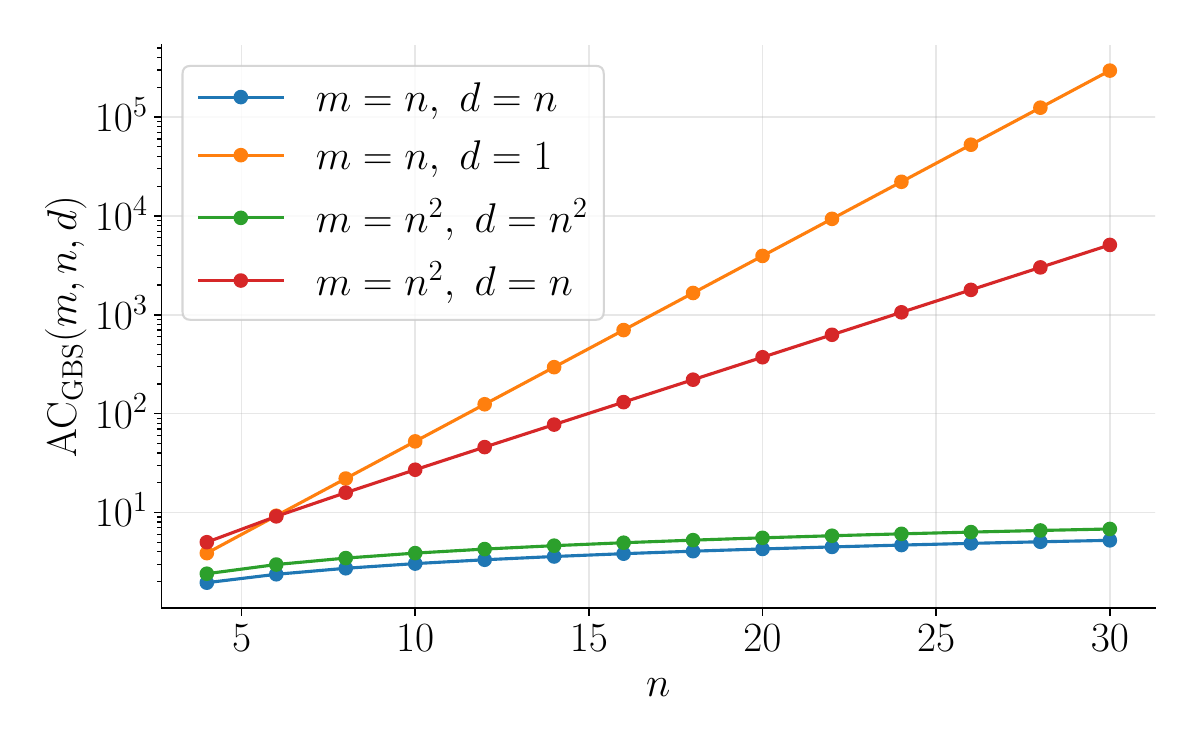}
         \caption{
            Anticoncentration scores for Gaussian Boson Sampling under several scalings of $m$ and $d$. Weak anticoncentration is observed in the regime $m=d$, while the remaining scalings show no comparable anticoncentration behavior over the plotted range.
         }
        \label{fig:gbs_ac_plot}
    \end{figure}

\bibliography{bibliography}

\newpage

\onecolumngrid

\section*{Appendices}
\appendix

This appendix gathers the technical results used throughout the main text. In \Cref{app:expressing_the_projector}, we show how to express the projector \(\mathbb{P}_k\) in terms of the bosonic swap operator, while \Cref{app:simplified_formulas} is devoted to simplifying the associated overlap formulas. In \Cref{app:proof_bosonic_swap_efficient}, we present an algorithm for computing bosonic swap expectation values for product states. In \Cref{app:bounds}, we derive bounds and asymptotic estimates for anticoncentration scores in photonic quantum advantage schemes, and in \Cref{app:further_results}, we discuss additional results and consequences of our work. Finally, \Cref{app:useful_comb} collects several combinatorial identities and estimates used repeatedly in the proofs. Some of these results may also be of independent interest beyond the scope of this work.

\section{Expressing $\mathbb{P}_k$ in terms of the bosonic swap operator}\label{app:expressing_the_projector}
    In this section, we show how to calculate the expression for the projectors onto the irreducible components appearing during the calculation of the $\mathrm{LXE}$ reference value. Roughly speaking, we will construct the projectors by projecting Young symmetrizers onto the symmetric subspace on the two copies. This will result in a linear combination of certain operators that have a particularly nice interpretation.

    Before delving into specific computations, let us start by recapping the relevant parts of \Cref{sec:framework}. To calculate the reference value of the $\mathrm{LXE}$ score, we need to investigate the irreducible decomposition of $\varphi_n^{\otimes 2}$. The subspace $\Sym^n(\mathbb{C}^m)$ is itself irreducible and when $m \geq 2$ we know that its two-fold tensor power decomposes into irreducible subspaces as
    \begin{equation}
        \Sym^n(\mathbb{C}^m) \otimes \Sym^n(\mathbb{C}^m) \cong \bigoplus_{k=0}^n \, \mathbb{S}_{(2n-k, k)}(\mathbb{C}^m).
    \end{equation}
    Here, $\mathbb{S}_{(2n-k, k)}(\mathbb{C}^m)$ represents the irreducible subspace corresponding to the Young diagram with $2n-k$ elements in the first row and $k$ elements in the second.
    For later purposes, let us visualize the irreducible decomposition via Young diagrams:
    \begin{equation}
        \begin{ytableau}
            *(white) & \none & \none[\cdots] & \none & *(white) 
        \end{ytableau}
        \,\otimes\,
        \begin{ytableau}
            *(black) & \none & \none[\cdots] & \none & *(black) 
        \end{ytableau}
        \;
        \cong
        \;
        \begin{ytableau}
            *(white) & \none & \none[\cdots] & \none & *(white) & *(black)
            & \none & \none[\cdots] & \none & *(black) 
        \end{ytableau}
        \,\oplus\,
        \raisebox{0.3em}{
        \begin{ytableau}
            *(white) & \none & \none[\cdots] & \none & *(white) & *(black)
            & \none & \none[\cdots] & \none & *(black) \\
            *(black)\\
        \end{ytableau}}
        \,\oplus\,\dots\,\oplus\,
        \raisebox{0.3em}{
        \begin{ytableau}
            *(white) & \none & \none[\cdots] & \none & *(white) & *(white) & *(black) \\
            *(black)  & \none & \none[\cdots] & \none & *(black)\\
        \end{ytableau}}
        \,\oplus\,
        \raisebox{0.3em}{
        \begin{ytableau}
            *(white) & \none & \none[\cdots] & \none & *(white)  \\
            *(black)  & \none & \none[\cdots] & \none & *(black)\\
        \end{ytableau}}
    \end{equation}
    Considering a Young diagram $\lambda$ of size $n$ the corresponding irreducible subspace can be expressed as the image of a linear operator:
    \begin{equation}\label{eq:S_lambda_by_young_symmetrizer}
        \mathbb{S}_{\lambda}(\mathbb{C}^m) = \mathrm{Im}( Y_{\mu} ),
    \end{equation}
    where $\mu$ is a \textit{Young tableau} of shape $\lambda$ and the linear operator $Y_{\mu}$ is the corresponding \textit{Young symmetrizer} constructed as follows.
    The symmetric group $S_n$ acts on $(\mathbb{C}^m)^{\otimes n}$ by permuting tensor factors.
    Let us denote by $R_\mu \leq S_n$ the permutation group that preserves the rows of the Young tableau $\mu$, and by $C_\mu \leq S_n$ the group that preserves the columns. The Young symmetrizer $Y_{\mu} \in \mathbb{C} S_{n}$ corresponding to this Young tableau is defined as
    \begin{subequations}
    \begin{align}
        Y_{\mu} &\coloneqq A_{\mu} B_{\mu}, \\
        A_{\mu} \coloneqq \sum_{g \in R_\mu} g, &\quad
        B_{\mu} \coloneqq \sum_{g \in C_\mu} \sgn(g) g.
    \end{align}
    \end{subequations}
    Via the natural action of $\mathbb{C}S_n$ on $(\mathbb{C}^m)^{\otimes n}$, the Young symmetrizer $Y_\mu$ defines a linear operator whose image is isomorphic to the Schur functor $\mathbb{S}_\lambda(\mathbb{C}^m)$ as in \Cref{eq:S_lambda_by_young_symmetrizer}.
    The image of $Y_{\mu}$ acting on $(\mathbb{C}^m)^{\otimes n}$ corresponds to an irreducible representation and that $Y_{\mu}$ is idempotent up to some scalar multiple. We note, that $Y_{\mu}$ is not unique as one would get different groups in place of $R_\mu$ and $C_\mu$ if one would choose a different Young tableau, although they are conjugate in the group algebra.
    
    Now, let us turn our attention towards giving an explicit expression for the projectors corresponding to the irreducible subspaces $\mathbb{S}_{(2n-k, k)}(\mathbb{C}^m)$, which we will denote by $\mathbb{P}_k$:
    \lemmapk*
    \begin{proof}
        Considering a Young tableau $\mu$ of shape $(2n-k, k)$ we denote the permutation group that preserves the rows by $R_{\mu} \cong S_{2n-k} \times S_k$, and the subgroup that preserves the columns by $C_\mu \cong S_2^{\times k}$. The Young symmetrizer $Y_\mu$ corresponding to this Young tableau is given by
        \begin{subequations}
        \begin{align}
            Y_\mu &\coloneqq A_\mu B_\mu, \\  
            A_\mu &\coloneqq \sum_{\tau \in R_\mu} \tau =  \left(\sum_{\sigma \in S_{2n-k}} \sigma\right)  \left(\sum_{\xi \in S_k} \xi\right), \\
            B_\mu &\coloneqq \sum_{\pi \in C_\mu} \sgn(\pi) \pi = \sum_{X \subseteq [k]} (-1)^{|X|} \prod_{i \in X} e_{(1, i), (2, i)},
        \end{align}
        \end{subequations}
        where $e_{(1, i), (2, i)} \in \mathbb{C}S_{2n}$ permutes the boxes $(1, i)$ and $(2, i)$ in the Young tableau $\mu$.
        However, we are interested in $\Sym^n(\mathbb{C}^m)^{\otimes 2}$, so we need to consider the symmetrizer acting on elements of $\Sym^n(\mathbb{C}^m)^{\otimes 2}$. The projector onto the symmetric subspace on one copy of $(\mathbb{C}^m)^{\otimes n}$ can be written as
        $
            \mathbb{P}_{\text{sym}} = \frac{1}{n!} \sum_{\pi \in S_n} \pi,
        $
        where $\pi$ is acting on $(\mathbb{C}^m)^{\otimes n}$ by permuting the components in the tensor product.
        Furthermore, let
        \begin{equation}
            \rho : \mathbb{C}S_{2n} \to \mathrm{End}\bigl((\mathbb{C}^m)^{\otimes 2n}\bigr)
        \end{equation}
        denote the natural linear extension of the permutation action of \(S_{2n}\) on \((\mathbb{C}^m)^{\otimes 2n}\). Explicitly, for each \(\pi \in S_{2n}\),
        \begin{equation}
            \rho(\pi)(v_1 \otimes \cdots \otimes v_{2n})
            =
            v_{\pi^{-1}(1)} \otimes \cdots \otimes v_{\pi^{-1}(2n)}.
        \end{equation}
        For brevity, let us denote $\hat{A} \coloneqq \rho(A)$ for any $A \in \mathbb{C}S_{2n}$ in the following discussion.
        
        Note, that $\mathbb{P}_{\text{sym}}^{\otimes 2} \hat{Y}_{\mu} \mathbb{P}_{\text{sym}}^{\otimes 2} \cong \mathbb{P}_{\text{sym}}^{\otimes 2} \hat{Y}_{\mu'} \mathbb{P}_{\text{sym}}^{\otimes 2}$, since
        $\mathbb{P}_{\text{sym}}^{\otimes 2} \hat{Y}_{\mu} \mathbb{P}_{\text{sym}}^{\otimes 2}$ is in the commutant of $\Sym^n(\mathbb{C}^m)^{\otimes 2}$, which is abelian.
        Therefore, we know that the projector in question must be proportional to the symmetrized projected to $\Sym^n(\mathbb{C}^m)^{\otimes 2}$, i.e.,
        \begin{equation}
            \mathbb{P}_k \propto  \mathbb{T}_k \coloneqq \mathbb{P}_{\text{sym}}^{\otimes 2} \hat{Y}_{\mu} \mathbb{P}_{\text{sym}}^{\otimes 2}.
        \end{equation}
        To obtain a closed expression, we can expand it and write
        \begin{subequations}
        \begin{align}
           \mathbb{T}_k &= \sum_{\sigma \in S_{2n-k}, \sigma' \in S_k} \sum_{X \subseteq [k]} (-1)^{|X|} \mathbb{P}_{\text{sym}}^{\otimes 2} (\hat{\sigma} \otimes \hat{\sigma}')  \left(\prod_{i \in X} \mathbb{F}_{i, i+n}\right) \mathbb{P}_{\text{sym}}^{\otimes 2}
            \\&
            =
            k! \sum_{\sigma \in S_{2n-k}} \sum_{X \subseteq [k]}  (-1)^{|X|} (\mathbbm{1} \otimes \mathbb{P}_{\text{sym}}) (\hat{\sigma} \otimes \mathbbm{1})  \left(\prod_{i \in X} \mathbb{F}_{i, i+n}\right) \mathbb{P}_{\text{sym}}^{\otimes 2}
            \\& 
            = k! \sum_{l = 0}^k \binom{k}{l} \sum_{\beta' \in S_n} \sum_{\sigma \in S_{2n-k}} \sum_{\alpha, \alpha' \in S_n} (-1)^{l} (\mathbbm{1} \otimes \mathbb{P}_{\text{sym}} ) (\hat{\sigma} \otimes \mathbbm{1})  \left(\prod_{i = 1}^l \mathbb{F}_{i, i+n} \right) \mathbb{P}_{\text{sym}}^{\otimes 2},
        \end{align}
        \end{subequations}
        where we denoted $\mathbb{F}_{i, i+n} \coloneqq \rho(e_{(1, i), (2, i)})$ for brevity.
        To keep track of the combinatorial factors, we visualize applying this operator to the boxes of the Young tableau, where the boxes corresponding to the second component are colored as $\,\begin{ytableau}
            *(pink)\\
        \end{ytableau}$\,. The colored boxes do not move when a permutation is applied---this is just for highlighting the terms corresponding to the second component during the computation.
        Moreover, we put labels in the boxes, which are numbers $1, \dots, n$ and $1', \dots, n'$, where the prime $'$ denotes that the label originates from the colored boxes. For example, for $n=3$ and $k=2$, we can write
        \ytableausetup{mathmode, boxframe=normal, boxsize=1.5em}
        \begin{equation}
            \begin{ytableau}
                1 & 2 & 3 & *(pink) 3'  \\
                *(pink) 1' & *(pink) 2' \\
            \end{ytableau}
        \end{equation}
        More generally, let us consider the Young tableau
        \ytableausetup{mathmode, boxframe=normal, boxsize=2.5em}%
        \begin{equation}
            \Gamma \;\coloneqq\; \begin{ytableau}
                1 & \none[\cdots] & k & \none[\cdots] & n &*(pink) \scriptscriptstyle (k+1)' & \none[\cdots] &  *(pink) n' \\
                *(pink) 1' & \none[\cdots] & *(pink) k' \\
            \end{ytableau}        
        \end{equation}
        \ytableausetup{mathmode, boxframe=normal, boxsize=1.5em}%
        When a permutation is applied, we can interpret it as exchanging the labels of the boxes.
        When we apply $\mathbb{P}_{\sym}^{\otimes 2}$, formally, it permutes $1, \dots, n$ and $1',\dots, n'$ independently:
        \begin{equation}
            \mathbb{P}_{\sym}^{\otimes 2} \Gamma
            \propto
            \Gamma + (\text{all possible independent permutations of $1, \dots, n$ and $1', \dots, n'$})
        \end{equation}
        Since the subsequent operations are also permutations, we can omit the labelings of the boxes and just write
        \begin{equation}
            \begin{ytableau}
                 *(white) & \none[\cdots] & *(white) &  \none[\cdots] & *(white) & *(pink) \bullet &  \none[\cdots] & *(pink) \bullet  \\
                 *(pink) \bullet &  \none[\cdots] & *(pink) \bullet \\
            \end{ytableau}
            \;\coloneqq\;
            \Gamma + (\text{all possible independent permutations of $1, \dots, n$ and $1', \dots, n'$})
        \end{equation}
        Now, applying $\sum_{l=0}^k \binom{k}{l} \prod_{i = 1}^l \mathbb{F}_i$ we get
        \begin{equation}
            \binom{k}{0}
            \cdot \,
            \raisebox{0.2em}{\begin{ytableau}
                 *(white) & \none[\cdots] & *(white) &  \none[\cdots] & *(white) & *(pink) \bullet &  \none[\cdots] & *(pink) \bullet  \\
                 *(pink) \bullet & \none[\cdots] & *(pink) \bullet \\
            \end{ytableau}}
            \,
            +
            \binom{k}{1}
            \cdot \,
            \raisebox{0.2em}{\begin{ytableau}
                 *(white) \bullet & \none[\cdots] & *(white) &  \none[\cdots] & *(white) & *(pink) \bullet &  \none[\cdots] & *(pink) \bullet  \\
                 *(pink) &  \none[\cdots] & *(pink) \bullet \\
            \end{ytableau}}
            \,
            +
            \dots
            + 
            \binom{k}{k}
            \cdot
            \,
            \raisebox{0.2em}{\begin{ytableau}
                 *(white) \bullet & \none[\cdots] & *(white) \bullet &  \none[\cdots] & *(white) & *(pink) \bullet &  \none[\cdots] & *(pink) \bullet  \\
                 *(pink) &  \none[\cdots] & *(pink) \\
            \end{ytableau}}
        \end{equation}
        If we apply the permutations on the the first row we get
        \begin{multline}
            \sum_{\sigma \in S_{2n-k}} \sigma\;
            \raisebox{-0.5em}{\begin{ytableau}
                 *(white) \bullet & \none[\cdots] & *(white) \bullet & *(white) & \none[\cdots] & *(white) & *(pink) \bullet &  \none[\cdots] & *(pink) \bullet
            \end{ytableau}}
            \\=
            (n - k + l)! (n-l)!
            \left[
                \raisebox{-0.5em}{\begin{ytableau}
                     *(white) \bullet & \none[\cdots] & *(white) \bullet &  *(white) & \none[\cdots] & *(white) & *(pink) \bullet &  \none[\cdots] & *(pink) \bullet
                \end{ytableau}}
                +
                (\text{all possible permutations})
            \right].
        \end{multline}
        We get the multiplier $(n - k + l)! (n-l)!$ in front because we have $n-l$ empty and $n-k+l$ filled boxes in the first row, but the permutation permutes the underlying labels $1, \dots, n$ and $1', \dots, n'$. Hence, we need to multiply by the number of possible label arrangements in the first row.
        Moreover, after performing a symmetrization again, it becomes clear that the result only depends on the number of bullets ending up in the first component. Let us call this number $q$. We know that $k-l \leq q \leq n-l$, and that for each $q$, there is a multiplicative factor of $\binom{n}{q}$, i.e., the number of ways one can arrange $q$ bullets in the first component, and a multiplicative factor of $\binom{n-k}{n-l-q}$, i.e., the number of ways one can arrange the remaining bullets in the colored $(n-k)$-element tail of the first row.
    
        Finally, we can write
        \begin{equation}
             \mathbb{T}_k =
             k!
             \sum_{l=0}^k (-1)^l \binom{k}{l} (n-k+l)! (n-l)! \sum_{q=k-l}^{n-l} \binom{n}{q} \binom{n-k}{n-l-q} \mathbb{S}_q,
        \end{equation}
        where we denoted by $\mathbb{S}_q$ the \textit{bosonic swap operator}, defined on two copies as
        \begin{equation}\label{eq:bosonic_swap_first_quantized}
            \mathbb{S}_q \coloneqq \mathbb{P}_{\text{sym}}^{\otimes 2} \left(\prod_{i = 1}^q \mathbb{F}_{i, i+n} \right) \mathbb{P}_{\text{sym}}^{\otimes 2}.
        \end{equation}
        However, in order to fully determine $\mathbb{P}_k$, we need to get the normalization factor that ensures the correct dimensionality using the hook-content formula~\cite{fultonharris}, i.e.,
        \begin{equation}
             \Tr \mathbb{P}_k = \dim \mathbb{S}_{(2n-k, k)}(\mathbb{C}^m) = 
             \prod_{(i, j) \in \text{shape}(\mu)} \frac{m-i-j}{h_{i,j}}
             =
             \frac{2n-2k+1}{2n-k+1}
             \binom{2n+m-k-1}{m-1}\binom{m-2+k}{m-2},
        \end{equation}
        where $h_{i, j}$ is the hook length corresponding to the box in row $i$ and column $j$.
        To obtain the scaling between $\mathbb{T}_k$ and $\mathbb{P}_k$, we view $\Tr \mathbb{P}_k$ as a polynomial in $m$, and extract the leading term corresponding to $m^{2n}$:
        \begin{equation}
            \Tr \mathbb{P}_k = \dim \mathbb{S}_{(2n-k, k)}(\mathbb{C}^m)
            = \frac{m^{2n}}{ \prod_{(i, j) \in \text{shape}(\mu)} h_{i, j}} + (\text{lower order terms}).
        \end{equation}
        Moreover, to obtain the same leading term for $\Tr \mathbb{P}_k$, we use the fact that the trace of the bosonic swap is
        \begin{align}\begin{split}
            \Tr \mathbb{S}_q = \Tr_{(\mathbb{C}^m)^{\otimes n} \otimes (\mathbb{C}^m)^{\otimes n}} \left[
                \mathbb{P}_\sym^{\otimes 2} 
                \prod_{i = 1}^q \mathbb{F}_{i, i+n}
            \right]
            &=
            \frac{1}{n!^2}
            \sum_{\pi, \sigma \in S_n}
            \Tr_{(\mathbb{C}^m)^{\otimes n} \otimes (\mathbb{C}^m)^{\otimes n}} \left[
                (\hat{\pi} \otimes \hat{\sigma})
                \prod_{i = 1}^q \mathbb{F}_{i, i+n}
            \right]
            \\&= \frac{\binom{m+n-1}{n}^2}{\binom{m+q-1}{q}}
            = \frac{q!}{n!^2} \frac{(m)_n^2}{(m)_q}.
        \end{split}\end{align}
        With this formula, we can also expand $\Tr \mathbb{T}_{\mu}$ and extract the leading coefficient corresponding to $m^{2n}$, which is
        \begin{equation}
            \Tr \mathbb{T}_{k} = \frac{(-1)^k}{\binom{n}{k}} m^{2n} + (\text{lower order terms}).
        \end{equation}
        We know that $\Tr \mathbb{P}_k \propto \Tr \mathbb{T}_{k}$, which means that the ratio of the leading order coefficient gives the ratio between $\mathbb{P}_k$ and $\mathbb{T}_k$.
        Hence, we can write that
        \begin{align}
        \begin{split}
            \mathbb{P}_k = \frac{\Tr \mathbb{P}_k}{\Tr \mathbb{T}_k} \mathbb{T}_k &= (-1)^k \frac{2n-2k+1}{(2n-k+1)!k!} \binom{n}{k} \mathbb{T}_k
            \\&=
            (-1)^k \frac{2n-2k+1}{(2n-k+1)!} \binom{n}{k}
             \sum_{l=0}^k (-1)^l \binom{k}{l} (n-k+l)! (n-l)! \sum_{q=k-l}^{n-l} \binom{n}{q} \binom{n-k}{n-l-q} \mathbb{S}_q.
        \end{split}
        \end{align}
        We can collect the factorials into binomial coefficients to get
        \begin{equation}
                \mathbb{P}_k = \frac{2n-2k+1}{2n-k+1} \binom{n}{k}
             \sum_{l=0}^k (-1)^{k-l} \frac{\binom{k}{l}}{\binom{2n-k}{n-l}} \sum_{q=k-l}^{n-l} \binom{n}{q} \binom{n-k}{n-l-q} \mathbb{S}_q.
        \end{equation}
        For practical reasons, we extract the coefficients of the bosonic swap operators as
        \begin{equation}
            \mathbb{P}_k
            = \sum_{q=0}^{n} c_{k,q}\,\mathbb{S}_q, \quad \text{where}\quad
            c_{k,q}
            \coloneqq
            \frac{2n-2k+1}{2n-k+1}\binom{n}{k}\binom{n}{q}
            \sum_{l=\max(0,k-q)}^{\min(k,n-q)}
            (-1)^{k-l}
            \frac{\binom{k}{l}\binom{n-k}{\,l+q-k\,}}{\binom{2n-k}{n-l}} .
        \end{equation}
        Finally, a key feature of the bosonic swap operator \(\mathbb{S}_q\) is that its expectation value depends only on the \(q\)-particle reduced states~\cite{Oszmaniec_2016}. More precisely, we can write
        \begin{equation}
            \Tr [\mathbb{S}_q (\rho \otimes \sigma) ]
            =
            \Tr [ \Tr_{n-q}(\rho)\,\Tr_{n-q}(\sigma)],
        \end{equation}
        where \(\Tr_{n-q}\) denotes the partial trace over \(n-q\) particles.
    \end{proof}
    \begin{remark}\label{remark:symmetry}
        The irreducible subspaces of the tensor squares are separated for the symmetric and antisymmetric parts as
        \begin{equation}
            \Sym^2(\Sym^n(\mathbb{C}^m)) \cong \bigoplus_{\substack{k=0 \\ k\text{ even}}}^n
            \mathbb{S}_{(2n-k, k)}(\mathbb{C}^m),  \qquad
            \bigwedge\! {}^2(\Sym^n(\mathbb{C}^m)) \cong \bigoplus_{\substack{k=1 \\ k\text{ odd}}}^n
            \mathbb{S}_{(2n-k, k)}(\mathbb{C}^m), 
        \end{equation}
        which is apparent from the fact that $\mathbb{P}_k$ contains $k$ column swaps, contributing a $(-1)^k$ term when the flip operator is applied on two copies of $\Sym^n(\mathbb{C}^m)$. Hence, when $k$ is odd, $\mathbb{P}_k$ maps two identical copies of the same state to $0$.
    \end{remark}
    The expression in \Cref{eq:bosonic_swap_to_trace} can be used for most computations, but we will also give an explicit expression of this operator in terms of the creation/annihilation operators in the form of
    \begin{equation}\label{eq:bosonic_swap_boxed}
        \mathbb{S}_q = \frac{1}{\binom{n}{q}^2} 
        \sum_{\substack{\bm{q} \geq \bm{0} \\ |\bm{q}| = q}} 
        \sum_{\substack{\bm{r} \geq \bm{0} \\ |\bm{r}| = q}}
        \frac{1}{\bm{q}! \bm{r}!} (\bm{a}^\dagger)^{\bm{r}} \bm{a}^{\bm{q}}
        \otimes 
        (\bm{a}^\dagger)^{\bm{q}} \bm{a}^{\bm{r}},
    \end{equation}
    where we denoted $\bm{a}^{\bm{q}} \coloneqq \prod_{i=1}^m a_i^{q_i}$ for brevity.
    To obtain this formula, it is useful to introduce the notation \( \phi_1 \vee \dots \vee \phi_n \) for the \textit{symmetrized tensor product}:
    \begin{equation}\label{eqn:vee}
        \phi_1 \vee \dots \vee \phi_n \coloneqq \frac{1}{\sqrt{n!}} \sum_{\pi \in S_n} \phi_{\pi(1)} \otimes \dots \otimes \phi_{\pi(n)},
    \end{equation}
    which yields totally symmetric vectors, i.e., \( U_\pi \phi_1 \vee \dots \vee \phi_n = \phi_1 \vee \dots \vee \phi_n \) for any $\pi \in S_n$ and $\phi_i \in \mathbb{C}^m$, and hence $\phi_1 \vee \dots \vee \phi_n \in \Sym^n(\mathbb{C}^m)$.
    Using this notation, we calculate the action of $\mathbb{S}_q$ directly on $\Sym^n(\mathbb{C}^m)^{\otimes 2}$. For basis elements of $e_{i_1},\dots,e_{i_m}, e_{j_1},\dots,e_{j_m} \in \mathbb{C}^m$, the bosonic swap gives
    \begin{multline}
        \mathbb{S}_q \left(
            \bigvee_{k=1}^n e_{i_{k}} \otimes 
            \bigvee_{l=1}^n e_{j_{l}}
        \right)
        = 
        \frac{1}{n!^2}
        \sum_{\pi, \sigma \in S_n}
        \mathbb{P}_{\mathrm{sym}}
        \left(\bigotimes_{a=1}^q e_{j_{\sigma(a)}} \otimes \bigotimes_{k=q+1}^n e_{i_{\pi(k)}}\right)
        \otimes 
        \mathbb{P}_{\mathrm{sym}}
        \left(\bigotimes_{b=1}^q e_{i_{\pi(b)}} \otimes \bigotimes_{l=q+1}^n e_{j_{\sigma(l)}}\right)
        \\=
        \frac{1}{\binom{n}{q}^2}
        \sum_{\pi, \sigma \in S_n / (S_q \times S_{n-q})}
        \left(\bigvee_{a=1}^q e_{j_{\sigma(a)}} \vee \bigvee_{k=q+1}^n e_{i_{\pi(k)}}\right)
        \otimes 
        \left(\bigvee_{b=1}^q e_{i_{\pi(b)}} \vee \bigvee_{l=q+1}^n e_{j_{\sigma(l)}}\right).
    \end{multline}
    In the following, we consider Fock basis states written in the first-quantized form as
    \begin{equation}
        \ket{n_1, \dots, n_m} = \frac{1}{\sqrt{n_1! \dots n_m!}}\; e_{1}^{\vee n_1} \vee \dots \vee e_{m}^{\vee n_m}, \quad
        \ket{m_1, \dots, m_m} = \frac{1}{\sqrt{m_1! \dots m_m!}}\; e_{1}^{\vee m_1} \vee \dots \vee e_{m}^{\vee m_m}.
    \end{equation}
    Let us denote $\ket{\bm{n}} = \ket{n_1, \dots, n_m}$ and $\bm{n!} = n_1! \cdots n_m!$ for brevity.
    Then, each term in the sum corresponds to the removal of particles corresponding to the occupation numbers $\bm{q}$ and $\bm{r}$ with $\sum_{i=1}^m q_i = \sum_{i=1}^m r_i = q$, but the removal of $\bm{q}$ (or $\bm{r}$) can be done $\binom{\bm{n}}{\bm{q}}$ (or $\binom{\bm{m}}{\bm{r}}$) ways. Therefore, we can write
    \begin{multline}
        \mathbb{S}_q (\ket{\bm{n}} \otimes \ket{\bm{m}})
        = \frac{1}{\sqrt{\bm{n}! \bm{m}!}} \mathbb{S}_q ( e_{1}^{\vee n_1} \vee \dots \vee e_{m}^{\vee n_m} \otimes e_{1}^{\vee m_1} \vee \dots \vee e_{m}^{\vee m_m})
        \\= \frac{1}{\binom{n}{q}^2} \frac{1}{\sqrt{\bm{n}! \bm{m}!}} \sum_{\substack{0 \leq \bm{q} \leq \bm{n} \\ |\bm{q}| = q}} 
        \sum_{\substack{0 \leq \bm{r} \leq \bm{m} \\ |\bm{r}| = q}} \!\!\!
        \binom{\bm{n}}{\bm{q}} \binom{\bm{m}}{\bm{r}} ( e_{1}^{\vee n_1-q_1+r_1} \vee \dots \vee e_{m}^{\vee n_m-q_m+r_m} \otimes e_{1}^{\vee m_1-r_1+q_1} \vee \dots \vee e_{m}^{\vee m_m-r_m+q_m})
        \\=
        \frac{1}{\binom{n}{q}^2}
        \sum_{\substack{0 \leq \bm{q} \leq \bm{n} \\ |\bm{q}| = q}} 
        \sum_{\substack{0 \leq \bm{r} \leq \bm{m} \\ |\bm{r}| = q}}
        \sqrt{\frac{(\bm{n}-\bm{q}+\bm{r})!
            (\bm{m}-\bm{r}+\bm{q})!}{\bm{n}! \bm{m}!}}
        \binom{\bm{n}}{\bm{q}} \binom{\bm{m}}{\bm{r}} \ket{\bm{n}-\bm{q}+\bm{r}} \otimes \ket{\bm{m}-\bm{r}+\bm{q}}.
    \end{multline}
    Comparing with \Cref{eq:ladder_operators_action}, we can write the final expression for the bosonic swap operator in second quantization picture as
    \begin{equation}
        \mathbb{S}_q = \frac{1}{\binom{n}{q}^2} 
        \sum_{\substack{\bm{q} \geq \bm{0} \\ |\bm{q}| = q}} 
        \sum_{\substack{\bm{r} \geq \bm{0} \\ |\bm{r}| = q}}
        \frac{1}{\bm{q}! \bm{r}!} (\bm{a}^\dagger)^{\bm{r}} \bm{a}^{\bm{q}}
        \otimes 
        (\bm{a}^\dagger)^{\bm{q}} \bm{a}^{\bm{r}}.
    \end{equation}

\section{Simplified overlap formulas for $\mathbb{P}_k$}\label{app:simplified_formulas}
    To obtain simplified formulas for reference LXEB values, we will need to compute certain overlaps of the projectors $\mathbb{P}_k$.
    To start off, we investigate the overlap with $\mathbb{D}_{m,n}$, which is used in \Cref{eq:lxe_ref_boxed} independently from the initial state considered. Using \Cref{eq:bosonic_swap_boxed} (or \cite{Carrasco_2016}) we can obtain the following bosonic swap expectation value:
    \begin{equation}
        \Tr \left[ \mathbb{S}_{q} 
        \mathbb{D}_{m,n} \right]
        = \frac{1}{| \mathcal{S}_{m, n}| \binom{n}{q}^2} \sum_{\bm{n} \in \mathcal{S}_{m, n} } \sum_{\substack{\bm{0} \leq \bm{q}\leq \bm{n} \\ |\bm{q}| = q}}
        \binom{\bm{n}}{\bm{q}}^2,
    \end{equation}
    However, this formula is not in an efficiently computable form. To simplify this, we will need to make use of the following fact:
    \begin{proposition}\label{prop:F_closed_form}
        The bivariate generating function
        \begin{equation}
        F(x,y) \coloneqq \sum_{a,b =  0}^\infty \binom{a+b}{a}^2 x^a y^b
        \end{equation}
        admits the closed form
        \begin{equation}
        F(x,y)=\frac{1}{\sqrt{1-2x-2y+x^2-2xy+y^2}}.
        \end{equation}
    \end{proposition}
    \noindent This fact can be easily deduced using Equation (16.13.4) from \cite{NIST:DLMF} and Equation (B.1) in \cite{Ananthanarayan_2023}.
    With this in mind, we can now turn our attention to simplifying the average-outcome formula for the expectation value of the bosonic swap:

    \begin{restatable}{proposition}{meanpurity}\label{prop:mean_purity_simple}
        For \(0\le q\le n\), we can write
        \begin{equation}
            \Tr \left[ \mathbb{S}_{q} 
            \mathbb{D}_{m,n} \right] = \frac{1}{\binom{n}{q}}\frac{\left(\frac{m+1}{2}\right)_n}
             {\left(\frac{m+1}{2}\right)_q
              \left(\frac{m+1}{2}\right)_{n-q}}.
        \end{equation}
    \end{restatable}
    \begin{proof}
        We start with the formula
        \begin{equation}
             \Tr \left[ \mathbb{S}_{q} 
            \mathbb{D}_{m,n} \right]
            =
            \frac{1}{|\mathcal S_{m,n}|\,\binom{n}{q}^2}
            \sum_{\substack{\bm n\in\mathbb N_0^m\\ |\bm n|=n}}
            \ \sum_{\substack{\bm q \le \bm n\\ |\bm q|=q}}
            \prod_{i=1}^m \binom{n_i}{q_i}^2.
        \end{equation}
        Setting \(\bm k\coloneqq\bm n-\bm l\), we have \(|\bm k|=q\), and therefore
        \begin{equation}
            \frac{1}{|\mathcal S_{m,n}|\,\binom{n}{q}^2}
            \sum_{\substack{\bm n\in\mathbb N_0^m\\ |\bm n|=n}}
            \ \sum_{\substack{\bm q \le \bm n\\ |\bm q|=q}}
            \prod_{i=1}^m \binom{n_i}{q_i}^2
            =
            \frac{1}{|\mathcal S_{m,n}|\,\binom{n}{q}^2}
            \sum_{\substack{\bm l\in\mathbb N_0^m\\ |\bm l|=n-q}}
            \ \sum_{\substack{\bm k\in\mathbb N_0^m\\ |\bm k|=q}}
            \prod_{i=1}^m \binom{l_i+k_i}{l_i}^2.
        \end{equation}
        Hence, defining
        \begin{equation}
        S_{n,q,m}
        \coloneqq
        \sum_{\substack{\bm l\in\mathbb N_0^m\\ |\bm l|=n-q}}
        \ \sum_{\substack{\bm k\in\mathbb N_0^m\\ |\bm k|=q}}
        \prod_{i=1}^m \binom{l_i+k_i}{l_i}^2,
        \end{equation}
        then we can write the expectation value as
        \begin{equation}
          \Tr \left[ \mathbb{S}_{q} 
            \mathbb{D}_{m,n} \right] =\frac{S_{n,q,m}}{|\mathcal S_{m,n}|\,\binom{n}{q}^2}.
        \end{equation}
        Now introduce the bivariate generating function from \Cref{prop:F_closed_form}:
        \begin{equation}
        F(x,y)
        \coloneqq
        \sum_{a, b = 0}^\infty \binom{a+b}{a}^2 x^a y^b .
        \end{equation}
        Then, we can write
        \begin{equation}
        S_{n,q,m}
        =
        [x^{\,n-q}y^q]\,
        F(x,y)^m .
        \end{equation}
        Moreover, using \Cref{prop:F_closed_form}, \(F(x,y)\) admits the closed form 
        \begin{equation}
        F(x,y)=\frac{1}{\sqrt{1-2x-2y+x^2-2xy+y^2}},
        \end{equation}
        Since we can rewrite the denominator as
        \begin{equation}
            1-2x-2y+x^2-2xy+y^2=(1-x-y)^2-4xy,
        \end{equation}
        we may write
        \begin{equation}
            F(x,y)^m
            =
            \bigl((1-x-y)^2-4xy\bigr)^{-m/2}
            =
            (1-x-y)^{-m}
            \left(1-\frac{4xy}{(1-x-y)^2}\right)^{-m/2}.
        \end{equation}
        Expanding the second factor gives
        \begin{equation}
            F(x,y)^m
            =
            \sum_{j\ge 0}\frac{\left(\frac m2\right)_j}{j!}\,
            4^j x^j y^j\,(1-x-y)^{-m-2j}.
        \end{equation}
        Therefore, we can rewrite $S_{n,q,m}$ as
        \begin{equation}
            S_{n,q,m}
            =
            \sum_{j=0}^q
            \frac{4^j\left(\frac m2\right)_j}{j!}
            [x^{n-q-j}y^{q-j}](1-x-y)^{-m-2j}.
        \end{equation}
        Using the fact that
        \(
        [x^a y^b](1-x-y)^{-s}=\frac{(s)_{a+b}}{a!\,b!}
        \)
        for \(a,b\in\mathbb N_0\), we obtain
        \begin{equation}
        S_{n,q,m}
        =
        \sum_{j=0}^q
        \frac{4^j\left(\frac m2\right)_j}{j!}\,
        \frac{(m+2j)_{n-2j}}{(n-q-j)!\,(q-j)!}.
        \end{equation}
        Next, by the duplication formula for rising Pochhammer symbols,
        \begin{equation}
        (m)_{2j}=4^j\left(\frac m2\right)_j\left(\frac{m+1}{2}\right)_j,
        \end{equation}
        and since
        \(
            (m)_n=(m)_{2j}(m+2j)_{n-2j},
        \)
        we get
        \begin{equation}
        4^j\left(\frac m2\right)_j(m+2j)_{n-2j}
        =
        \frac{(m)_n}{\left(\frac{m+1}{2}\right)_j}.
        \end{equation}
        As a result, we get the formula
        \begin{equation}
            S_{n,q,m}
            =
            \frac{(m)_n}{(n-q)!\,q!}
            \sum_{j=0}^q
            \frac{(-q)_j\,(q-n)_j}{\left(\frac{m+1}{2}\right)_j\,j!}.
        \end{equation}

        Using \Cref{prop:elementary_vandermonde} with
        \(
        a=\frac{m+1}{2}
        \) and \(
        b=q-n,
        \)
        and the identity
        \(
        \frac{(-q)_j}{j!}=(-1)^j\binom{q}{j},
        \)
        we obtain
        \begin{equation}
        \sum_{j=0}^q
        \frac{(-q)_j\,(q-n)_j}{\left(\frac{m+1}{2}\right)_j\,j!}
        =
        \frac{\left(\frac{m+1}{2}+n-q\right)_q}{\left(\frac{m+1}{2}\right)_q} = \frac{\left(\frac{m+1}{2}\right)_n}
         {\left(\frac{m+1}{2}\right)_q
          \left(\frac{m+1}{2}\right)_{n-q}}
        \end{equation}
        resulting in
        \begin{equation}
            S_{n,q,m}
            =
            \frac{(m)_n}{(n-q)!\,q!}
            \frac{\left(\frac{m+1}{2}\right)_n}
         {\left(\frac{m+1}{2}\right)_q
          \left(\frac{m+1}{2}\right)_{n-q}}.
        \end{equation}
        Finally, since
        $
        |\mathcal S_{m,n}|=\binom{n+m-1}{n}=\frac{(m)_n}{n!},
        $
        we can conclude that
        \begin{equation}
            \Tr \left[ \mathbb{S}_{q} 
                \mathbb{D}_{m,n} \right] 
            =
            \frac{S_{n,q,m}}{|\mathcal S_{m,n}|\,\binom{n}{q}^2}
            = 
            \frac{1}{\binom{n}{q}}\frac{\left(\frac{m+1}{2}\right)_n}
             {\left(\frac{m+1}{2}\right)_q
              \left(\frac{m+1}{2}\right)_{n-q}}
        \end{equation}
        as claimed.
    \end{proof}

    Next, we will utilize \Cref{prop:mean_purity_simple} to give a simple closed expression for $\Tr [\mathbb{P}_k  \; \mathbb{D}_{m, n}]$ (note, that the odd case is trivial by \Cref{remark:symmetry} from \Cref{app:expressing_the_projector}):
    \trpkdmn*
    \begin{proof}

        Considering the case of $k = 2r$, let us define
        \(
            a\coloneqq\frac{m+1}{2}
        \)
        for brevity.
        With the help of the formulas from \Cref{sec:framework}, we get that
        \begin{equation}
            \Tr [\mathbb{P}_k  \; \mathbb{D}_{m, n}]
            =
            \frac{2n-2k+1}{2n-k+1}\binom{n}{k}(a)_n
            \sum_{q=0}^n
            \frac{1}{(a)_q(a)_{n-q}}
            \sum_{l=\max(0,k-q)}^{\min(k,n-q)}
            (-1)^{k-l}
            \frac{\binom{k}{l}\binom{n-k}{l+q-k}}{\binom{2n-k}{n-l}}.
            \label{eq:Tk1}
        \end{equation}   
        Now let us substitute $k = 2r
        $ and $
            N\coloneqq n-2r
        $.
        Then, \eqref{eq:Tk1} becomes
        \begin{equation}
            \Tr [\mathbb{P}_{2r}  \; \mathbb{D}_{m, n}]
            =
            \frac{2n-4r+1}{2n-2r+1}\binom{n}{2r}(a)_n\,A_r,
        \end{equation}
        where
        \begin{equation}
            A_r
            \coloneqq
            \sum_{q=0}^n
            \frac{1}{(a)_q(a)_{n-q}}
            \sum_{l=\max(0,2r-q)}^{\min(2r,n-q)}
            (-1)^l
            \frac{\binom{2r}{l}\binom{N}{l+q-2r}}{\binom{2n-2r}{n-l}}.
        \end{equation}
        Setting $q=2r-l+s$,
        we observe that \(0\le l\le 2r\) and \(0\le s\le N\), and hence we can simplify $A_r$ as
        \begin{equation}
            A_r
            =
            \sum_{l=0}^{2r}
            \frac{(-1)^l\binom{2r}{l}}{\binom{2n-2r}{n-l}}
            \sum_{s=0}^{N}
            \frac{\binom{N}{s}}{(a)_{2r-l+s}(a)_{N+l-s}}.
            \label{eq:Ar_after_change}
        \end{equation}
        For the inner sum, we can use the identities
        $
        (a)_{2r-l+s}=(a)_{2r-l}(a+2r-l)_s
        $
        and
        \begin{equation}
            \frac{1}{(a)_{N+l-s}}
            =
            \frac{(-1)^s(1-a-N-l)_s}{(a)_{N+l}}
        \end{equation}
        to obtain
        \begin{align}
        \begin{split}
            \sum_{s=0}^{N}
            \frac{\binom{N}{s}}{(a)_{2r-l+s}(a)_{N+l-s}}
            &=
            \frac{1}{(a)_{2r-l}(a)_{N+l}}
            \sum_{s=0}^{N}
            \binom{N}{s}
            \frac{(-1)^s(1-a-N-l)_s}{(a+2r-l)_s}
            \\&=
            \frac{1}{(a)_{2r-l}(a)_{N+l}}
            \sum_{s=0}^{N}
            \frac{(-N)_s(1-a-N-l)_s}{(a+2r-l)_s\,s!}.
        \end{split}
        \end{align}
        Using the identity from \Cref{prop:elementary_vandermonde}, we get
        \begin{equation}
            \sum_{s=0}^{N}
            \frac{(-N)_s(b)_s}{(c)_s\,s!}
            =
            \frac{(c-b)_N}{(c)_N},
        \end{equation}
        and we obtain the simplification
        \begin{equation}
        \sum_{s=0}^{N}
        \frac{\binom{N}{s}}{(a)_{2r-l+s}(a)_{N+l-s}}
        =
        \frac{(2a+n-1)_N}{(a)_{n-l}(a)_{N+l}}.
        \end{equation}
        Substituting into \eqref{eq:Ar_after_change} gives
        \begin{equation}
        A_r
        =
        (2a+n-1)_N
        \sum_{l=0}^{2r}
        \frac{(-1)^l\binom{2r}{l}}
        {\binom{2n-2r}{n-l}(a)_{n-l}(a)_{N+l}}.
        \label{eq:Ar_one_sum}
        \end{equation}
        
        Again, we rewrite everything in terms of Pochhammer symbols using the formulas
        \begin{subequations}
        \begin{align}
            (-1)^l\binom{2r}{l}&=\frac{(-2r)_l}{l!},\\
            \frac{1}{\binom{2n-2r}{n-l}}
            &=
            \frac{1}{\binom{2n-2r}{n}}
            (-1)^l\frac{(N+1)_l}{(-n)_l},
        \end{align}
        \end{subequations}
        and
        \begin{equation}
            \frac{1}{(a)_{n-l}}
            =
            \frac{(-1)^l(1-a-n)_l}{(a)_n}, \qquad
            \frac{1}{(a)_{N+l}}
            =
            \frac{1}{(a)_N(a+N)_l},
        \end{equation}
        the one-sum formula for \(A_r\) becomes
        \begin{equation}
        A_r
        =
        \frac{(2a+n-1)_N}{\binom{2n-2r}{n}(a)_n(a)_N}
        \sum_{l=0}^{2r}
        \frac{(-2r)_l\,(N+1)_l\,(1-a-n)_l}
        {(-n)_l\,(a+N)_l\,l!}.
        \label{eq:Ar_finite_sum}
        \end{equation}
        Similarly, using \Cref{prop:pfaff_saalschutz}, the inner sum evaluates to
        \begin{equation}
            \sum_{l=0}^{2r}
            \frac{(-2r)_l\,(N+1)_l\,(1-a-n)_l}
            {(-n)_l\,(a+N)_l\,l!}
            =
            \frac{\binom{2r}{r}}{\binom{n}{r}}
            \frac{(a-1)_r}{(a+N)_r}.
            \label{eq:finite_sum_eval}
        \end{equation}
        Substituting \Cref{eq:finite_sum_eval} into \Cref{eq:Ar_finite_sum} gives
        \begin{equation}
            A_r
            =
            \frac{\binom{2r}{r}}{\binom{n}{r}}
            \frac{(2a+n-1)_N}{\binom{2n-2r}{n}(a)_n(a)_N}
            \frac{(a-1)_r}{(a+N)_r}.
        \end{equation}
        Since \(N=n-2r\) and
        $(a)_N(a+N)_r=(a)_{N+r}=(a)_{n-r}$,
        this simplifies to
        \begin{equation}
        A_r
        =
        \frac{\binom{2r}{r}}{\binom{n}{r}}
        \frac{(2a+n-1)_{n-2r}(a-1)_r}
        {\binom{2n-2r}{n}(a)_n(a)_{n-r}}.
        \end{equation}
        Now,
        using
        $
        \binom{n}{2r}\binom{2r}{r}
        =
        \binom{n-r}{r}\binom{n}{r}
        $
        and substituting $a=\frac{m+1}{2}$,
        we obtain
        \begin{equation}
            \Tr [\mathbb{P}_{2r}  \; \mathbb{D}_{m, n}]
            =
            \frac{2n-4r+1}{2n-2r+1}
            \frac{\binom{n-r}{r}}{\binom{2n-2r}{n}}
            \frac{(m+n)_{n-2r}\left(\frac{m-1}{2}\right)_r}
            {\left(\frac{m+1}{2}\right)_{n-r}}
            .
        \end{equation}
    \end{proof}

    To simplify the LXEB reference values in Boson Sampling, we will need to compute the overlap with the collision-free input state $\ket{\bm{n}_0}$.
    \begin{proposition}\label{prop:trace_pk_n0}
        For every integer \(r\) with \(0 \le r \le \lfloor n/2 \rfloor\),
        \begin{equation}\label{eq:tr_pk_n0_prop}
            \Tr\!\left[\mathbb{P}_{2r}\,\ketbra{\bm{n}_0}^{\otimes 2}\right]
            =
            2^{\,n-2r}\,\frac{2n-4r+1}{2n-2r+1}\,
            \frac{\binom{n}{r}}{\binom{2n-2r}{\,n-r\,}}.
        \end{equation}
    \end{proposition}

    \begin{proof}
        Using \Cref{lem:explicit_P_k} and the fact that
        \begin{equation}
            \Tr\!\left[\mathbb{S}_q\,\ketbra{\bm{n}_0}^{\otimes 2}\right]
            = \frac{1}{\binom{n}{q}},
        \end{equation}
        we obtain
        \begin{align}
            \Tr\!\left[\mathbb{P}_{2r}\,\ketbra{\bm{n}_0}^{\otimes 2}\right]
            =
            \frac{2n-4r+1}{2n-2r+1}\binom{n}{2r}
            \sum_{q=0}^{n}
            \sum_{l=\max(0,2r-q)}^{\min(2r,n-q)}
            (-1)^l
            \frac{\binom{2r}{l}\binom{n-2r}{l+q-2r}}{\binom{2n-2r}{n-l}}
            =
            2^{\,n-2r}\,\frac{2n-4r+1}{2n-2r+1}\binom{n}{2r}\,\Sigma_{n,r},
            \label{eq:trace_reduced}
        \end{align}
        where
        \begin{equation}\label{eq:def_sigma_nr}
            \Sigma_{n,r}
            :=
            \sum_{l=0}^{2r}
            (-1)^l \binom{2r}{l}\binom{2n-2r}{n-l}^{-1}.
        \end{equation}
        Indeed, for fixed \(l\),
        \begin{equation}\label{eq:q_sum}
            \sum_{q=0}^{n}\binom{n-2r}{l+q-2r}
            = \sum_{j=0}^{n-2r}\binom{n-2r}{j}
            = 2^{\,n-2r}.
        \end{equation}
        Next, using
        \begin{equation}\label{eq:binom_reciprocal}
            \binom{2n-2r}{n-l}^{-1}
            =
            \binom{2n-2r}{n}^{-1}
            (-1)^l \frac{(n-2r+1)_l}{(-n)_l},
        \end{equation}
        we can rewrite \eqref{eq:def_sigma_nr} as
        \begin{equation}\label{eq:sigma_hyper}
            \Sigma_{n,r}
            =
            \binom{2n-2r}{n}^{-1}
            \sum_{l=0}^{2r}
            \frac{(-2r)_l (n-2r+1)_l}{(-n)_l\,l!}\,(-1)^l
        \end{equation}
        Now, using \Cref{prop:elementary_vandermonde} gives
        \begin{equation}\label{eq:vandermonde_eval}
            \Sigma_{n,r}
            =
            \binom{2n-2r}{n}^{-1}
            \frac{\left(\tfrac12-r\right)_r}{(-n)_r}.
        \end{equation}
        Substituting \eqref{eq:vandermonde_eval} into \eqref{eq:sigma_hyper}, and using
        \begin{equation}\label{eq:pochhammer_halfint}
            \left(\tfrac12-r\right)_r = (-1)^r \frac{(2r)!}{4^r r!},
            \qquad
            (-n)_r = (-1)^r \frac{n!}{(n-r)!},
        \end{equation}
        we obtain
        \begin{equation}\label{eq:sigma_closed_form}
            \Sigma_{n,r}
            =
            \frac{(2r)!(n-r)!(n-2r)!}{r!(2n-2r)!}
            =
            \frac{\binom{n}{r}}{\binom{n}{2r}\binom{2n-2r}{n-r}}.
        \end{equation}
        Finally, inserting \eqref{eq:sigma_closed_form} into \eqref{eq:trace_reduced} yields \eqref{eq:tr_pk_n0}.
    \end{proof}

    Moreover, for obtaining simplified LXEB reference values for the case of Gaussian Boson Sampling as well, we need to use the following formula for the overlap of a projector $\mathbb{P}_k$ with a uniformly squeezed state~\cite{second_quantized_unpublished}:
    \begin{theorem}\label{thm:P_k_uniform_squeezed_closed}
        Let \(n=2N\), and let $\rho$ be the normalized $n$-particle restriction of the product of $d$ uniform single-mode squeezed states.
        Then, we can write
        \begin{equation}
            \Tr\!\bigl[\mathbb{P}_k \rho^{\otimes 2}\bigr] = 0
            \qquad\text{for odd }k,
        \end{equation}
        and for \(k=2j\), \(j=0,\dots,N\),
        \begin{equation}
        \label{eq:Pk_uniform_squeezed_closed}
            \Tr\!\bigl[\mathbb{P}_{2j}\rho^{\otimes 2}\bigr]
            =
            \frac{2^{2j+1} N!^2}{j!}\,
            \binom{2N-2j}{N-j}^{2}\,
            \frac{(2N-j+1)!\,(4N-4j+1)}{(4N-2j+2)!}\,
            \frac{\left(\frac{d-1}{2}\right)_j\left(\frac d2+N\right)_{N-j}}{\left(\frac d2\right)_N}.
        \end{equation}
    \end{theorem}

\section{Efficient algorithm for computing expectation values of bosonic swap operators}\label{app:proof_bosonic_swap_efficient}
    For computing the LXEB reference value for a wider range of photonic quantum advantage experiments efficiently, we rely on the following fact:
     \begin{restatable}{proposition}{polynomialmachinery} \label{prop:polynomial_machinery}
         Consider a polynomials $P_i \in \mathbb{C}[x_1, \dots, x_k]$, and fix a multi-index $\bm{\alpha} \geq \bm{0}$ and $m \geq 1$. 
        The coefficient
        \begin{equation}
            [\bm{x}^{\bm{\alpha}}]\prod_{i=1}^m P_i(\bm{x})
            =[x_1^{\alpha_1}\cdots x_k^{\alpha_k}]\prod_{i=1}^m P_i(x_1, \dots, x_k)
        \end{equation}
        can be computed in $O(m \, M \log M )$ time using $O(M)$ memory where $M \coloneqq \prod_{i=1}^{k}(\alpha_i+1)$.
        Here, we assume that the polynomials $P_i$ are given explicitly and that
        only coefficients up to total degree $\bm{\alpha}$ are required.
    \end{restatable}
    \begin{proof}
        To extract the coefficient corresponding to $x_1^{\alpha_1}\cdots x_k^{\alpha_k}$, we only need to work in the ring $R$ defined as
        \begin{equation}
            R \coloneqq \mathbb{C}[x_1,\dots,x_k]/I,
        \end{equation}
        where $I$ is the truncation ideal
        $
            I \coloneqq \langle x_1^{\alpha_1+1},\dots,x_k^{\alpha_k+1}\rangle
            \triangleleft \mathbb{C}[x_1,\dots,x_k].
        $
        Roughly speaking, arithmetic in $R$ is just multivariate polynomial arithmetic in which every monomial
        $x_1^{\beta_1}\cdots x_k^{\beta_k}$ with $\beta_i>\alpha_i$ for some $i$ is set to zero.
        Let $Q_i \in R$ denote the image of $P_i$ under the quotient map, i.e., $Q_i\equiv P_i \pmod I$.
        Then
        \begin{equation}\label{eq:coef_trunc_eq}
        [\bm{x}^{\bm{\alpha}}]\prod_{i=1}^m P_i(\bm{x})
        =
        [\bm{x}^{\bm{\alpha}}]\prod_{i=1}^m Q_i(\bm{x}).
        \end{equation}
        Indeed, we can write $P_i = Q_i + H_i$ with $H_i \in I$ and expand $\prod_{i=1}^m (Q_i + H_i)$. Every term
        containing at least one factor $H_i$ lies in $I$ and thus it is $0$ in $R$.
    
        In the quotient ring $R$, a polynomial is uniquely determined by the coefficients of monomials
        $\bm{x}^{\bm{\beta}}=x_1^{\beta_1}\cdots x_k^{\beta_k}$ with $0\le \beta_i\le \alpha_i$.
        Hence, storing an element of $R$ requires
        $
            M \coloneqq \prod_{i=1}^{k}(\alpha_i+1)
        $
        coefficients, i.e., $O(M)$ memory.
        To compute the product $\prod_{i=1}^m Q_i$ in $R$, we perform all polynomial multiplications in the
        quotient ring $R$ (equivalently, we truncate modulo $I$ after each multiplication).
        This algorithm performs $O(m)$ multiplications in $R$ and stores only $O(1)$ elements of $R$,
        so the memory requirement is $O(M)$.
        The execution time is
        $
            O\bigl(m \; T_{\mathrm{mul}}(M)\bigr),
        $
        where $T_{\mathrm{mul}}(M)$ denotes the cost of multiplying two elements of $R$.
        In particular, naive dense multiplication gives $T_{\mathrm{mul}}(M)=O(M^2)$, whereas faster
        convolution-based implementations can reduce this to $O(M\log M)$ (see, e.g., Chapter 4 from Ref.~\cite{czapor_1992}).
        Finally, after computing $\prod_{i=1}^m Q_i \in R$, the desired coefficient is obtained by reading off the
        coefficient of $\bm{x}^{\bm{\alpha}}$ using \Cref{eq:coef_trunc_eq}.
    \end{proof}
    As a remark, in the special case where all factors are identical, i.e., $Q_i \equiv Q$ and we need $Q^m$, we use
    binary (square-and-multiply) exponentiation in $R$. This strategy goes as follows: initialize $E \leftarrow 1$ and $A \leftarrow Q$, and while $m>0$, repeat
    \begin{equation}
        \text{if $m$ is odd then } E \leftarrow E\cdot A;\quad
        A \leftarrow A\cdot A;\quad
        m \leftarrow \left\lfloor \frac{m}{2}\right\rfloor,
    \end{equation}
    where every multiplication is carried out in $R$, i.e., followed by discarding
    all monomials $x_1^{\beta_1}\cdots x_k^{\beta_k}$ with $\beta_i>\alpha_i$ for some $i$.
    Upon termination, $E=Q^m$ in $R$.
    This algorithm performs $O(\log m)$ multiplications in $R$ instead of $O(m)$.

    \medskip
    Using \Cref{prop:polynomial_machinery}, we can prove the following statement:
    \swapexpval*
    \begin{proof}
        We can write
        \begin{equation}
            \rho = \bigotimes_{a=1}^m \sigma^{(a)}, \qquad         \sigma^{(a)} = \sum_{n, m \geq 0} c_{n,m}^{(a)} \ketbra{n}{m}.
        \end{equation}
        Then, we can write its $n$-particle component as
        \begin{equation}
            \rho |_{n} = \sum_{\substack{\bm{n}\geq 0 \\ |\bm{n}| = n}}
            \sum_{\substack{\bm{m}\geq 0 \\ |\bm{m}|  = n}}
            \prod_{a=1}^m c^{(a)}_{n_a, m_a} \ketbra{\bm{n}}{\bm{m}}.
        \end{equation}
        Using the algorithm in \Cref{prop:polynomial_machinery}, we can calculate its norm efficiently using
        \begin{equation}
            \Tr [\rho |_{n}] = \sum_{\substack{\bm{n}\geq 0 \\ |\bm{n}| = n}}
            \prod_{a=1}^m c^{(a)}_{n_a, n_a} = [x^n] \prod_{a=1}^m P_{(a)}(x),
        \end{equation}
        where $P_{(a)}(x) \coloneqq \sum_{k=0}^n c_{k, k}^{(a)} x^k$. We also remark that this method yields an efficient algorithm for computing the global photon-number detection probabilities~\cite{Avagyan_2023}.
        
        In the final result, we can divide by $ \Tr [\rho |_{n}]^2$, and we could focus on $\rho |_n$ in the following.
        Using \Cref{eq:bosonic_swap_boxed}, the expectation value of the $q$-fold bosonic swap operator can be expressed as
        \begin{align}\begin{split}
            \Tr [ \mathbb{S}_q \rho |_n^{\otimes 2}]
            = 
            \frac{1}{\binom{n}{q}^2} 
            \hspace{-1em}
            \sum_{\substack{\bm{q}, \bm{r}  \geq \bm{0} \\ |\bm{q}| = |\bm{r}|  = q}}
            \hspace{-1em}
            A(\bm{r}, \bm{q}) A(\bm{q}, \bm{r}),
        \end{split}
        \end{align}
        where we denoted
        \begin{equation}
            A(\bm{r}, \bm{q}) \coloneqq 
            \Tr [\frac{(\bm{a}^\dagger)^{\bm{r}}}{\sqrt{\bm{r}!}} \frac{\bm{a}^{\bm{q}}}{\sqrt{\bm{q}!}}\, \rho |_n
            ].
        \end{equation}
        Furthermore, we directly calculate
        \begin{equation}
            A(\bm{r}, \bm{q})
            = 
            \hspace{-1.5em} \sum_{\substack{\bm{n}, \bm{m}\geq 0 \\ |\bm{n}| = |\bm{m}| = n}}
            \hspace{-1.5em} c_{\bm{n}, \bm{m}} 
            \Tr [\frac{(\bm{a}^\dagger)^{\bm{r}}}{\sqrt{\bm{r}!}} \frac{\bm{a}^{\bm{q}}}{\sqrt{\bm{q}!}} \ketbra{\bm{n}}{\bm{m}} ]
            = 
             \sum_{\substack{\bm{k}\geq \bm{0} \\ |\bm{k}| = n-q}}
             c_{\bm{k} + \bm{q}, \bm{k} +\bm{r}} 
            \sqrt{\binom{\bm{k} + \bm{q}}{\bm{k}} \binom{\bm{k} + \bm{r}}{\bm{k}}},
        \end{equation}
        where we denoted $c_{\bm{n}, \bm{m}} \coloneqq \prod_{a=1}^m c^{(a)}_{n_a, m_a} $ for brevity and used that
        \begin{equation}
            \Tr [(\bm{a}^\dagger)^{\bm{r}} \bm{a}^{\bm{q}} \ketbra{\bm{n}}{\bm{m}} ]
            = \frac{\sqrt{\bm{n}! \bm{m}!}}{(\bm{n}-\bm{q})!} \delta_{\bm{n}-\bm{q}, \bm{m}-\bm{r}} \mathbbm{1}_{\bm{n} \geq \bm{q}},
        \end{equation}
        where $\mathbbm{1}_{\bm{n} \geq \bm{q}}$ denotes the indicator function yielding $1$ when $\bm{n} \geq \bm{q}$, otherwise $0$. Here, $\bm{n} \geq \bm{q}$ is understood elementwise, i.e., we write \(\bm n \ge \bm q\) when \(n_i \ge q_i\) for every component \(i\).
        As a result, we can write
        \begin{equation}
             \Tr [ \mathbb{S}_q \rho |_n^{\otimes 2}]
             =
             \frac{1}{\binom{n}{q}^2} 
            \sum_{\substack{\bm{q}, \bm{r}, \bm{k}, \bm{l}  \geq \bm{0} \\ |\bm{q}| = |\bm{r}|  = q \\ |\bm{k}| = |\bm{l}| = n-q}} 
            \prod_{a=1}^m K^{(a)}_{k_a, l_a, q_a, r_a},
        \end{equation}
        where we have defined
        \begin{subequations}
        \begin{align}
            K^{(a)}_{k_a, l_a, q_a, r_a} &\coloneqq L^{(a)}_{k_a, q_a, r_a} L^{(a)}_{l_a, r_a,  q_a}, \\
            L^{(a)}_{x, y, z} &\coloneqq
            c^{(a)}_{x+y, x+z}
            \sqrt{\binom{x+y}{x} \binom{x+z}{z}}
        \end{align}
        \end{subequations}
        The key to give an efficient algorithm for computing $\Tr [\mathbb{S}_q (\rho_{(n)})^{\otimes 2}]$ is to define a 4-variable polynomial with $K^{(a)}_{k_a, l_a, q_a, r_a}$ as coefficients, where each variable corresponds to an index in $K^{(a)}_{k_a, l_a, q_a, r_a}$. More concretely, we define
        \begin{equation}
            K^{(a)}(t, u, v, w) \! \coloneqq \!
            \sum_{j=0}^{n-q} \sum_{k=0}^{n-q} \sum_{l=0}^{q} \sum_{m=0}^{q} 
            \! K^{(a)}_{j,k,l,m} t^j u^k v^l w^m.
        \end{equation}
        Using this polynomial, we can easily rewrite $ \Tr [ \mathbb{S}_q \rho|_n^{\otimes 2}]$ using the coefficient of $\prod_a K^{(a)}$ corresponding to the term $t^{n-q} u^{n-q} v^q w^q$, yielding
        \begin{equation} \label{eq:swap_expval_algorithm_form}
             \Tr [ \mathbb{S}_q \rho_{(n)}^{\otimes 2}]
             \!=\!
             \frac{1}{\binom{n}{q}^2\Tr [\rho |_{n}]^2}[t^{n-q} u^{n-q} v^q w^q] \prod_{a=1}^m K^{(a)}(t, u, v, w) .
        \end{equation}
        One can compute this quantity using $O((n-q+1)^2 (q+1)^2)$ memory in $O(m (n-q+1)^2 (q+1)^2 \log( (n-q+1) (q+1) )$ time according to \Cref{prop:polynomial_machinery}, implying the same for $\Tr [ \mathbb{S}_q \rho_{(n)}^{\otimes 2}]$.
    \end{proof}

\section{Bounds on the anticoncentration scores for Boson Sampling, Scattershot Boson Sampling and Gaussian Boson Sampling}\label{app:bounds}

    \subsection{Proof of \Cref{eq:ACaverage}}\label{app:proof_ac_bound}
        \begin{restatable}{lemma}{acbound}\label{lem:ac}
            We can upper bound the anticoncentration score $\mathrm{AC}(m, n)$ by
            \begin{equation}
                \mathrm{AC}(m, n)\le \const
            \end{equation}
            for all $m\ge n\ge 1$ and $m \geq 2$.
        \end{restatable}
                    
        \begin{proof}
            We start from the form derived in \Cref{prop:bs_lxeref_simplified} and first dispose of the case \(n=1\). Then only \(r=0\) contributes, and
            \begin{align}
                \mathrm{AC}(m,1)=\binom{m}{1}\,2\,\frac{1!}{\binom{2}{1}\left(\frac{m+1}{2}\right)}
                =\frac{2m}{m+1}
                \le 2
                \le 8m+9.
            \end{align}
            Hence it remains to consider \(n\ge 2\). In particular, since \(m\ge n\), we have
            $M\coloneqq m+n\ge 4$.
            Set
            \begin{equation}
                s\coloneqq n-2r\in\{0,1,\dots,n\},\qquad s\equiv n \pmod 2.
            \end{equation}
            Using
            \begin{equation}
                \binom{n+m-1}{n}\,n!=(m)_n
            \end{equation}
            and Proposition~\ref{prop:duplication_formula}, one can verify that
            \begin{equation}
                \mathrm{AC}(m,n)
                =
                \sum_{\substack{0\le s\le n\\ s\equiv n \!\!\!\!\pmod 2}}
                U_{m,n}(s),
            \end{equation}
            where
            \begin{equation}
                U_{m,n}(s)
                =
                \frac{2s+1}{n+s+1}\,
                \frac{
                \Gamma\!\left(\frac{n-s+1}{2}\right)
                \Gamma\!\left(\frac{n+s+2}{2}\right)
                }{
                \Gamma\!\left(\frac{n-s+2}{2}\right)
                \Gamma\!\left(\frac{n+s+1}{2}\right)
                }
                \prod_{j=0}^{s-1}\frac{M-s+j}{M-s+2j+1}.
            \end{equation}
            For \(1\le s\le n-1\), Proposition~\ref{prop:gautschi} gives
            \begin{equation}
                \frac{\Gamma\!\left(\frac{n-s+1}{2}\right)}
                {\Gamma\!\left(\frac{n-s+2}{2}\right)}
                <
                \sqrt{\frac{2}{n-s}}, \qquad                 \frac{\Gamma\!\left(\frac{n+s+2}{2}\right)}
                {\Gamma\!\left(\frac{n+s+1}{2}\right)}
                <
                \sqrt{\frac{n+s+1}{2}}.
            \end{equation}
            Therefore, we get that
            \begin{equation}\label{eq:Umn_basic_gamma_bound}
                U_{m,n}(s)
                \le
                \frac{2s+1}{\sqrt{(n+s+1)(n-s)}}\,P_s,
            \end{equation}
            where $P_s$ is the product
            \begin{equation}
                P_s:=\prod_{j=0}^{s-1}\frac{M-s+j}{M-s+2j+1}.
            \end{equation}
            
            We next bound \(P_s\). The numerator terms are
            \(M-s,\ M-s+1,\ \dots,\ M-1,\)
            whose arithmetic mean equals \(M-\frac{s+1}{2}\), while the denominator terms are
            \(M-s+1,\ M-s+3,\ \dots,\ M+s-1,\)
            whose arithmetic mean equals \(M\). By the arithmetic-geometric mean inequality,
            \begin{equation}
                P_s
                \le
                \left(\frac{M-\frac{s+1}{2}}{M}\right)^s
                =
                \left(1-\frac{s+1}{2M}\right)^s
                \le
                e^{-\,\frac{s(s+1)}{2M}}
                \le
                e^{-\,\frac{s^2}{2M}}.
            \end{equation}
            Combining this with \eqref{eq:Umn_basic_gamma_bound}, we obtain
            \begin{equation}\label{eq:Umn_master_bound}
                U_{m,n}(s)
                \le
                \frac{2s+1}{\sqrt{(n+s+1)(n-s)}}\,e^{-\,\frac{s^2}{2M}}
                \qquad (1\le s\le n-1).
            \end{equation}
            We now split the sum into several ranges.
            
            \medskip
            
            \noindent\textbf{Step 1: the smallest admissible value of \(s\).}
            If \(n\) is even, then \(s=0\) occurs, and directly
            \begin{equation}
                U_{m,n}(0)=\frac{1}{n+1}\le 1.
            \end{equation}
            If \(n\) is odd, then the smallest admissible value is \(s=1\). Since \(n\ge 3\) in this case, \eqref{eq:Umn_master_bound} yields
            \begin{equation}
                U_{m,n}(1)
                \le
                \frac{3}{\sqrt{(n+2)(n-1)}}\,e^{-1/(2M)}
                \le
                \frac{3}{\sqrt{10}}
                <1.
            \end{equation}
            Thus, in all cases, the contribution of the smallest admissible \(s\) is at most \(1\).
            
            \medskip
            
            \noindent\textbf{Step 2: the range \(2\le s\le \frac{3n}{4}\).}
            For such \(s\),
            \begin{equation}
                (n+s+1)(n-s)
                \ge
                \left(\frac{7n}{4}\right)\left(\frac{n}{4}\right)
                =
                \frac{7n^2}{16},
            \end{equation}
            hence \eqref{eq:Umn_master_bound} implies
            \begin{equation}
                U_{m,n}(s)
                \le
                \frac{4}{\sqrt7}\frac{2s+1}{n}\,e^{-s^2/(2M)}.
            \end{equation}
            Since \(s\ge 2\), we have \(2s+1\le \frac52 s\), so
            \begin{equation}
                U_{m,n}(s)
                \le
                \frac{10}{\sqrt7}\frac{s}{n}\,e^{-s^2/(2M)}.
            \end{equation}
            Therefore
            \begin{equation}\label{eq:medium_range_start}
                \sum_{2\le s\le 3n/4} U_{m,n}(s)
                \le
                \frac{10}{\sqrt7\,n}
                \sum_{s=2}^{\infty} s\,e^{-s^2/(2M)}.
            \end{equation}
            To bound this series, let \(f(x)=x e^{-x^2/(2M)}\). Using that \(f\ge 0\) and that $f$ is increasing on $[1, \sqrt{M}]$ and decreasing on $[\sqrt{M}, \infty)$, we can bound every term in $\sum_{s=2}^{\infty} f(s)$ by an integral over a neighboring unit interval, except one term adjacent to the maximum. Hence, 
            \begin{equation}
                \sum_{s=2}^{\infty} f(s)
                \le
                \int_1^{\infty} f(x)\,dx+\sup_{x\ge 1} f(x).
            \end{equation}
            Since $
                \int_1^{\infty} x e^{-x^2/(2M)}\,dx
                =
                M e^{-1/(2M)}
            $
            and
            $
                \sup_{x\ge 1} x e^{-x^2/(2M)}
                \le
                \sqrt{M / e},
            $
            we can write
            \begin{equation}
                \sum_{s=2}^{\infty} s\,e^{-s^2/(2M)}
                \le
                M\left(e^{-1/(2M)}+\frac{1}{\sqrt{eM}}\right).
            \end{equation}
            The function
            \begin{equation}
                g(M)\coloneqq e^{-1/(2M)}+\frac{1}{\sqrt{eM}}
            \end{equation}
            is decreasing for \(M\ge 1\), because
            \begin{equation}
                g'(M)
                =
                \frac{e^{-1/(2M)}}{2M^2}
                -
                \frac{e^{-1/2}}{2M^{3/2}}
                \le 0
            \end{equation}
            for \(M\ge 1\). Since \(M\ge 4\), we get
            \begin{equation}
                g(M)\le g(4)<\frac65.
            \end{equation}
            Thus
            \begin{equation}
                \sum_{s=2}^{\infty} s\,e^{-s^2/(2M)}
                \le
                \frac65\,M.
            \end{equation}
            Substituting into \eqref{eq:medium_range_start}, we obtain
            \begin{equation}\label{eq:medium_range_final}
                \sum_{2\le s\le 3n/4} U_{m,n}(s)
                \le
                \frac{12}{\sqrt7}\,\frac{M}{n}.
            \end{equation}
            
            \medskip
            
            \noindent\textbf{Step 3: the range \(\frac{3n}{4}<s<n\).}
            Write \(u:=n-s\), so that \(1\le u< n/4\). By \eqref{eq:Umn_master_bound}, we can write that
            \begin{equation}
                U_{m,n}(s)
                \le
                \frac{2s+1}{\sqrt{(n+s+1)u}}\,e^{-s^2/(2M)}.
            \end{equation}
            Also,
            \begin{equation}
                (2s+1)^2\le (2n+1)(n+s+1),
            \end{equation}
            because
            \begin{equation}
                (2n+1)(n+s+1)-(2s+1)^2
                =
                (n-s)(2n+4s+3)\ge 0.
            \end{equation}
            Hence
            \begin{equation}
                U_{m,n}(s)
                \le
                \sqrt{\frac{2n+1}{u}}\,e^{-s^2/(2M)}.
            \end{equation}
            Since \(n\ge 2\), we have \(2n+1\le \frac52 n\), and since \(s>\frac{3n}{4}\), we have \(s^2\ge \frac{9n^2}{16}\). Therefore
            \begin{equation}
                U_{m,n}(s)
                \le
                \sqrt{\frac{5n}{2u}}\,e^{-9n^2/(32M)}.
            \end{equation}
            Summing over \(u=1,\dots,\lfloor n/4\rfloor\) and using Proposition~\ref{prop:inverse_sqrt_ub},
            \begin{align}\begin{split}
                \sum_{3n/4<s<n} U_{m,n}(s)
                \le
                \sqrt{\frac{5n}{2}}\,e^{-9n^2/(32M)}
                \sum_{u=1}^{\lfloor n/4\rfloor}\frac{1}{\sqrt{u}}\le
                \sqrt{\frac{5n}{2}}\,e^{-9n^2/(32M)} \cdot 2\sqrt{\frac{n}{4}}=
                \sqrt{\frac52}\,n\,e^{-9n^2/(32M)}.
            \end{split}\end{align}
            Let
            \begin{equation}
                x:=\frac{n^2}{M}.
            \end{equation}
            Then \(n=\frac{M}{n}x\), so
            \begin{equation}
                \sqrt{\frac52}\,n\,e^{-9n^2/(32M)}
                =
                \sqrt{\frac52}\,\frac{M}{n}\,x e^{-9x/32}.
            \end{equation}
            Since \(\sup_{x>0} x e^{-9x/32}=\frac{32}{9e}\), it follows that
            \begin{equation}\label{eq:large_range_final}
                \sum_{3n/4<s<n} U_{m,n}(s)
                \le
                \frac{32}{9e}\sqrt{\frac52}\,\frac{M}{n}.
            \end{equation}
            
            \medskip
            
            \noindent\textbf{Step 4: the endpoint \(s=n\).}
            This corresponds to \(r=0\), and directly from the definition,
            \begin{equation}
                U_{m,n}(n)
                =
                \frac{4^n}{\binom{2n}{n}}
                \prod_{j=0}^{n-1}\frac{m+j}{m+2j+1}.
            \end{equation}
            Using Proposition~\ref{lem:bin_bounds} with \(N=2n\) and \(k=n\), we get
            \begin{equation}
                \binom{2n}{n}\ge \frac{4^n}{2\sqrt{n}},
            \end{equation}
            hence
            \begin{equation}
                \frac{4^n}{\binom{2n}{n}}\le 2\sqrt{n}.
            \end{equation}
            Also, from previous discussions we already know that
            \begin{equation}
                \prod_{j=0}^{n-1}\frac{m+j}{m+2j+1}
                \le
                e^{-n^2/(2M)},
            \end{equation}
            which yields
            $
                U_{m,n}(n)\le 2\sqrt{n}\,e^{-n^2/(2M)}.
            $
            With \(x=n^2/M\), and using \(M/n\ge 1\), we have
            \begin{equation}
                2\sqrt{n}\,e^{-x/2}
                =
                2\sqrt{\frac{M}{n}}\sqrt{x}\,e^{-x/2}
                \le
                2\frac{M}{n}\sqrt{x}\,e^{-x/2}.
            \end{equation}
            Since \(\sup_{x>0} \sqrt{x}\,e^{-x/2}=
            e^{-1/2}\), it follows that
            \begin{equation}\label{eq:endpoint_final}
                U_{m,n}(n)\le 2e^{-1/2}\,\frac{M}{n}.
            \end{equation}
            
            Collecting the contribution of the smallest admissible \(s\) which can be upper bounded by $1$, together with
            \eqref{eq:medium_range_final}, \eqref{eq:large_range_final}, and \eqref{eq:endpoint_final}, we obtain
            \begin{equation}
                \mathrm{AC}(m,n)
                \le
                1+
                \left(
                \frac{12}{\sqrt7}
                +
                \frac{32}{9e}\sqrt{\frac52}
                +
                2e^{-1/2}
                \right)\frac{M}{n}.
            \end{equation}
            Set
            \begin{equation}
                A:=
                \frac{12}{\sqrt7}
                +
                \frac{32}{9e}\sqrt{\frac52}
                +
                2e^{-1/2}.
            \end{equation}
            Numerically,
            \begin{equation}
                A<7.817.
            \end{equation}
            Since \(M/n = m/n + 1\), we conclude that
            \begin{equation}
                \mathrm{AC}(m,n)
                \le
                A\frac{m}{n}+(A+1)
                \le
                8\,\frac{m}{n}+9.
            \end{equation}
            This proves the claim.
            \end{proof}

    \subsection{Asymptotic behavior of $\mathrm{AC}(m, n)$} \label{app:ac_asymptotic}
        \begin{lemma}\label{lemma:ac_asymptotic}
            Let \((m_n)_{n\ge1}\) be a sequence of positive integers such that
            \(
            \frac{m_n}{n}\to \alpha\in[0,\infty).
            \)
            Then
            \begin{equation}
            \mathrm{AC}(m_n,n)\to 1+\alpha.
            \end{equation}
        \end{lemma}
        \begin{proof}
            Let
            \(
            \alpha_n:=\frac{m_n}{n},
            \) so that \(
            \alpha_n\to\alpha\in[0,\infty)
            \)
            Choose \(A>\alpha\). Then \(\alpha_n\in[0,A]\) for all sufficiently large \(n\).
            As in \Cref{app:proof_ac_bound}, let us cast $\mathrm{AC}(m_n,n)$ in the form of
            \begin{equation}
            \mathrm{AC}(m_n,n)=\sum_{r=0}^{\lfloor n/2\rfloor} T_{n,m_n,r},
            \end{equation}
            where
            \begin{equation}
            T_{n,m,r}
            =
            \frac{
            \sqrt{\pi}\,(2n-4r+1)\,\Gamma(n+m)\,
            \Gamma\!\left(r+\frac12\right)\,
            \Gamma(n-r+1)
            }{
            2^{n+m}\,
            \Gamma\!\left(\frac m2+r\right)\,
            \Gamma(r+1)\,
            \Gamma\!\left(n-r+\frac32\right)\,
            \Gamma\!\left(\frac m2+n-r+\frac12\right)
            }.
            \end{equation}
            For practical purposes, let us reparametrize the sum and set
            \begin{equation}
                s:=n-2r,
                \qquad r=\frac{n-s}{2},
            \end{equation}
            so that \(s\ge 0\) and \(s\equiv n \pmod 2\).
            The proof strategy is as follows: we distribute the sum into two regimes, and show that the latter contribution is negligible. More concretely, let us take $K > 0$, and let us denote
            \begin{align}
                0 \leq s \leq K \sqrt{n} &\qquad (\text{\textbf{Main contribution}}),\\
                 K \sqrt{n} < s \leq n &\qquad (\text{\textbf{Tail}}).
            \end{align}
            The tail can be made arbitrarily small by choosing $K$ large, uniformly in large $n$.
        
            \medskip
            \noindent \textbf{1. Main contribution.}
             Using the duplication formula from \Cref{prop:duplication_formula}, we may write
            \begin{equation}
            T_{n,m_n,\frac{n-s}{2}}
            =
            \frac{2s+1}{2}\,R_{n,\alpha_n}(s)\,U_n(s),
            \end{equation}
            where
            \begin{equation}
            R_{n,\beta}(s)
            :=
            \frac{
            \Gamma\!\left(\frac{(1+\beta)n}{2}\right)
            }{
            \Gamma\!\left(\frac{(1+\beta)n-s}{2}\right)
            }
            \frac{
            \Gamma\!\left(\frac{(1+\beta)n+1}{2}\right)
            }{
            \Gamma\!\left(\frac{(1+\beta)n+s+1}{2}\right)
            },
            \end{equation}
            and
            \begin{equation}
            U_n(s)
            \coloneqq
            \frac{
            \Gamma\!\left(\frac{n-s+1}{2}\right)
            }{
            \Gamma\!\left(\frac{n-s+2}{2}\right)
            }
            \frac{
            \Gamma\!\left(\frac{n+s+2}{2}\right)
            }{
            \Gamma\!\left(\frac{n+s+3}{2}\right)
            }.
            \end{equation}
            By the standard
            gamma-ratio estimate from \Cref{prop:gamma_ratio_asymptotic}, we can write
            \begin{equation}
            \frac{\Gamma(t+a)}{\Gamma(t+b)}=t^{a-b}(1+o(1))
            \qquad (t\to\infty),
            \end{equation}
            hence, we have
            \begin{equation}
            U_n(s)=\frac{2}{n}(1+o(1))
            \end{equation}
            uniformly for \(0\le s\le K\sqrt n\).
            Next, for \(\beta\in[0,A]\), let
            \begin{equation}
            z:=\frac{(1+\beta)n}{2},
            \qquad
            u:=\frac{s}{2}.
            \end{equation}
            Using Stirling's approximation from \Cref{prop:stirling_bounds}, we get
            \begin{equation}
                R_{n,\beta}(s)
                =
                \exp\!\left(-\frac{s^2}{2(1+\beta)n}\right)(1+o(1))
            \end{equation}
            uniformly for \(0\le s\le K\sqrt n\) and \(\beta\in[0,A]\). Substituting
            \(\beta=\alpha_n\), we obtain
            \begin{equation}\label{eq:local-T}
            T_{n,m_n,\frac{n-s}{2}}
            =
            \frac{2s+1}{n}
            \exp\!\left(-\frac{s^2}{2(1+\alpha_n)n}\right)(1+o(1)),
            \end{equation}
            uniformly for \(0\le s\le K\sqrt n\). Before we proceed, let us discuss the tail bound.
        
            \medskip
        
            \noindent \textbf{2. Tail.}
            We next show that the complementary range is negligible. For \(\beta\in[0,A]\),
            define
            \begin{equation}
            \phi_\beta(x)
            :=
            \log\!\frac{(1+\beta)^{1+\beta}}
            {2^{1+\beta}\left(\frac\beta2+x\right)^{\frac\beta2+x}
            \left(1+\frac\beta2-x\right)^{1+\frac\beta2-x}},
            \qquad 0\le x\le \frac12.
            \end{equation}
            By Stirling's formula, for \(\beta\in[0,A]\) and \(0\le r\le n/2\) we can write
            \begin{equation}
            T_{n,\beta n,r}
            \le
            C_A\,
            \frac{2n-4r+1}{\sqrt{(r+1)(n-r+1)}}\,
            e^{\,n\phi_\beta(r/n)}
            \end{equation}
            for some irrelevant constant \(C_A>0\).
            Moreover, we can differentiate $\phi_\beta$ and check that
            \begin{equation}
            \phi_\beta\!\left(\frac12\right)=0,
            \qquad
            \phi_\beta'\!\left(\frac12\right)=0,
            \end{equation}
            and
            \begin{equation}
            \phi_\beta''(x)
            =
            -\frac{1}{\frac\beta2+x}
            -\frac{1}{1+\frac\beta2-x}
            \le
            -\frac{4}{A+1},
            \end{equation}
            for all \(\beta\in[0,A]\) and \(x\in[0,1/2]\). Hence
            \begin{equation}
            \phi_\beta(x)\le -\frac{2}{A+1}\left(x-\frac12\right)^2.
            \end{equation}
            Next, we choose $\delta \in (0, 1/2)$ and subdivide the tail contribution into two parts, $K \sqrt{n} \leq s \leq n(1-2\delta)$ and $s > n (1-2\delta)$.
            Considering the latter, since \(\phi_\beta\) is continuous in \((\beta,x)\) and strictly negative on
            \([0,A]\times[0,\delta]\), there exists \(\eta_{A, \delta}>0\) such that
            \begin{equation}
                \phi_\beta(x)\le -\eta_{A, \delta}
                \qquad
                (\beta\in[0,A],\ 0\le x\le \delta).
            \end{equation}
            Therefore, we can write
            \begin{equation}
                \sum_{0\le r\le \delta n} T_{n,m_n,r}
                \le
                C_A n^{3/2}e^{-\eta_{A, \delta} n}
                =o(1).
            \end{equation}
            Moreover, considering $K \sqrt{n} \leq s \leq n(1-2\delta)$, writing again \(r=(n-s)/2\), and using that $\frac{r}{n} - \frac{1}{2} = - \frac{s}{2n}$, we show that
            \begin{equation}\label{eq:tail_bound}
                \sum_{\substack{K \sqrt{n} \leq s \leq n(1-2\delta)\\ s\equiv n\!\!\!\pmod 2}}
                T_{n,m_n,\frac{n-s}{2}}
                \le
                \frac{C_A}{n}\sum_{K \sqrt{n} \leq s \leq n(1-2\delta)}(s+1)e^{-c_A s^2/n},
            \end{equation}
            for some \(c_A>0\), where we used that
            \begin{equation}
                \frac{2n-4r+1}{\sqrt{(r+1)(n-r+1)}} = \frac{2s+1}{\sqrt{((n-s)/2+1)((n+s)/2+1)}} =  O\left( \frac{s+1}{n} \right)
            \end{equation}
            whenever $s \leq n(1-2\delta)$.
            For fixed \(K > 0\), the right-hand side of \Cref{eq:tail_bound} tends to
            \begin{equation}
            C_A\int_K^\infty t e^{-c_A t^2}\,dt,
            \end{equation}
            which goes to \(0\) as \(K\to\infty\). Thus the tail \(s\ge K\sqrt n\) is negligible as we increase $K$, i.e.,
            \begin{equation}
                \lim_{K \to \infty} \limsup_{n \to \infty} | R_{n, K}| = 0,
            \end{equation}
            where $R_{n, K}$ is the remainder
            \begin{equation}
                R_{n, K} = \sum_{\substack{s\ge K\sqrt n\\ s\equiv n\!\!\!\pmod 2}}
                T_{n,m_n,\frac{n-s}{2}}
            \end{equation}
        
            \medskip
            \noindent \textbf{Combining the two contributions.}
            In total, we can write
            \begin{equation}
                \mathrm{AC}(m_n,n)
                =
                \sum_{\substack{0\le s\le K\sqrt n\\ s\equiv n\!\!\!\pmod 2}}
                T_{n,m_n,\frac{n-s}{2}}
                +R_{n, K}.
            \end{equation}
            Using \eqref{eq:local-T}, we can write
            \begin{equation}
                \mathrm{AC}(m_n,n)
                =
                \sum_{\substack{0\le s\le K\sqrt n\\ s\equiv n\!\!\!\pmod 2}}
                a_{n, s}
                +o(1)+R_{n, K},
            \end{equation}
            where we denoted
            \begin{equation}
                a_{n, s} \coloneqq \frac{2s+1}{n}
                \exp\!\left(-\frac{s^2}{2(1+\alpha_n)n}\right)
            \end{equation}
            Now set \(t=s/\sqrt n\). Since admissible values of \(s\) have spacing \(2\), the
            mesh in \(t\) is \(2/\sqrt n\), and
            \begin{equation}
            \frac{2s+1}{n}
            =
            t\,\frac{2}{\sqrt n}+o\!\left(\frac{1}{\sqrt n}\right).
            \end{equation}
            Hence, the main contribution can be viewed as a Riemann sum and tends to an integral:
            \begin{equation}
                 \sum_{\substack{0\le s\le K\sqrt n\\ s\equiv n\!\!\!\pmod 2}}
                 a_{n, s} = 
                 \sum_{\substack{0\le s\le K\sqrt n\\ s\equiv n\!\!\!\pmod 2}}
                \frac{2s+1}{n}
                \exp\!\left(-\frac{s^2}{2(1+\alpha_n)n}\right)
                \to
                \int_0^K t\,e^{-t^2/(2(1+\alpha))}\,dt,
            \end{equation}
            where we have used that \(\alpha_n\to\alpha\). Finally, 
            for fixed $K$, we can write
            \begin{equation}
                \mathrm{AC}(m_n,n)
                =
                \sum_{\substack{0\le s\le K\sqrt n\\ s\equiv n\!\!\!\pmod 2}}
                a_{n, s}
                +o(1)+R_{n, K},
            \end{equation}
            which, spelled out, means that
            \begin{equation}
                \limsup_{n\to\infty} \bigg| \mathrm{AC}(m_n,n)
                -
                \sum_{\substack{0\le s\le K\sqrt n\\ s\equiv n\!\!\!\pmod 2}}
                a_{n, s}
                \bigg| \leq \limsup_{n\to\infty} | R_{n, K} |,
            \end{equation}
            and taking also $K \to \infty$, the right hand side tends to $0$, yielding
            \begin{equation}
            \lim_{n\to\infty}\mathrm{AC}(m_n,n)
            =
            \int_0^\infty t\,e^{-t^2/(2(1+\alpha))}\,dt
            =
            1+\alpha,
            \end{equation}
            which proves the lemma.
        \end{proof}

    \subsection{Bounds for the anticoncentration score in Scattershot Boson Sampling when $d=m$}\label{app:scattershot_bounds}
        With the simplified formula from \Cref{prop:scattershot_final_form}, we can prove the following upper and lower bounds for the anticoncentration score in Scattershot Boson Sampling in the case when all the pairs of modes are initialized with uniform two-mode squeezed states:
        \begin{proposition} \label{prop:SBS_lxe_bounds}
            Let \(m\ge 3\) and \(n\ge 12\).
            We can write that
            \begin{equation}
                \frac{\sqrt2\,n(m-1)}{24\sqrt{(3m+2n)(2m+3n)}}
                \le
                \mathrm{AC}_{\mathrm{SBS}}(m, n, m)
                \le
                \frac{3(m-1)}{4}\sqrt{\frac{2\pi}{m-2}} + 3(m-1)\sqrt{\frac{n}{m}}.
            \end{equation}
            In particular, we can write
            \begin{align}\begin{split}
                \mathrm{AC}_{\mathrm{SBS}}(m, n, m)
                &=
                \Omega\!\left(\frac{mn}{m+n}\right), \\
                \mathrm{AC}_{\mathrm{SBS}}(m, n, m)
                &=
                O\!\left(\sqrt{mn}\right).
            \end{split}\end{align}
        \end{proposition}
    
        \begin{proof}
            Using \Cref{eq:ac_sbs_lxeb}, we can write
            \begin{equation}
                \mathrm{AC}_{\mathrm{SBS}}(m, n) = |\mathcal{S}_{m,n}|^2 \mathrm{LXE}_{\mathrm{ref, SBS}}(m, n) = 
                |\mathcal S_{m,n}|^2
                \sum_{r=0}^{\lfloor n/2\rfloor}
                    T_{m, n, r},
            \end{equation}
            where
            \begin{equation}
                T_{m, n, r} \coloneqq \frac{2n-4r+1}{2n-2r+1}
                \frac{\binom{n-r}{r}^2}{\binom{2n-2r}{n}^2}
                \frac{(m+n)_{n-2r}^2\left(\frac{m-1}{2}\right)_r^2}
                {\left(\frac{m+1}{2}\right)_{n-r}^2}
                \frac{1}{\binom{2n+m-2r-1}{m-1}\binom{m+2r-2}{m-2}}.
            \end{equation}
            Using
            \(
                (x)_k=\frac{\Gamma(x+k)}{\Gamma(x)},
            \)
            and
            \(
                \binom{u}{v}=\frac{\Gamma(u+1)}{\Gamma(v+1)\Gamma(u-v+1)},
            \)
            together with the duplication formula (\Cref{prop:duplication_formula}), a direct simplification gives
            \begin{equation}
                T_{m,n, r}=\frac{m-1}{4 \binom{n+m-1}{n}^2}\,u_{m,n,r},
            \end{equation}
            where
            \begin{equation}
                u_{m,n,r}:=
                (2n-4r+1)
                \frac{\Gamma(r+\tfrac12)}{\Gamma(r+1)}
                \frac{\Gamma(\tfrac m2+r-\tfrac12)}{\Gamma(\tfrac m2+r)}
                \frac{\Gamma(n-r+1)}{\Gamma(n-r+\tfrac32)}
                \frac{\Gamma(\tfrac m2+n-r)}{\Gamma(\tfrac m2+n-r+\tfrac12)}.
            \end{equation}
            Hence, we can rewrite the full expression for $|\mathcal{S}_{m,n}|^2 \mathrm{LXE}_{\mathrm{ref, SBS}}(m, n)$ as
            \begin{equation}
                |\mathcal{S}_{m,n}|^2 \mathrm{LXE}_{\mathrm{ref, SBS}}(m, n) = 
                \frac{m-1}{4}
                \sum_{r=0}^{\lfloor n/2\rfloor}
                    u_{m,n,r},
            \end{equation}
        
            The remaining task is to estimate the remaining sum
            \begin{equation}
                S_{m,n}:=\sum_{r=0}^{\lfloor n/2\rfloor}u_{m,n,r}.
            \end{equation}
            Note, that every term in $u_{m, n, r}$ can be bounded by Gautschi's inequalities from \Cref{prop:gautschi} as
            \begin{equation}
                \frac{1}{\sqrt{x+1}}
                \le
                \frac{\Gamma(x+\frac12)}{\Gamma(x+1)}
                \le
                \frac{1}{\sqrt{x}},
                \qquad x>0,
            \end{equation}
            and
            \begin{equation}
                \frac{1}{\sqrt{x}}
                \le
                \frac{\Gamma(x)}{\Gamma(x+\frac12)}
                \le
                \frac{1}{\sqrt{x-\frac12}},
                \qquad x>\frac12.
            \end{equation}
            \textit{Upper bound.} We first derive an upper bound, valid for all \(m\ge 3\) and \(n\ge 12\).
            For \(r=0\),
            \begin{equation}
                u_{m,n,0}
                =
                (2n+1)\sqrt{\pi}\,
                \frac{\Gamma(\frac m2-\frac12)}{\Gamma(\frac m2)}
                \frac{\Gamma(n+1)}{\Gamma(n+\frac32)}
                \frac{\Gamma(\frac m2+n)}{\Gamma(\frac m2+n+\frac12)}.
            \end{equation}
            By the above inequalities, we can write
            \begin{equation}
                \frac{\Gamma(\frac m2-\frac12)}{\Gamma(\frac m2)}
                \le \frac{1}{\sqrt{\frac m2 -1}},
                \qquad
                \frac{\Gamma(n+1)}{\Gamma(n+\frac32)}
                \le \frac{1}{\sqrt n},
                \qquad
                \frac{\Gamma(\frac m2+n)}{\Gamma(\frac m2+n+\frac12)}
                \le \frac{1}{\sqrt n},
            \end{equation}
            and therefore we get the following upper bound for the $u_{m,n,0}$ term:
            \begin{equation}
                u_{m,n,0}
                \le
                (2n+1)\sqrt{\frac{\pi}{\frac m2 -1}} \frac1n
                \le
                3\sqrt{\frac{2\pi}{m-2}}.
            \end{equation}
            Now let \(1\le r\le n/2\). Applying the upper bounds termwise gives
            \begin{equation}
                u_{m,n,r}
                \le
                (2n+1)\,
                \frac{1}{\sqrt r}\,
                \frac{1}{\sqrt{\frac m2+r-1}}\,
                \frac{1}{\sqrt{n-r+\frac12}}\,
                \frac{1}{\sqrt{\frac m2+n-r-\frac12}}.
            \end{equation}
            Since \(m\ge 3\) and \(r\le n/2\),
            \begin{equation}
                \frac m2+r-1\ge \frac m2,
                \qquad
                n-r+\frac12\ge \frac n2,
                \qquad
                \frac m2+n-r-\frac12\ge \frac n2.
            \end{equation}
            Hence
            \begin{equation}
                u_{m,n,r}
                \le
                \frac{2\sqrt2(2n+1)}{n\sqrt{mr}}
                \le
                \frac{6\sqrt2}{\sqrt{mr}}.
            \end{equation}
            Summing over \(1\le r\le \lfloor n/2\rfloor\) and using \Cref{prop:inverse_sqrt_ub} we can write
            \begin{equation}
                \sum_{r=1}^{N}r^{-1/2}\le 2\sqrt N,
            \end{equation}
            which means that
            \begin{equation}
                \sum_{r=1}^{\lfloor n/2\rfloor}u_{m,n,r}
                \le
                \frac{6\sqrt2}{\sqrt m}
                \sum_{r=1}^{\lfloor n/2\rfloor}\frac{1}{\sqrt r}
                \le
                12\sqrt{\frac{n}{m}}.
            \end{equation}
            Therefore
            \begin{equation}
                S_{m,n}
                \le
                3\sqrt{\frac{2\pi}{m -2}}+12\sqrt{\frac{n}{m}}.
                \label{eq:Smn_upper_no_assumption}
            \end{equation}
    
            \textit{Lower bound.} We turn our attention now to the lower bound. Assume \(n\ge 12\), and consider the block
            \begin{equation}
                I_n:=
                \left\{
                    r:
                    \left\lceil\frac n4\right\rceil
                    \le r\le
                    \left\lfloor\frac n3\right\rfloor
                \right\}.
            \end{equation}
            Then \(|I_n|\ge n/12\). For every \(r\in I_n\), the lower bounds give
            \begin{equation}
                u_{m,n,r}
                \ge
                (2n-4r+1)\,
                \frac{1}{\sqrt{r+1}}\,
                \frac{1}{\sqrt{\frac m2+r}}\,
                \frac{1}{\sqrt{n-r+1}}\,
                \frac{1}{\sqrt{\frac m2+n-r}}.
            \end{equation}
            Now, for \(r\in I_n\),
            \begin{equation}
                2n-4r+1\ge \frac{2n}{3},
                \qquad
                r+1\le \frac{4n}{3},
                \qquad
                \frac m2+r\le \frac m2+\frac n3=\frac{3m+2n}{6},
            \end{equation}
            and also
            \begin{equation}
                n-r+1\le n,
                \qquad
                \frac m2+n-r\le \frac m2+\frac{3n}{4}=\frac{2m+3n}{4}.
            \end{equation}
            Therefore
            \begin{equation}
                u_{m,n,r}
                \ge
                \frac{2\sqrt2}{\sqrt{(3m+2n)(2m+3n)}}.
            \end{equation}
            Summing over \(r\in I_n\), we obtain
            \begin{equation}
                S_{m,n}
                \ge
                \frac{\sqrt2\,n}{6\sqrt{(3m+2n)(2m+3n)}}.
                \label{eq:Smn_lower_no_assumption}
            \end{equation}
    
            Combining \Cref{eq:Smn_upper_no_assumption,eq:Smn_lower_no_assumption}, we conclude that for all \(m\ge 3\) and \(n\ge 12\),
            \begin{equation}
                \frac{\sqrt2\,n}{6\sqrt{(3m+2n)(2m+3n)}}
                \le
                S_{m,n}
                \le
                3\sqrt{\frac{2\pi}{m-2}} + 12\sqrt{\frac{n}{m}}.
            \end{equation}
            Since
            \begin{equation}
                \mathrm{AC}_{\mathrm{SBS}}(m, n)
                =
                \frac{m-1}{4}\,S_{m,n},
            \end{equation}
            and it follows that
            \begin{equation}
                \frac{\sqrt2\,n(m-1)}{24\sqrt{(3m+2n)(2m+3n)}}
                \le
                \mathrm{AC}_{\mathrm{SBS}}(m, n)
                \le
                \frac{3(m-1)}{4}\sqrt{\frac{2\pi}{m-2}} + 3(m-1)\sqrt{\frac{n}{m}}.
            \end{equation}
        \end{proof}

    \subsection{Bounds for the anticoncentration score in Gaussian Boson Sampling when $d = m$}\label{app:gbs_bounds}
        \begin{proposition}\label{prop:gbs_bounds}
            Let \(d=m\) and \(m\ge 2\). In this case, we can write
            \begin{equation}
                \sqrt{\frac{\pi n}{2}}\,
                \frac{m-1}{\sqrt{m(m+n-1)}}
                \;\le\;
                \mathrm{AC}_{\mathrm{GBS}}(m,n,m)
                \;\le\;
                \sqrt{\frac{\pi(n+2)(m+n)}{2(m-1)}}.
            \end{equation}
            Consequently,
            \begin{equation}
                \mathrm{AC}_{\mathrm{GBS}}(m,n,m)
                =
                \Omega\!\left(\sqrt{\frac{nm}{m+n}}\right),
                \qquad
                \mathrm{AC}_{\mathrm{GBS}}(m,n,m)
                =
                O\!\left(\sqrt{\frac{n(m+n)}{m}}\right).
            \end{equation}
        \end{proposition}
        \begin{proof}
            Let \(n=2N\) and set \(a:=m/2\), so \(a\ge 1\) because \(m\ge 2\).
            Starting from
            \begin{equation}
                \mathrm{AC}_{\mathrm{GBS}}(m,n,m)
                =
                \binom{m+n-1}{n}\,\mathrm{LXE}_{\mathrm{ref}}^{(n)}(\rho),
            \end{equation}
            and substituting \(d=m=2a\), \(n=2N\), we obtain
            \begin{equation}
                \mathrm{AC}_{\mathrm{GBS}}(m,n,m)
                =
                \binom{2a+2N-1}{2N}
                N!^2
                \sum_{j=0}^{N}
                \frac{4N-4j+1}{4N-2j+1}
                \frac{\binom{2j}{j}\binom{2N-2j}{N-j}^{2}}
                {\binom{4N-2j}{2N-j}}
                \frac{\left(a-\frac12\right)_j(a+N)_{N-j}}
                {(a)_N(a)_j\left(a+\frac12\right)_{2N-j}}.
            \end{equation}
            Using
            \begin{equation}
                \binom{2a+2N-1}{2N}=\frac{(2a)_{2N}}{(2N)!},
                \qquad
                (2a)_{2N}=4^N(a)_N\left(a+\tfrac12\right)_N,
                \qquad
                \frac{N!^2}{(2N)!}=\frac{1}{\binom{2N}{N}},
            \end{equation}
            this becomes
            \begin{equation}
                \mathrm{AC}_{\mathrm{GBS}}(m,n,m)
                =
                \frac{4^N}{\binom{2N}{N}}
                \sum_{j=0}^{N}
                \frac{4N-4j+1}{4N-2j+1}
                \frac{\binom{2j}{j}\binom{2N-2j}{N-j}^{2}}
                {\binom{4N-2j}{2N-j}}
                \frac{\left(a+\frac12\right)_N\left(a-\frac12\right)_j(a+N)_{N-j}}
                {(a)_j\left(a+\frac12\right)_{2N-j}}.
            \end{equation}
            Now write the Pochhammer symbols in terms of gamma functions. A straightforward simplification gives
            \begin{equation}
                \mathrm{AC}_{\mathrm{GBS}}(m,n,m)
                =
                P_{a,N}\sum_{j=0}^{N} W_{N,j}\,V_{a,N}(j),
            \end{equation}
            where
            \begin{equation}
                P_{a,N}
                :=
                \frac{4^N}{\binom{2N}{N}}
                \frac{\Gamma(a)}{\Gamma\!\left(a+\tfrac12\right)}
                \frac{\Gamma\!\left(a+N+\tfrac12\right)}{\Gamma(a+N)},
            \end{equation}
            \begin{equation}
                W_{N,j}
                :=
                \frac{4N-4j+1}{4N-2j+1}
                \frac{\binom{2j}{j}\binom{2N-2j}{N-j}^{2}}
                {\binom{4N-2j}{2N-j}},
            \end{equation}
            and
            \begin{equation}
                V_{a,N}(j)
                \coloneqq
                \left(a-\tfrac12\right)
                \frac{\Gamma\!\left(a-\tfrac12+j\right)}{\Gamma(a+j)}
                \frac{\Gamma(a+2N-j)}{\Gamma\!\left(a+2N-j+\tfrac12\right)}.
            \end{equation}
            Using Equation 16.4.9 from Ref.~\cite{NIST:DLMF}, we can show that the weights \(W_{N,j}\) satisfy
            \begin{equation}
            \sum_{j=0}^{N}W_{N,j}=1,
            \end{equation}
            hence, $\{ W_{N,j}\}_{j=0}^N$ is a probability distribution.
            Next, \(V_{a,N}(j)\) is decreasing in \(j\). Indeed,
            \begin{equation}
                \frac{V_{a,N}(j+1)}{V_{a,N}(j)}
                =
                \frac{\left(a+j-\tfrac12\right)\left(a+2N-j-\tfrac12\right)}
                {(a+j)(a+2N-j-1)},
            \end{equation}
            and
            \(
                (a+j)(a+2N-j-1)
                -
                \left(a+j-\tfrac12\right)\left(a+2N-j-\tfrac12\right)
                =
                N-j-\tfrac14>0
            \)
            for \(0\le j\le N-1\), meaning that the denominator is always bigger.
            Therefore
            \begin{equation}
            V_{a,N}(N)\le \sum_{j=0}^{N}W_{N,j}V_{a,N}(j)\le V_{a,N}(0).
            \end{equation}
            Consequently,
            \begin{equation}
            P_{a,N}V_{a,N}(N)
            \le
            \mathrm{AC}_{\mathrm{GBS}}(m,n,m)
            \le
            P_{a,N}V_{a,N}(0).
            \end{equation}
            
            We now estimate the two endpoints. First, we can write
            \begin{equation}
                V_{a,N}(N)
                =
                \left(a-\tfrac12\right)
                \frac{\Gamma\!\left(a+N-\tfrac12\right)}{\Gamma(a+N)}
                \frac{\Gamma(a+N)}{\Gamma\!\left(a+N+\tfrac12\right)}
                =
                \frac{a-\tfrac12}{a+N-\tfrac12}.
            \end{equation}
            Also, using \Cref{prop:gautschi} we get
            \begin{equation}
                V_{a,N}(0)
                =
                \frac{\Gamma\!\left(a+\tfrac12\right)}{\Gamma(a)}
                \frac{\Gamma(a+2N)}{\Gamma\!\left(a+2N+\tfrac12\right)}
                \le \sqrt{\frac{a}{a+2N-\frac 12}} \leq 1.
            \end{equation}
            Hence, we can write
            \begin{equation}
                P_{a,N}\frac{a-\tfrac12}{a+N-\tfrac12}
                \le
                \mathrm{AC}_{\mathrm{GBS}}(m,n,m)
                \le
                P_{a,N}.
            \end{equation}
            
            The remaining task is to bound \(P_{a,N}\). Firstly, using \Cref{prop:stirling_bounds} for the central binomial coefficient, we get
            \begin{equation}
                \sqrt{\pi N}
                \le
                \frac{4^N}{\binom{2N}{N}}
                \le
                \sqrt{\pi(N+1)}.
            \end{equation}
            Secondly, \Cref{prop:gautschi} yields the bounds
            \begin{equation}
                \frac{1}{\sqrt a}
                \le
                \frac{\Gamma(a)}{\Gamma\!\left(a+\tfrac12\right)}
                \le
                \frac{1}{\sqrt{a-\tfrac12}}, \qquad         \sqrt{a+N-\tfrac12}
            \le
            \frac{\Gamma\!\left(a+N+\tfrac12\right)}{\Gamma(a+N)}
            \le
            \sqrt{a+N}.
            \end{equation}
            Consequently, $P_{a,N}$ can be bounded as
            \begin{equation}
                \sqrt{\pi N}\sqrt{\frac{a+N-\tfrac12}{a}}
                \le
                P_{a,N}
                \le
                \sqrt{\pi(N+1)}\sqrt{\frac{a+N}{a-\tfrac12}}.
            \end{equation}
            Finally, combining the bounds and substituting \(a=m/2, N=n/2\) gives
            \begin{equation}
                \sqrt{\frac{\pi n}{2}}\,
                \frac{m-1}{\sqrt{m(m+n-1)}}
                \le
                \mathrm{AC}_{\mathrm{GBS}}(m,n,m)
                \le
                \sqrt{\frac{\pi(n+2)(m+n)}{2(m-1)}}.
            \end{equation}
        \end{proof}

\section{Further results}\label{app:further_results}

\subsection{Sample complexity of certification tests}\label{app:sample_complexity}
    The anticoncentration bounds established in \Cref{sec:anticoncentration} have consequences beyond hardness-of-simulation arguments. In particular, they imply limitations on the sample complexity of certifying the photonic quantum advantage experiments considered here. Using the framework developed by Hangleiter et al.~\cite{Hangleiter_2019}, based on earlier results of Valiant and Valiant~\cite{valiantvaliant}, we derive lower bounds on the certification sample complexity for Boson Sampling, Scattershot Boson Sampling, and Gaussian Boson Sampling. These lower bounds arise from the fact that the associated anticoncentration scores scale at most polynomially in \(n\) and \(m\).

    In the following, we will treat the quantum advantage experiment as a black box. Suppose we are given sample access to a such black-box device, and its output distribution is an unknown probability distribution \(q\). The task is to distinguish the case \(q=p\) from the case \(\norm{p-q}_1>\epsilon\) with high success probability. This is formalized by the following definition.
    \begin{definition}[$\epsilon$-certification test~\cite{Hangleiter_2019}]\label{def:certification}
          Let $p$ be a (target) probability distribution on a sample space $\mathcal E$.
          We call
          $\mathcal T: \mathcal E^s \to \{0,1\}$
          an $\epsilon$-\emph{certification test} of $p$ from $s$ samples if the following completeness and soundness conditions are satisfied for any distribution $q$ over $\mathcal E$:
          \begin{subequations}
              
          \begin{align}
                q = p
                & \ \Rightarrow \
                \Pr_{ X \sim q^s } [\mathcal T(X) =1 ] \geq \frac23,
                \\
                \norm{p-q}_{1} > \epsilon
                &\ \Rightarrow\
                \Pr_{ X \sim q^s} [\mathcal T(X) =1 ] < \frac13 \, .
          \end{align}
          \end{subequations}
          Moreover, for a family of probability distributions we call a family of tests $\{\mathcal T_{n}\}$ a \emph{sample-efficient $\epsilon$-certification test} if for every $n$, $\mathcal T_{n}$ is an $\epsilon$-certification test from $s = O(\poly(n,1/\epsilon))$ samples.
    \end{definition}
    \noindent It turns out that the optimal sample complexity for $\epsilon$-certifying any distribution $p$ is essentially controlled by a single quantity $\| p^{-\max}_{-\epsilon}\|_{2/3}$, which we define as follows.
    Given a probability distribution $p$, let $p^{-\max}_{-\epsilon}$ denote the vector obtained from the image of $p$ by setting its largest entry to $0$ and also iteratively removing the smallest elements from $p$ with the condition that the sum of the removed entries is at most $\epsilon$. Moreover, we will make use of the $\ell_{2/3}$-norm as given by $\norm{\bm{x}}_{2/3} = ( \sum_i |x_i|^{2/3} )^{3/2}$.
    The following result gives guarantees on the sample efficiency of $\epsilon$-certification tests using $\| p^{-\max}_{-\epsilon}\|_{2/3}$
    \begin{theorem}[Optimal $\epsilon$-certification tests~\cite{valiantvaliant}]\label{thm:valiantvaliant}
          There exist constants $c_1, c_2 > 0 $ such that for any $\epsilon > 0$ and any target distribution $p$, there exists an $\epsilon$-certification test from \begin{equation}
            c_1\,\max \{ \frac 1 \epsilon, \frac1{\epsilon^2}\,\norm{p_{- \epsilon/16}^{- \mathrm{max}} } _{{2/3}} \}    
          \end{equation}
          many samples, but there exists no $\epsilon$-certification test from fewer than \begin{equation}
            c_2\, \max \{ \frac 1 \epsilon, \frac 1{\epsilon^2}\,\norm{p_{- 2\epsilon}^{- \mathrm{max}} } _{{2/3}}  \}
        \end{equation}
        samples.
    \end{theorem}
    \noindent Hence, the quantity $\norm{p_{- \epsilon}^{- \mathrm{max}} } _{{2/3}}$ controls the sample complexity of the $\epsilon$-certification. We can give a lower bound of this quantity using the min-entropy $H_\infty(p)$ as follows (a special case of Lemma 3 from~\cite{Hangleiter_2019}): 
    \begin{equation}
      2^{\frac12 H_\infty(p)}  \left(1- \epsilon - 2^{-H_\infty(p)} \right)^{3/2}
      \leq \norm{p_{- \epsilon}^{- \mathrm{max}} } _{{2/3}}.
    \end{equation}
    In Boson Sampling, and likewise in Scattershot and Gaussian Boson Sampling when \(d=m\), the anticoncentration score, which we denote by \(\mathrm{AC}(m,n)\), is bounded above by a polynomial in \(m\) and \(n\), i.e., $\mathrm{AC}(m,n)\le \poly(m,n)$.
    As a result, the averaged second moment of the output distribution is of the same order as the inverse number of possible outcomes.
    This immediately yields a lower bound on the min-entropy of the output distribution.
    Indeed, the tail bound of the min-entropy from Ref.~\cite[Lemma 5]{Hangleiter_2019} states that for every \(\delta\in(0,1)\),
    \begin{equation}
        \Pr_{U\sim \mathrm{Haar}(m)}
        \!\left[
        H_\infty(p_U)
        \ge
        \frac12\left(
        \log \delta
        -
        \log
        \sum_{\bm n \in \mathcal E}
        \underset{U\sim\mathrm{Haar}(m)}{\mathbb E}\!\bigl[p_U(\bm n)^2\bigr]
        \right)
        \right]
        \ge 1-\delta,
    \end{equation}
    which can be cast into the following form using the anticoncentration score $\mathrm{AC}(m, n)$:
    \begin{equation}            
        \Pr_{U\sim \mathrm{Haar}(m)}
        \!\left[
        H_\infty(p_U)
    \ge
    \frac12 \log\!\bigl(\delta |\mathcal E|\bigr)
    -
    \log \mathrm{AC}(m, n)
        \right]
        \ge 1-\delta.
    \end{equation}
    Consequently, for a typical Haar-random interferometer, in the considered photonic quantum advantage schemes, the output distribution has min-entropy of order \(\frac12 \log |\mathcal E |\).
    Roughly speaking, an $\epsilon$-certification test requires at least
    \begin{equation}
        c_2 \frac{\delta^{1/4} |\mathcal E|^{1/4}}{\epsilon^2 \sqrt{\mathrm{AC}(m, n)}} \left( 1 - 2\epsilon - \frac{\mathrm{AC}(m, n)}{\delta^{1/2} |\mathcal E|^{1/2}}\right)^{3/2}
    \end{equation}
    many samples with probability $1-\delta$ over the choice of $U \sim \mathrm{Haar}(m)$, where $c_2$ is the constant from \Cref{thm:valiantvaliant}.
    Therefore, when we fix \(\epsilon,\delta\in(0,1)\) and assume that
    \(\mathrm{AC}(m,n)\ll |\mathcal E|^{1/2}\), the sample complexity is lower bounded by
    \begin{equation}
    \Omega\!\left(
    |\mathcal E|^{1/4}\,\mathrm{AC}(m,n)^{-1/2}
    \right)
    \end{equation}
    with probability at least \(1-\delta\) over \(U\sim\mathrm{Haar}(m)\). Using the results from \Cref{sec:anticoncentration}, the assumption \(\mathrm{AC}(m,n)\ll |\mathcal E|^{1/2}\) trivially applies to Boson Sampling, Scattershot Boson Sampling, and Gaussian Boson Sampling, and hence, none of these models admits a sample-efficient \(\epsilon\)-certification test.
        
\subsection{Proof of Hunter-Jones conjecture for $t=2$}\label{app:proof_hunter_jones}
    In this section, we set out to prove the following statement, a special case of the Hunter-Jones conjecture ~\cite{nezami2021permanentrandommatricesrepresentation}:
    \begin{restatable}[Hunter-Jones conjecture for \(t=2\)]{proposition}{hunterjones}
        \label{prop:hunter_jones_t2}
        Let \(U\sim \mathrm{Haar}(n)\), and define
        \begin{equation}
            X_U \coloneqq |\! \per(U)|^2.
        \end{equation}
        Then
        \begin{equation}
            \lim_{n\to\infty}
            \frac{
                \underset{U\sim \mathrm{Haar}(n)}{\mathbb{E}}\!\left[X_U^2\right]
            }{
                \left(
                \underset{U\sim \mathrm{Haar}(n)}{\mathbb{E}}\!\left[X_U\right]
                \right)^2
            }
            =2.
        \end{equation}
    \end{restatable}
    \begin{proof}
        For brevity, let us denote
        \begin{equation}
            A_n
            \coloneqq
            \frac{
                \underset{U\sim \mathrm{Haar}(n)}{\mathbb{E}}\!\left[X_U^2\right]
            }{
                \left(
                \underset{U\sim \mathrm{Haar}(n)}{\mathbb{E}}\!\left[X_U\right]
                \right)^2
            }.
        \end{equation}
        By the discussion from \Cref{sec:framework}, we may write
        \begin{equation}
            A_n
            =
           | \mathcal{S}_{n,n}|^2
            \sum_{k=0}^n
            \frac{\Tr[ \mathbb{P}_k \ketbra{\bm{n}_0}^{\otimes 2} ]^2}{\Tr P_k},
        \end{equation}
        where \(\bm n_0=(1,\dots,1)\) and
        we use the dimension formula from \Cref{eq:dimensions},
        \begin{equation}
            \Tr \mathbb{P}_k
            =
            \frac{2n-2k+1}{2n-k+1}
            \binom{3n-k-1}{n-2}\binom{n+k-2}{n-2}
            =
            \frac{n^2(n-1)^2(2n-2k+1)}
            {(2n-k+1)^2(3n-k)(n+k)(n+k-1)}
            \binom{3n-k}{n}\binom{n+k}{n}.
        \end{equation}
        together with the collision-free input formula from \Cref{eq:tr_pk_n0}.
        Moreover, for brevity, we use the notation
        \begin{subequations}
        \begin{align}
            p_k(n) &\coloneqq \Tr[ \mathbb{P}_k \ketbra{\bm{n}_0}^{\otimes 2} ],\\
            a_k(n) &\coloneqq | \mathcal{S}_{n,n}|^2 \,\frac{p_k(n)}{\Tr P_k},\\
            b_k(n) &\coloneqq a_k(n)\,p_k(n).
        \end{align}
        \end{subequations}
        Importantly, $p_k(n)$ is a probability distribution over $k$ as $p_k(n) \geq 0$ and $\sum_{k=0}^n p_k(n) = 1$.
        We may simply express $S_n$ as
        \begin{equation}
            A_n=\sum_{k=0}^n b_k(n).
        \end{equation}
        To simplify the expressions, note that
        \begin{equation}
            |\mathcal{S}_{n,n}| = \binom{2n-1}{n} = \frac12 \binom{2n}{n}.
        \end{equation}
        Hence, for even \(k\),
        \begin{equation}
            a_k(n)
            =
            \binom{2n}{n}^2
            \frac{2^{\,n-k}(2n-k+1)(3n-k)(n+k)(n+k-1)}
            {4n^2(n-1)^2}
            \frac{\binom{n}{k/2}}
            {\binom{2n-k}{n-k/2}\binom{3n-k}{n}\binom{n+k}{n}},
        \end{equation}
        while \(a_k(n)=0\) for odd \(k\).
        Next, we choose
        \(
            \alpha_n\coloneqq \lfloor n^\gamma\rfloor,
            \) and \( \frac12<\gamma<\frac34
        \),
        and split the sum as
        \begin{equation}
            A_n
            =
            \underbrace{\sum_{k=0}^{n-\alpha_n} b_k(n)}_{\text{bulk contribution}}
            +
            \underbrace{
            \sum_{k=n-\alpha_n+1}^{n} b_k(n)
            }_{\text{edge contribution}}
            .
        \end{equation}
        The proof strategy is as follows: we will show that the bulk contribution converges to $0$, and that the edge contribution can be upper bounded by $2$ in the limit $n \to \infty$, yielding that 
        \begin{equation}
            \limsup_{n\to\infty} A_n \leq 2.
        \end{equation}
        In the following, we process the two pieces separately, starting with the edge contribution.
    
        \medskip
        \noindent\textbf{1. Edge contribution.}
        Let \(k=n-\beta\), where \(0\le \beta\le \alpha_n\). Since \(p_{n-\beta}(n)=0\) unless \(n-\beta\) is even, only admissible \(\beta\) contribute. For such \(\beta\), a direct calculation gives
        \begin{equation}
            a_{n-\beta}(n)
            =
            2^\beta
            \frac{(n+\beta+1)(4-\beta^2/n^2)(2n-\beta-1)}{4(n-1)^2}
            \frac{
                (2n)!^2
                \bigl(\frac{n+\beta}{2}\bigr)!
                (n-\beta)!
            }{
                \bigl(\frac{n-\beta}{2}\bigr)!
                (2n+\beta)!
                (2n-\beta)!
                n!
            }.
        \end{equation}
        Write \(\varepsilon_n\coloneqq \beta/n\). Since \(\beta\le \alpha_n=o(n)\), Stirling's formula from \Cref{prop:stirling_bounds} yields
        \begin{equation}
            \frac{
                (2n)!^2
                \bigl(\frac{n+\beta}{2}\bigr)!
                (n-\beta)!
            }{
                \bigl(\frac{n-\beta}{2}\bigr)!
                (2n+\beta)!
                (2n-\beta)!
                n!
            }
            =
            \sqrt{\frac{4n(n+\beta)}{4n^2-\beta^2}}\,
            2^{-\beta}
            \exp \bigl(n\,\Xi(\varepsilon_n)\bigr)
            \bigl(1+O(n^{-1})\bigr),
        \end{equation}
        where the exponent is an entropic term of the form
        \begin{equation}
            \Xi(x)
            \coloneqq
            \frac12\Bigl[
                (1+x)\log(1+x)
                +(1-x)\log(1-x)
                -4\Bigl(1+\frac{x}{2}\Bigr)\log\Bigl(1+\frac{x}{2}\Bigr)
                -4\Bigl(1-\frac{x}{2}\Bigr)\log\Bigl(1-\frac{x}{2}\Bigr)
            \Bigr].
        \end{equation}
        Expanding at \(x=0\), the quadratic term cancels and we get that
        \begin{equation}
            \Xi(x)=O(x^4).
        \end{equation}
        Hence
        \begin{equation}
            a_{n-\beta}(n)
            =
            \sqrt{\bigl(1-(\varepsilon_n/2)^2\bigr)(1+\varepsilon_n)}
            \,\frac{(2-\varepsilon_n-1/n)(1 + \epsilon_n + 1/n)}{(1-1/n)^2}\,
            \exp \bigl(n\,O(\varepsilon_n^4)\bigr)
            \bigl(1+O(n^{-1})\bigr).
        \end{equation}
        Since \(\varepsilon_n=O(n^{\gamma-1})\), we obtain uniformly for \(0\le \beta\le \alpha_n\) that
        \begin{equation}
            a_{n-\beta}(n)
            =
            2+O(n^{\gamma-1})+O(n^{4\gamma-3})+O(n^{-1})
            =
            2+o(1),
        \end{equation}
        because \(\gamma<3/4\). Therefore,
        \begin{equation}
            \sum_{k=n-\alpha_n+1}^{n} b_k(n)
            =
            \sum_{\beta=0}^{\alpha_n-1} a_{n-\beta}(n)\,p_{n-\beta}(n)
            =
            (2+o(1))
            \sum_{k=n-\alpha_n+1}^{n} p_k(n).
        \end{equation}
        Since \(\sum_{k=0}^n p_k(n)=1\), it follows that
        \begin{equation}\label{eq:edge_bound_hj}
            \limsup_{n\to\infty}
            \sum_{k=n-\alpha_n+1}^{n} b_k(n)
            \le 2.
        \end{equation}
    
        \medskip
        \noindent\textbf{2. Bulk contribution.}
        We first isolate the \(k=0\) term. Using the exact formulas,
        \begin{equation}
        p_0(n)=\frac{2^n}{\binom{2n}{n}},
        \qquad
        \Tr \mathbb{P}_0=\binom{3n-1}{n-2},
        \qquad
        |\mathcal S_{n,n}|=\binom{2n-1}{n},
        \end{equation}
        we obtain
        \begin{equation}
        b_0(n)
        =
        |\mathcal S_{n,n}|^2\,\frac{p_0(n)^2}{\Tr \mathbb{P}_0}
        =
        \frac{4^{n-1}}{\binom{3n-1}{n-2}}.
        \end{equation}
        In particular, \(b_0(n)=e^{-\Omega(n)}\), and hence \(b_0(n)\to0\).
        It therefore remains to consider \(1\le k\le n-\alpha_n\).
        Set
        \begin{equation}
            x\coloneqq \frac{k}{n}\in \Bigl[\frac1n,1-\frac{\alpha_n}{n}\Bigr].
        \end{equation}
        Using the binomial bounds from \Cref{lem:bin_bounds}, one obtains
        \begin{equation}
            p_k(n)
            \le
            \frac{4}{\pi}
            \sqrt{\frac{n(n-k/2)}{(2n-k)(k/2)}}
            \exp\!\left(
                n\left[
                    h(x)-(1-x)\log 2-\Bigl(1-\frac{x}{2}\Bigr)
                    h\!\left(\frac{x}{2-x}\right)
                \right]
            \right),
        \end{equation}
        and
        \begin{multline}
            \Tr P_k
            \ge
            \frac{1}{8}
            \frac{n(n-1)^2}{(n+k)(n+k-1)(3n-k)(2n-k+1)}
            \sqrt{\frac{(3n-k)(n+k)}{(2n-k)k}}
            \\
            \times
            \exp\!\left(
                n\left[
                    (3-x)h\!\left(\frac{1}{3-x}\right)
                    +(1+x)h\!\left(\frac{1}{1+x}\right)
                \right]
            \right),
        \end{multline}
        where \(h\) denotes the binary entropy with natural logarithms. Furthermore,
        \begin{equation}
            | \mathcal{S}_{n,n}|^2
            =
            \binom{2n-1}{n}^2
            \le
            C\frac{16^n}{n}
        \end{equation}
        for some absolute constant \(C>0\). Combining the preceding bounds, we obtain
        \begin{equation}
            b_k(n)
            \le
            C_0\,
            \frac{(1+x)(1+x-1/n)(3-x)(2-x+1/n)}{2x(1-1/n)^2}
            \sqrt{\frac{(2-x)x}{(3-x)(1+x)}}
            \exp\bigl(nF(x)\bigr),
        \end{equation}
        for some absolute constant \(C_0>0\), where
        \begin{equation}
            F(x)
            \coloneqq
            2h(x)+2(1+x)\log 2
            -(2-x)h\!\left(\frac{x}{2-x}\right)
            -(3-x)h\!\left(\frac{1}{3-x}\right)
            -(1+x)h\!\left(\frac{1}{1+x}\right).
        \end{equation}
        A direct simplification gives
        \begin{equation}
            F(x)=4\log 2-(3-x)\log(3-x)-(1+x)\log(1+x).
        \end{equation}
        In particular,
        \begin{equation}
            F(1)=0,\qquad F'(1)=0,
            \qquad
            F''(x)
            =
            -\frac{1}{3-x}-\frac{1}{1+x}
            \le -1
            \qquad (0\le x\le 1).
        \end{equation}
        Therefore, Taylor's theorem around \(x=1\) yields
        \begin{equation}
            F(x)\le -\frac{(1-x)^2}{2}
            \qquad (0\le x\le 1).
        \end{equation}
        Also, the prefactor is \(O(x^{-1/2})\) uniformly on \(x\in[1/n,1]\), so for \(1\le k\le n-\alpha_n\),
        \begin{equation}
            b_k(n)
            \le
            C_1\sqrt{\frac{n}{k}}
            \exp\!\left(-\frac{(n-k)^2}{2n}\right)
        \end{equation}
        for some absolute constant \(C_1>0\). Since \(n-k\ge \alpha_n\), we infer that
        \begin{equation}
            b_k(n)
            \le
            C_1\sqrt{\frac{n}{k}}
            \exp\!\left(-\frac{\alpha_n^2}{2n}\right).
        \end{equation}
        Summing over \(1\le k\le n-\alpha_n\), we obtain
        \begin{align}
            \sum_{k=1}^{n-\alpha_n} b_k(n)
            \le
            C_1
            \exp\!\left(-\frac{\alpha_n^2}{2n}\right)
            \sqrt{n}\sum_{k=1}^{n-\alpha_n}\frac{1}{\sqrt{k}}
            \le
            2C_1 n
            \exp\!\left(-\frac{\alpha_n^2}{2n}\right),
        \end{align}
        where we have used the standard fact that $\sum_{k=1}^{n-\alpha_n}\frac{1}{\sqrt{k}} \leq 2 \sqrt{n-\alpha_n} \leq 2 \sqrt{n}$.
        Since \(\alpha_n=\lfloor n^\gamma\rfloor\) with \(\gamma>1/2\), we have \(\alpha_n^2/n\to\infty\), and therefore
        \begin{equation}\label{eq:bulk_bound_hj}
            \lim_{n\to\infty}\sum_{k=1}^{n-\alpha_n} b_k(n)=0.
        \end{equation}
        Together with \(b_0(n)\to0\), this implies
        \begin{equation}
            \lim_{n\to\infty}\sum_{k=0}^{n-\alpha_n} b_k(n)=0.
        \end{equation}
    
        Combining this with \Cref{eq:edge_bound_hj}, we conclude that
        \begin{equation}
            \limsup_{n\to\infty} A_n\le 2.
        \end{equation}
        On the other hand, Theorem~5.2 of Ref.~\cite{nezami2021permanentrandommatricesrepresentation} gives the matching lower bound
        \begin{equation}
            \liminf_{n\to\infty} A_n\ge 2.
        \end{equation}
        Hence, we can conclude that
        $
            \lim_{n\to\infty} A_n=2
        $, as claimed.
    \end{proof}

\subsection{Bosonic swaps and R\'enyi entropies}\label{app:entanglement}

We show the relevance of our framework for calculating the 
particle-R\'enyi entropy. From \Cref{eq:bosonic_swap_to_trace} it follows directly that the R\'enyi-2 entropy of a particle reduced density matrix $\rho_q = \Tr_{n-q}(\rho)$ can be expressed as the negative logarithm  of the expectation value of $\mathbb{S}_q$ in the state $\rho \otimes \rho$: 
\begin{equation}     
S_2(\rho_q) = -\log \Tr\left[\Tr_{n-q}(\rho)^2\right]= -\log \Tr\left[\mathbb{S}_q\,\rho^{\otimes 2}\right] .
    \end{equation}

Importantly, \Cref{prop:efficient_swap_product} makes this connection computationally useful: for fixed-particle-number projections of product states, the relevant bosonic-swap expectation values can be evaluated efficiently. Since these determine the second Rényi and Tsallis entropies of particle-reduced states, they provide an efficient numerical tool for studying particle entanglement in a broad class of bosonic many-body states. In particular, this applies naturally to generalized Dicke states and to many-body models whose ground states admit such a description~\cite{Carrasco_2016}.

    Considering bosonic Fock states, let $\ket{\bm n}$ be the Fock state with occupation vector $\bm n=(n_1,\dots,n_m)$ and total particle number $|\bm n|=n$. Then
    \begin{align}
        \Tr\!\left[\mathbb{S}_q\,\ketbra{\bm n}^{\otimes 2}\right]
        =
        \Tr\!\left[\left(\Tr_q\ketbra{\bm n}\right)^2\right]
        =
        \frac{1}{\binom{n}{q}^2}
        \sum_{\substack{\bm q\geq \bm 0\\ |\bm q|=q}}
        \prod_{i=1}^m \binom{n_i}{q_i}^2,
    \end{align}
    see also Refs.~\cite{Popkov_2005,Carrasco_2016}. In other words, the purity of the $q$-particle reduced state of a Fock state is given by a weighted sum of squared multinomial coefficients.
    This gives a closed-form expression for the particle entanglement of Dicke and generalized Dicke states in the bosonic language, and it can be evaluated efficiently using \Cref{prop:efficient_swap_product}.

      Moreover, \Cref{prop:mean_purity_simple} can be interpreted as the exact ensemble average of the purity of particle-reduced density matrices for uniformly random bosonic Fock states. Hence, we can write
    \begin{equation}
        \underset{\bm n\sim \mathcal U_{m,n}}{\mathbb E}
        \Tr\!\left[\left(\Tr_q\ketbra{\bm n}\right)^2\right]
        =
        \frac{1}{\binom{n}{q}}\frac{\left(\frac{m+1}{2}\right)_n}
             {\left(\frac{m+1}{2}\right)_q
              \left(\frac{m+1}{2}\right)_{n-q}}.
    \end{equation}
    Thus, the quantity computed in \Cref{prop:mean_purity_simple} is precisely the mean purity of the particle-reduced state obtained by tracing out $q$ particles from a uniformly random $n$-particle Fock state in $m$ modes.
    
    We emphasize that the formula above gives the \emph{average purity} of a $q$-particle-reduced density matrix, and does not translate to the average Rényi entropy itself. Nevertheless, it immediately yields a lower bound, since the function $x\mapsto -\log x$ is convex:
    \begin{equation}
        \underset{\bm n\sim \mathcal U_{m,n}}{\mathbb E}\,
        S_2\!\left(\Tr_q \ketbra{\bm n}\right)
        \ge
        -\log
        \underset{\bm n\sim \mathcal U_{m,n}}{\mathbb E}\,
        \Tr\!\left[\left(\Tr_q \ketbra{\bm n}\right)
^2\right].
    \end{equation}
    Substituting the exact expression for the average purity therefore gives
    \begin{equation}
        \underset{\bm n\sim \mathcal U_{m,n}}{\mathbb E}\,
        S_2\!\left(\Tr_q \ketbra{\bm n}\right)
        \ge
        -\log\!\left[
        \frac{1}{\binom{n}{q}}\frac{\left(\frac{m+1}{2}\right)_n}
             {\left(\frac{m+1}{2}\right)_q
              \left(\frac{m+1}{2}\right)_{n-q}}
        \right].
    \end{equation}
    Thus, our formula yields an explicit lower bound on the average second Rényi entropy of particle-reduced density matrices for uniformly random bosonic Fock states.

 Taken together, these observations show that bosonic-swap expectation values play a dual role. On the one hand, they encode directly meaningful information about the particle-entanglement structure of bosonic many-body states. On the other hand, through \Cref{prop:efficient_swap_product}, they provide an efficient numerical tool for studying second-order entanglement entropies in systems whose ground states admit a Dicke or generalized Dicke description. This is particularly useful for quantum spin models with permutation-invariant ground states, where the bosonic swap operator serves as a probe of correlations and critical behavior.

    Let us denote
    \begin{equation}
        E(n, m, q)
        \coloneqq
        -\log \left[
        \frac{1}{\binom{n}{q}}
        \frac{\left(\frac{m+1}{2}\right)_n}
        {\left(\frac{m+1}{2}\right)_{q}\left(\frac{m+1}{2}\right)_{n-q}}
        \right].
    \end{equation}
    and set \(a = (m+1)/2\) for brevity. Also, let us consider the case when $m = \Theta(n)$ and $\beta = \Theta(n)$, and let us denote $\alpha = m / n$ and $\beta = q / n$. 
    Firstly, using \Cref{lem:bin_bounds}, we can write that
    \begin{equation}
        \log \binom{n}{q} = n h\left(\frac{q}{n}\right) + \frac{1}{2} \log n - \frac{1}{2} \log(q (n-q)) + O(1),
    \end{equation}
    where $h(x) = -x \log x - (1-x) \log(1-x)$ is the binary entropy.
    Moreover, we can write the Pochhammer symbols using Gamma functions as \((x)_k=\Gamma(x+k)/\Gamma(x)\), we can rewrite the second term as
    \begin{equation}\label{eq:entropy_second_term}
        \log\frac{(a)_n}{(a)_q(a)_{n-q}}
        =
        \log\Gamma(n+a) + \log\Gamma(a)
        -\log\Gamma(q+a) - \log\Gamma(n-q+a).
    \end{equation}
    Applying the bounds from \Cref{prop:stirling_bounds} to the Gamma-function representation in \Cref{eq:entropy_second_term}, and collecting the leading terms, we obtain
    \begin{equation}
        E(n,m,q)=n\,\Psi(\alpha,\beta)+O(\log n),
    \end{equation}
    where
    \begin{equation}
        \Psi(\alpha,\beta)
        =
        h(\beta)
        -
        (1+\alpha)\left[
            h\left(\frac{\beta+\alpha/2}{1+\alpha}\right)
            -
            h\left(\frac{\alpha}{2(1+\alpha)}\right)
        \right].
    \end{equation}
    In summary, whenever \(m/n \to \alpha>0\) and \(q/n \to \beta\in(0,1)\), the quantity \(E(n,m,q)\) grows linearly in \(n\). Thus, we get a linear lower bound on R\'enyi-2 entropies, i.e., a volume law holds.

\section{Useful combinatorial identities}\label{app:useful_comb}
    In this section, we compile the formulas that are used in the throughout the appendix. We begin with Stirling’s formula, which is one of the main tools for proving the bounds developed throughout this work.
    \begin{proposition}[Stirling's formula~\cite{Robbins1955ARO}]\label{prop:stirling_bounds}
        For every \(x>0\),
        \begin{equation}\label{eq:stirling_bounds_gamma}
        \sqrt{2\pi}\,x^{\,x+\frac12}e^{-x+\frac{1}{12x+1}}
        <
        \Gamma(x+1)
        <
        \sqrt{2\pi}\,x^{\,x+\frac12}e^{-x+\frac{1}{12x}}.
        \end{equation}
        In particular, for every \(t\ge 1\), we can write
        \begin{equation}\label{eq:stirling_crude}
        \sqrt{2\pi}\,t^{\,t-\frac12}e^{-t}
        <
        \Gamma(t)
        \le
        e\,t^{\,t-\frac12}e^{-t}.
        \end{equation}
        Moreover, equality in the upper bound holds if and only if \(t=1\).
    \end{proposition}

    \noindent As a consequence, we obtain the following bounds for binomial coefficients, which are also ubiquitous throughout our proofs.
    \begin{proposition}[Bounds for binomial coefficients~\cite{Oszmaniec_2022}]\label{lem:bin_bounds}
        Let $n,k$ be a natural numbers such that $k\in\{1,\ldots,n-1\}$, and let $x=\frac{k}{n}$ denote their ratio. The binomial coefficients are bounded as
        \begin{equation}
        \frac{1}{2\sqrt{2}} \sqrt{\frac{n}{k(n-k)}} \exp\left(n\, h(x)\right) \leq \binom{n}{k} \leq \frac{1}{\sqrt{2\pi}} \sqrt{\frac{n}{k(n-k)}} \exp\left(n\, h(x)\right),
        \end{equation}
        where $h(x)=-x\log(x) -(1-x)\log(1-x)$ is the binary entropy.
    \end{proposition}

    \noindent We also use bounds on a certain ratio of Gamma functions. The following is the special case \(s=\frac12\) of Gautschi's inequality~\cite[Eq.~(5.6.4)]{NIST:DLMF}:
    \begin{proposition}[Gautschi's inequality, \(s=\frac12\)]\label{prop:gautschi}
    Let \(x>\frac12\). Then
    \begin{equation}
        \frac{1}{\sqrt{x}}
        <
        \frac{\Gamma(x)}{\Gamma\!\left(x+\frac12\right)}
        <
        \frac{1}{\sqrt{x-\frac12}}.
    \end{equation}
    \end{proposition}

    \noindent We also use a similar, looser but more generally applicable fact:
    \begin{proposition}[Gamma-ratio asymptotics]\label{prop:gamma_ratio_asymptotic}
        Let \(a,b\in\mathbb R\) be fixed. Then, as \(t\to\infty\),
        \begin{equation}
        \frac{\Gamma(t+a)}{\Gamma(t+b)}
        =
        t^{a-b}\bigl(1+o(1)\bigr),
        \end{equation}
        or, equivalently,
        \begin{equation}
        \frac{\Gamma(t+a)}{\Gamma(t+b)\,t^{a-b}}\to 1.
        \end{equation}
    \end{proposition}

    \noindent Next, we make use of the duplication formula for the Gamma functions~\cite[Eq.~(5.5.5)]{NIST:DLMF}:
    \begin{proposition}[Legendre duplication formula]\label{prop:duplication_formula}
        For every \(z\in\mathbb{C}\) such that neither side is singular, one has
        \begin{equation}\label{eq:duplication_explicit}
            \Gamma(z)\Gamma\!\left(z+\frac12\right)
            =
            2^{1-2z}\sqrt{\pi}\,\Gamma(2z).
        \end{equation}
    \end{proposition}
    
    \noindent The following statement simplifies certain sums of rational expressions containing $2$ Pochhammer symbols in the numerator and $1$ Pochhammer symbols in the denominator:\begin{proposition}[Chu-Vandermonde identity~\cite{Koepf1998}]\label{prop:elementary_vandermonde}
        For every integer \(N \ge 0\),
        \begin{equation}
            \sum_{j=0}^N (-1)^j \binom{N}{j}\frac{(b)_j}{(a)_j}
            =
            \sum_{j=0}^N \frac{(-N)_{j}(b)_j}{(a)_j j!}
            = 
            \frac{(a-b)_N}{(a)_N}.
        \end{equation}
    \end{proposition}

    \noindent Similarly, the following statement simplifies certain sums containing $3$ Pochhammer symbols in the numerator and $2$ Pochhammer symbols in the denominator~\cite[Equation (35.8.6)]{NIST:DLMF}:
    \begin{proposition}[Pfaff-Saalsch\"utz formula]\label{prop:pfaff_saalschutz}
        For every \(N \in \mathbb{Z}_{\ge 0}\),
        \begin{equation}
            \sum_{j=0}^{N}
            \frac{(-N)_j\,(A)_j\,(B)_j}
            {(C)_j\,(1+A+B-C-N)_j\,j!}
            =
            \frac{(C-A)_N\,(C-B)_N}
            {(C)_N\,(C-A-B)_N},
        \end{equation}
        provided that the Pochhammer symbols in the denominator are nonzero.
    \end{proposition}

    \noindent Finally, we also use the following standard identity:
    \begin{proposition}\label{prop:inverse_sqrt_ub}
        Let \(N \in \mathbb{N}\). Then
        \begin{equation}
            \sum_{k=1}^N \frac{1}{\sqrt{k}} \le 2\sqrt{N}.
        \end{equation}
    \end{proposition}

\end{document}